\documentclass[letterpaper,11pt,oneside,reqno]{amsart}
\usepackage[utf8]{inputenc}%
\usepackage[english]{babel}%
\usepackage{amsmath,amssymb,amsthm,amsfonts}%
\usepackage{hyperref}%
\usepackage{enumerate}%
\usepackage{array}%
\usepackage{graphicx}
\usepackage[mathscr]{euscript}
\usepackage{color,tikz}
\usetikzlibrary{calc}
\usetikzlibrary{patterns}
\usetikzlibrary{arrows}
\usetikzlibrary{decorations.pathreplacing}
\usetikzlibrary{decorations.pathmorphing}
\usepackage[DIV15]{typearea}
\usepackage[width=.9\textwidth]{caption}
\allowdisplaybreaks%
\numberwithin{equation}{section}%

\newcommand{\Z}{\mathbb{Z}}
\newcommand{\C}{\mathbb{C}}
\newcommand{\R}{\mathbb{R}}

\renewcommand{\i}{\mathbf{i}}
\DeclareMathOperator{\sgn}{\mathrm{sgn}}

\newcommand{\al}{\alpha}
\newcommand{\si}{\sigma}
\newcommand{\la}{\lambda}

\newcommand{\be}{\beta}
\newcommand{\ga}{\boldsymbol\gamma}

\newcommand{\Weyl}[1]{\ensuremath{\mathbb{W}}^{#1}}
\newcommand{\tWeyl}[1]{\ensuremath{\widetilde{\mathbb{W}}}^{#1}}
\newcommand{\z}{{\vec z}}
\newcommand{\vxi}{{\vec \xi}}
\newcommand{\vom}{{\vec \varsigma}}
\newcommand{\om}{{\varsigma}}
\newcommand{\w}{{\vec w}}
\newcommand{\n}{{\vec n}}
\renewcommand{\c}{{\vec c}}
\newcommand{\y}{{\vec y}}

\newcommand{\x}{{\vec x}}

\newcommand{\Psibwd}{\Psi^\mathrm{bwd}}
\newcommand{\Psicfwd}{\Psi^\mathrm{cfwd}}
\newcommand{\Psifwd}{\Psi^\mathrm{fwd}}
\newcommand{\Psir}{\Psi^{r}}
\newcommand{\Psird}{\Psi^{r;\dilp}}
\newcommand{\Psil}{\Psi^{\ell}}
\newcommand{\Psild}{\Psi^{\ell;\dilp}}

\newcommand{\Phir}{\Phi^{r}}
\newcommand{\Phird}{\Phi^{r;\dilp}}
\newcommand{\Phil}{\Phi^{\ell}}
\newcommand{\Phild}{\Phi^{\ell;\dilp}}

\newcommand{\Psiasep}{\Psi^{\mathrm{ASEP}}}
\newcommand{\Phiasep}{\Phi^{\mathrm{ASEP}}}

\newcommand{\Refl}{\mathcal{R}} 
\newcommand{\Psiswap}{\mathcal{P}} 

\newcommand{\ev}{\mathsf{ev}_{\mu,\nu}}
\newcommand{\evasep}{\mathsf{ev}_{\mathrm{ASEP}}}
\newcommand{\Sm}{\mathcal{S}_{q,\nu}}
\newcommand{\Smasep}{\mathcal{S}_{\mathrm{ASEP}}}
\newcommand{\Smd}{\mathcal{S}_{q,\nu,\dilp}}
\newcommand{\Smv}{\mathcal{S}_{6V}}
\newcommand{\Vand}{\mathbf{V}}

\let\oldphi\phi \let\phi\varphi \let\varphi\oldphi

\newcommand{\Yb}[1]{\ensuremath{\mathbb{Y}}^{#1}} 
\newcommand{\Ybn}[1]{\ensuremath{\mathbb{Y}}_{#1}} 
\newcommand{\Xbn}[1]{\ensuremath{\mathbb{X}}_{#1}} 
\newcommand{\st}{\mathfrak{m}_{q,\nu}} 
\newcommand{\std}{\mathfrak{m}_{q,\nu,\dilp}} 

\newcommand{\Pld}{\mathcal{F}^{q,\nu}} 
\newcommand{\Pli}{\mathcal{J}^{q,\nu}} 
\newcommand{\Plspatial}{\mathcal{K}^{q,\nu}} 
\newcommand{\Plspectral}{\mathcal{M}^{q,\nu}} 

\newcommand{\Plde}{\mathcal{F}} 
\newcommand{\Plie}{\mathcal{J}} 
\newcommand{\Plspatiale}{\mathcal{K}} 
\newcommand{\Plspectrale}{\mathcal{M}} 

\newcommand{\Across}{\mathsf{A}} 

\newcommand{\Wc}{\mathcal{W}} 
\newcommand{\tWc}{\widetilde{\mathcal{W}}} 
\newcommand{\Cc}{\mathcal{C}} 

\newcommand{\Plm}{\mathsf{m}} 

\newcommand{\llangle}[0]{\ensuremath{\big\langle}} 
\newcommand{\rrangle}[0]{\ensuremath{\big\rangle}} 

\newcommand{\xinu}{\Xi} 
\newcommand{\xiasep}{\Xi_{\mathrm{ASEP}}} 

\newcommand{\LHS}[1]{\mathsf{LHS}\big[#1\big]} 

\newcommand{\Hbwd}{\ensuremath{\mathcal{H}^{\mathrm{bwd}}_{q,\mu,\nu}}} 
\newcommand{\Hfwd}{\ensuremath{\mathcal{H}^{\mathrm{fwd}}_{q,\mu,\nu}}} 
\newcommand{\Hcfwd}{\ensuremath{\mathcal{H}^{\mathrm{cfwd}}_{q,\mu,\nu}}} 
\newcommand{\Abwd}{\ensuremath{\mathcal{A}^{\mathrm{bwd}}_{q,\mu,\nu}}} 
\newcommand{\Lbwd}{\ensuremath{\mathcal{L}^{\mathrm{bwd}}_{q,\mu,\nu}}} 
\newcommand{\Lcfwd}{\ensuremath{\mathcal{L}^{\mathrm{cfwd}}_{q,\mu,\nu}}} 
\newcommand{\nbwd}{\nabla^{\mathrm{bwd}}_{\mu,\nu}} 
\newcommand{\nfwd}{\nabla^{\mathrm{fwd}}_{\mu,\nu}} 
\newcommand{\fbwd}{f^\mathrm{bwd}} 
\newcommand{\Bbwd}{\mathcal{B}^{\mathrm{bwd}}_{q,\nu}} 

\newcommand{\Hasep}{\ensuremath{\mathcal{H}_{\tau}^{\mathrm{ASEP}}}} 
\newcommand{\Lasep}{\ensuremath{\mathcal{L}_{\tau}^{\mathrm{ASEP}}}}
\newcommand{\nasep}{\nabla^{\mathrm{ASEP}}_{\tau}} 

\newcommand{\Hxxz}{\ensuremath{\mathcal{H}_{\svd}^{\mathrm{XXZ}}}} 
\newcommand{\tHxxz}{\ensuremath{\widetilde{\mathcal{H}}_{\svd}^{\mathrm{XXZ}}}} 

\newcommand{\Id}{\mathrm{Id}} 

\newcommand{\qhwt}{\mathsf{wt}} 

\newcommand{\Dil}[1]{\mathscr{G}_{#1}} 
\newcommand{\dilp}{\boldsymbol\theta} 

\newcommand{\pasep}{\mathsf{p}}
\newcommand{\qasep}{\mathsf{q}}

\newcommand{\svd}{\boldsymbol{\Delta}} 
\newcommand{\svpar}{\boldsymbol{\delta}} 

\newtheorem{proposition}{Proposition}[section]
\newtheorem{lemma}[proposition]{Lemma}
\newtheorem{corollary}[proposition]{Corollary}
\newtheorem{theorem}[proposition]{Theorem}
\theoremstyle{definition}
\newtheorem{definition}[proposition]{Definition}
\newtheorem{remark}[proposition]{Remark}

\begin{document}

\title[Spectral theory for interacting particle systems]{Spectral theory for interacting particle systems\\solvable by coordinate Bethe ansatz}
\author[A. Borodin]{Alexei Borodin}
\address{A. Borodin,
Massachusetts Institute of Technology,
Department of Mathematics,
77 Massachusetts Avenue, Cambridge, MA 02139-4307, USA,
and Institute for Information Transmission Problems, Bolshoy Karetny per. 19, Moscow 127994, Russia}
\email{borodin@math.mit.edu}

\author[I. Corwin]{Ivan Corwin}
\address{I. Corwin, Columbia University,
Department of Mathematics,
2990 Broadway,
New York, NY 10027, USA,
and Clay Mathematics Institute, 10 Memorial Blvd. Suite 902, Providence, RI 02903, USA,
and
Institut Henri Poincar\'e,
11 Rue Pierre et Marie Curie, 75005 Paris, France,
and Massachusetts Institute of Technology,
Department of Mathematics,
77 Massachusetts Avenue, Cambridge, MA 02139-4307, USA}
\email{ivan.corwin@gmail.com}

\author[L. Petrov]{Leonid Petrov}
\address{L. Petrov, University of Virginia, Department of Mathematics,
141 Cabell Drive, Kerchof Hall,
P.O. Box 400137,
Charlottesville, VA 22904-4137, USA
and
Northeastern University, Department of Mathematics, 360 Huntington ave.,
Boston, MA 02115, USA, and Institute for Information Transmission Problems, Bolshoy Karetny per. 19, Moscow, 127994, Russia}
\email{lenia.petrov@gmail.com}

\author[T. Sasamoto]{Tomohiro Sasamoto}
\address{T. Sasamoto,
Department of Physics, Tokyo Institute of Technology
2-12-1 Ookayama, Meguro-ku, Tokyo, 152-8550, JAPAN}

\email{sasamoto@phys.titech.ac.jp}

\date{}
\maketitle

\begin{abstract}
We develop spectral theory for the $q$-Hahn stochastic particle system introduced recently by Povolotsky. That is, we establish a Plancherel type isomorphism result which implies completeness and biorthogonality statements for the Bethe ansatz eigenfunctions of the system. Owing to a Markov duality with the $q$-Hahn TASEP (a discrete-time generalization of TASEP with particles' jump distribution being the orthogonality weight for the classical $q$-Hahn orthogonal polynomials), we write down moment formulas which characterize the fixed time distribution of the $q$-Hahn TASEP with general initial data.

The Bethe ansatz eigenfunctions of the $q$-Hahn system degenerate into eigenfunctions of other (not necessarily stochastic) interacting particle systems solvable by the coordinate Bethe ansatz. This includes the ASEP, the (asymmetric) six-vertex model, and the Heisenberg XXZ spin chain (all models are on the infinite lattice). In this way, each of the latter systems possesses a spectral theory, too. In particular, biorthogonality of the ASEP eigenfunctions which follows from the corresponding $q$-Hahn statement implies symmetrization identities of Tracy and Widom (for ASEP with either step or step Bernoulli initial configuration) as corollaries. Another degeneration takes the $q$-Hahn system to the $q$-Boson particle system (dual to $q$-TASEP) studied in detail in our previous paper (2013).

Thus, at the spectral theory level we unify two discrete-space regularizations of the Kardar--Parisi--Zhang equation / stochastic heat equation, namely, $q$-TASEP and ASEP.
\end{abstract}

\setcounter{tocdepth}{1}
\tableofcontents
\setcounter{tocdepth}{2}

\section{Introduction} 
\label{sec:introduction}

How to decompose functions onto distinguished bases --- harmonic analysis --- is a central theme touching many fields of mathematics. In this work we will answer this question for a class of bases which arise from the coordinate Bethe ansatz on $\Z$, and in so doing tie together (through various degenerations of our results) certain problems in probability, quantum integrable systems, symmetric functions and combinatorics. By working at a high algebraic level, our results are achieved through relatively soft methods (some of which do not survive various degenerations) such as contour deformations and residue computations.

The most basic instance of harmonic analysis --- Fourier analysis --- involves decomposing $L^2(\R)$ functions onto exponentials (or sine/cosines). Exponential functions are distinguished for many reasons. In particular, they are eigenfunctions for the Laplacian on $\R$ (as well as for the whole algebra of differential operators with constant coefficients that commute with the Laplacian). In a different direction, they are the characters of the (abelian) additive group of reals. Harmonic analysis generalizes in many directions, the simplest such involving higher dimensional $L^2(\R^k)$ functions. Still in the abelian setting, $\R^k$ may also be replaced by a locally compact abelian group.

Generalizing into the non-commutative setting of representations of Lie groups/algebras over real, complex, finite or p-adic fields, symmetric spaces, and further yet, to Hecke algebras of different sorts, the role of exponentials (or multivariate exponentials) is played by a hierarchy of symmetric polynomials and their limits that arise as characters and spherical functions. Like the exponentials, these polynomials are eigenfunctions for certain families of commuting operators. The operators, however, are generally non-local in variables indices, so unlike the (local) $k$-dimensional Laplacian, these operators mix all $k$-coordinates in non-trivial manners.

Locality is considered to be essential for physical models of the real world (with the notable exception of Coulomb and gravitational interactions), and it is highly desirable within models for equilibrium and non-equilibrium statistical and quantum mechanical systems. Thus, it is of paramount importance to discover such systems (described by Hamiltonians which often can be considered as versions of the Laplacian) whose eigenfunctions are explicit and form nice functional bases.

Early examples of such local, diagonalizable systems include the Heisenberg XXX and XXZ spin chains, as well as the transfer matrix for the six-vertex (or square ice) model. The coordinate Bethe ansatz, introduced in pioneering work of Bethe \cite{Bethe1931}, provides means to write down eigenfunctions for these Hamiltonians. It is worth noting that most presently known examples solvable by the coordinate Bethe ansatz are one-dimensional. 

Motivated by the study of finite volume statistical or quantum mechanical systems, research in this area dealt primarily with the problem of diagonalizing these operators acting on finite lattices with particular types of boundary conditions. These boundary conditions introduce restrictions on the eigenfunctions which go by the name of Bethe equations and typically are hard to analyze.

The motivation behind the present work is more probabilistic --- we seek to study fluctuation limit theorems 
for (stochastic) interacting particle systems and random growth processes on $\Z$ whose stochastic generators can be diagonalized via coordinate Bethe ansatz. While in physics it is typical to first 
consider finite systems and then take the limit of infinite system size,
working directly on $\Z$ (or~$\R$) is quite natural probabilistically. Moreover, certain Markov dualities enable us to reduce considerations involving infinite particle systems to those involving a finite (though arbitrarily large) number of particles.

Working on $\Z$ simplifies the coordinate Bethe ansatz as there are no boundary conditions, and hence no Bethe equations. On the other hand, there are now infinitely many eigenfunctions and the challenge becomes to figure out which ones participate in the diagonalization and with respect to which (spectral or Plancherel) measure. Results of this kind (generalized by the present work) regarding the XXX/XXZ model and quantum delta Bose gas go back to early work of Babbitt-Thomas \cite{BabbittThomas}, Babbitt-Gutkin \cite{BabbittGutkin},
Gutkin \cite{Gutkin}, Oxford \cite{Oxford1979}, Heckman-Opdam \cite{HeckmannOpdam1997}. 

In this paper we develop the above mentioned theory for what at the moment looks like the most general class of eigenfunctions diagonalizing vertex type models (and their degenerations) arising from the quantum affine algebra $U_q(\widehat{sl_2})$. In particular, these (and their degenerations) diagonalize Povolotsky's $q$-Hahn Boson system, the q-Boson system of \cite{SasamotoWadati1998}, \cite{BorodinCorwinPetrovSasamoto2013}, the six-vertex model, the Heisenberg XXZ and XXX spin chains (equivalently, the asymmetric simple exclusion process), Van Diejen's discrete delta Bose gas \cite{vanDiejen2004HL}, a semi-discrete delta Bose gas from \cite{BorodinCorwin2011Macdonald} and the continuous delta Bose gas. Using our results along with methods developed in related works (e.g. \cite{BorodinCorwin2011Macdonald}, \cite{BorodinCorwinFerrari2012}, \cite{BorodinCorwinFerrariVeto2013}) in the past few years, it is possible to access very precise asymptotic information about some of these systems (for various types of initial data) and probe universal limits as well as new phenomena.

The eigenfunctions we study are interesting in their own right. They form a one-parameter generalizations of the Hall-Littlewood symmetric polynomials. Interestingly, this is a different direction of generalization than that of the celebrated Macdonald symmetric polynomials. A particular degeneration ($\nu=0$) is, however, closely related to $t=0$ Macdonald symmetric polynomials \cite{BorodinCorwin2013discrete} (also known as $q$-Whittaker functions). It remains a mystery as to whether these two classes of symmetric polynomials (one arising in relation to non-local operators and the other in relation to local operators) can be united under a single generalization.

\subsection{Main results for the $q$-Hahn system eigenfunctions} 
\label{sub:main_results_for_the_qhahn}

The $q$-Hahn system
introduced by Povolotsky \cite{Povolotsky2013}
is a discrete-time
stochastic Markov dynamics on
$k$-particle configurations on $\Z$
(where $k\ge1$ is arbitrary) in which multiple particles at a
site are allowed (in fact, it is
a totally asymmetric zero-range process, see \cite{Liggett1985} for a general background).
At each step of the $q$-Hahn system dynamics,
independently at every occupied site
$i\in\Z$ with $y_i\ge1$ particles,
one randomly selects $s_i\in\{0,1,\ldots,y_i\}$
particles according to the probability distribution
\begin{align*}
	\phi_{q,\mu,\nu}(s_i\mid y_i)=\mu^{s_i}\frac{(\nu/\mu;q)_{s_i}(\mu;q)_{y_i-s_i}}{(\nu;q)_{y_i}}\frac{(q;q)_{y_i}}{(q;q)_{s_i}(q;q)_{y_i-s_i}},
	\qquad (a;q)_{n}:=\prod_{j=1}^{n}(1-aq^{j-1}),
\end{align*}
where $0<q<1$ and $0\le \nu\le\mu<1$
are three parameters of the model.
These selected $s_i$ particles
are immediately moved to the left (i.e., to site $i-1$).
This update occurs in parallel for each site.
See Fig.~\ref{fig:qBoson} in \S \ref{sub:hopping_distribution_and_q_hahn_orthogonality_weights},
left panel.

Configurations can be encoded by
vectors $\n=(n_1\ge \ldots\ge n_k)$, $n_i\in\Z$,
where $n_i$ is the position of the $i$-th particle from the right.
We denote by $\Weyl{k}$ the space of all such vectors.
Therefore, the backward Markov transition operator
$\Hbwd$ of the $q$-Hahn stochastic process
acts on the space of
compactly supported functions in the spatial variables $\n$.
We denote the latter space by $\Wc^{k}$.
The left and right eigenfunctions of the
operator $\Hbwd$ are, respectively\footnote{Note that the eigenfunctions do not belong to the space $\Wc^{k}$.}
\begin{align}
	\begin{array}{>{\displaystyle}rc>{\displaystyle}l}
	\Psil_{\z}(\n)&=&
	\sum_{\sigma\in S(k)}\prod_{1\le B<A\le k}
	\frac{z_{\sigma(A)}-qz_{\sigma(B)}}
	{z_{\sigma(A)}-z_{\sigma(B)}}\prod_{j=1}^{k}
	\left(\frac{1-z_{\sigma(j)}}{1-\nu z_{\sigma(j)}}\right)^{-n_j};
	\\
	\Psir_{\z}(\n)&=&
	(-1)^k(1-q)^{k}q^{\frac{k(k-1)}{2}}\st(\n)
	\sum_{\sigma\in S(k)}\prod_{1\le B<A\le k}
	\frac{z_{\sigma(A)}-q^{-1}z_{\sigma(B)}}
	{z_{\sigma(A)}-z_{\sigma(B)}}\prod_{j=1}^{k}
	\left(\frac{1-z_{\sigma(j)}}{1-\nu z_{\sigma(j)}}\right)^{n_j},
	\end{array}
	\label{eigen_intro}
\end{align}
where
$\z=(z_1,\ldots,z_k)\in(\C\setminus\{1,\nu^{-1}\})^{k}$,
and
$\st(\n)$ is an explicit quantity given in \eqref{stationary_measure}.
Here and below $S(k)$ is the symmetric group of all permutations of $\{1,2,\ldots,k\}$.
The corresponding eigenvalues are
$\ev(\z):=\prod_{j=1}^{k}\frac{1-\mu z_j}{1-\nu z_j}$.
Eigenfunctions \eqref{eigen_intro} were obtained in \cite{Povolotsky2013} by applying the coordinate
Bethe ansatz (a procedure dating back
to \cite{Bethe1931}) to the operator $\Hbwd$.
These eigenfunctions are also related to a deformation
of an affine Hecke algebra \cite{Takeyama2014}.
The latter object also leads to a stochastic interacting particle
system which is a continuous-time limit of the $q$-Hahn system.
Remarkably, the $q$-Hahn eigenfunctions depend only on the parameters
$(q,\nu)$, thus making $\mu$ an additional free parameter.
This implies that for fixed $q$ and
$\nu$ the operators $(\Hbwd)_{\mu\in[\nu,1)}$
form a commuting family.

Let $\Pld$ be the \emph{direct $q$-Hahn transform}
which takes a function $f\in\Wc^{k}$ in the spatial variables
$\n$ and produces a function in the spectral variables $\z$
according to
\begin{align*}
	(\Pld f)(\z)=\llangle f,\Psir_{\z}\rrangle_{\Wc^k}
	=\sum_{\n\in\Weyl{k}}f(\n)\Psir_{\z}(\n).
\end{align*}
The pairing above is the (obvious)
bilinear pairing in the space $\Wc^{k}$
\eqref{W_pairing}. The function
$\Pld f$ is a symmetric Laurent polynomial
in $(1-z_j)/(1-\nu z_j)$, $j=1,\ldots,k$.
We denote the space of
such Laurent polynomials by $\Cc^{k}_{z}$.

Let $\Pli$ be the \emph{inverse $q$-Hahn transform}
which maps Laurent polynomials $G\in\Cc^{k}_{z}$
to functions from $\Wc^{k}$
according to the following nested contour integration formula:
\begin{align*}
	(\Pli G)(
	\n)=\oint_{\ga_1}\frac{dz_1}{2\pi\i}
	\ldots
	\oint_{\ga_k}\frac{dz_k}{2\pi\i}
	\prod_{1\le A<B\le k}\frac{z_A-z_B}{z_A-qz_B}
	\prod_{j=1}^{k}
	\frac{1}{(1-z_j)(1-\nu z_j)}
	\left(\frac{1-z_j}{1-\nu z_j}\right)^{-n_j}
	G(\z).
\end{align*}
The contour $\ga_k$ is a small circle around $1$,
$\ga_A$ contains $q\ga_B$ contour for all $B>A$,
and all contours do not contain $\nu^{-1}$
(see also Definition \ref{def:contours}).
One can interpret $\Pli$
as a bilinear pairing
$\llangle G,\Psil(\n)\rrangle_{\Cc_{z}^k}$,
where $\Psil(\n)$ is viewed as a function in $\z$.
This pairing is defined in terms of integration
in which all variables belong to the same
contour (and not to various nested contours),
see \S \ref{sub:spectral_bilinear_pairing}.

The main results of the present paper
concerning the $q$-Hahn eigenfunctions \eqref{eigen_intro} are the following:
\begin{enumerate}[\bf1.]
	\item {\rm{}\bf{}(Plancherel formulas)}
	The transforms $\Pld$ and $\Pli$ are mutual inverses in the sense that
	$\Pli\Pld$ acts as the identity
	on $\Wc^{k}$,
	and $\Pld\Pli$
	is the identity map on
	$\Cc^{k}_{z}$.
	\item {\rm{}\bf{}(Plancherel isomorphism theorem)} The two function spaces
	$\Wc^{k}$ and $\Cc^{k}_{z}$
	are isomorphic as linear spaces
	with bilinear forms $(f,g)\mapsto\llangle f,\Psiswap^{-1}
	g\rrangle_{\Wc^{k}}$
	and $(F,G)\mapsto\llangle F,G\rrangle_{\Cc^{k}_{z}}$,
	where $\Psiswap$ is the operator in $\Wc^{k}$
	which swaps left and right eigenfunctions:
	$\big(\Psiswap^{-1} \Psil_{\z}\big)(\n)=\Psir_{\z}(\n)$.
	The map $\Psiswap$ also has a simple independent definition,
	see \eqref{Psiswap_operator} below.

	\item {\rm{}\bf{}(Completeness of the Bethe ansatz)}
	Any compactly supported function
	$f(\n)$ can be expressed through the eigenfunctions as
	\begin{align*}
		f(\n)&=\oint_{\ga_1}\frac{dz_1}{2\pi\i}
		\ldots
		\oint_{\ga_k}\frac{dz_k}{2\pi\i}
		\prod_{1\le A<B\le k}\frac{z_A-z_B}{z_A-qz_B}
		\prod_{j=1}^{k}
		\frac{1}{(1-z_j)(1-\nu z_j)}
		\left(\frac{1-z_j}{1-\nu z_j}\right)^{-n_j}
		\sum_{\y\in\Weyl{k}}f(\y)\Psir_{\z}(\y).
	\end{align*}
	\item {\rm{}\bf{}(Spatial biorthogonality)}
	The left and right eigenfunctions
	are biorthogonal viewed as elements of $\Cc^{k}_{z}$:
	$\llangle \Psil(\n),\Psir(\vec m)\rrangle_{\Cc_{z}^k}=\mathbf{1}_{\vec m=\n}$.\footnote{Here and below $\mathbf{1}_{\{\cdot\cdot\cdot\}}$
	denotes the indicator function.}

	\item {\rm{}\bf{}(Spectral biorthogonality)}
	Viewed as
	functions in the spatial variables,
	the left and right
	eigenfunctions are biorthogonal in the following
	formal way:
	\begin{align*}
		\sum_{\n\in\Weyl{k}}
		\Psir_{\z}(\n)\Psil_{\w}(\n)
		\Vand(\z)\Vand(\w)=
		(-1)^{\frac{k(k-1)}{2}}
		\prod_{j=1}^{k}(1-z_j)(1-\nu z_j)
		\prod_{A\ne B}(z_A-qz_B)
		\det[\delta({z_i-w_j})]_{i,j=1}^{k},
	\end{align*}
	where $\Vand(\z)=\prod_{1\le i<j\le k}(z_i-z_j)$ is the Vandermonde determinant.
	The above identity should be understood in a certain integrated sense.
	The simplest such interpretation is to multiply
	both sides by Laurent polynomials
	(not necessarily symmetric)
	in $(1-z_i)/(1-\nu z_i)$ and $(1-w_j)/(1-\nu w_j)$,
	respectively, and integrate all $z_i$ and $w_j$
	over a small circle around~$1$.
\end{enumerate}

We establish the spatial Plancherel formula
(that $\Pli\Pld$ is the identity map on the space of compactly supported
functions in the spatial variables $\n$)
in Theorem \ref{thm:spatial_Plancherel}.
Our argument relies on residue considerations
involving
shrinking (to $1$) or expanding (to $\nu^{-1}$) of
integration contours
in $\Pli$. The nested contour form
of the integration in
$\Pli$ is especially well adapted to this proof (there are other
ways to write down $\Pli$, see \S \ref{sub:various_forms_of_contour_integration}).

The spectral Plancherel formula
(that $\Pld\Pli$ is the identity map on $\Cc^{k}_{z}$, see Theorem \ref{thm:spectral_Plancherel})
is derived from the spectral biorthogonality (Theorems \ref{thm:spectral_biorthogonality}
and \ref{thm:spectral_biorthogonality_xi}). The latter
statement follows from the existence of one-parameter family
of commuting operators $(\Hbwd)_{\mu\in[\nu,1)}$ with different eigenvalues,
which produces many relations satisfied by the eigenfunctions \eqref{eigen_intro}.

The Plancherel isomorphism theorem is a direct consequence of the
two Plancherel formulas.
The need of the swapping operator $\Psiswap$
is evident since the natural basis of indicator functions
$\{\mathbf{1}_{\x}(\n)\}_{\x\in\Weyl{k}}$ in $\Wc^{k}$
is orthogonal, and the left and right eigenfunctions
are biorthogonal in $\Cc^{k}_{z}$.

The completeness of the Bethe ansatz for the $q$-Hahn system
(Corollary \ref{cor:completeness})
and the spatial biorthogonality of the eigenfunctions (Corollary \ref{cor:C_biorthogonality})
readily follow from the spatial Plancherel formula.

\smallskip

One immediate application of our main results
is the solution of the forward and backward Kolmogorov equations
for the $q$-Hahn stochastic process with general initial data.
By a duality result of \cite{Corwin2014qmunu},
the $q$-Hahn backward Kolmogorov equations
govern evolution of $q$-moments
of the $q$-Hahn Totally Asymmetric Simple Exclusion Process (TASEP).
The latter system (first introduced in \cite{Povolotsky2013}) is a discrete-time
three-parameter
generalization
of the usual continuous-time TASEP
with particles' jump distribution being the orthogonality weight for the classical $q$-Hahn orthogonal polynomials
(\S \ref{sub:hopping_distribution_and_q_hahn_orthogonality_weights}).
Consequently, this provides nested contour integral expressions
for $q$-moments
of the $q$-Hahn TASEP with general initial data
(see \S \ref{sub:nested_contour_integral_formulas_for_the_q_hahn_tasep}).

\smallskip

In principle, one could use moment formulas to address
asymptotic questions about the $q$-Hahn TASEP with various types of
initial data
and obtain results of
Kardar--Parisi--Zhang (KPZ) universality type.
See \cite{Veto2014qhahn}
for the treatment of the $q$-Hahn TASEP with step initial condition
(based on the moment and Fredholm determinantal
formulas of \cite{Corwin2014qmunu}),
and also
\cite{TW_ASEP2},
\cite{TW_total_current2009},
\cite{TW_ASEP4},
\cite{AmirCorwinQuastel2011},
\cite{CorwinQuastel2013},
\cite{BorodinCorwinFerrari2012},
\cite{FerrariVeto2013},
\cite{BorodinCorwinFerrariVeto2013},
\cite{MorenoQuastelRemenik2014OCYKPZ},
\cite{BCG6V}
for other systems (note that most of these models
are diagonalized by degenerations of the $q$-Hahn eigenfunctions,
cf. \S \ref{sub:degenerations}
below).
We do not pursue large time asymptotic problems in the present paper.


\begin{figure}[p]
	\begin{center}
		\scalebox{.86}{\begin{tikzpicture}[
		    scale=1.4, very thick,
		    axis/.style={thick, ->, >=stealth'},
		    block/.style ={rectangle, draw=red, align=center, rounded corners, minimum height=1em},
		    sblock/.style ={rectangle, draw=red, shade, shading=axis,
		    left color=lightgray, right color=white,
    		shading angle=45, align=center, rounded corners, minimum height=1em},
			miniblock/.style ={rectangle, draw=none, align=center},
			rotblock/.style ={rectangle, draw=none, align=center, rotate=90},
			rotblock1/.style ={rectangle, draw=none, align=center, rotate=90},
			rotblock2/.style ={rectangle, draw=none, align=center, rotate=90}]

			\def\cqh{(-.5,1)}
			\def\sv{(-.7,-3)}
			\def\xxz{(-1.3,-6)}
			\def\asep{(5.7,-6)}
			\def\qh{(6.5,1)}
			\def\qb{(7,-3)}
			\def\hl{(3,-9.8)}
			\def\ocy{(8,-9.7)}
			\def\kpz{(5,-13.5)}

			\draw \cqh node[sblock] (cqh) {\textbf{Conjugated $q$-Hahn, \S \ref{sec:spectral_theory_for_the_conjugated_hahn_boson_operator}}
			\\
			$\displaystyle\sum_{\sigma\in S(k)}\prod_{B<A}
			\frac{z_{\sigma(A)}-qz_{\sigma(B)}}
			{z_{\sigma(A)}-z_{\sigma(B)}}\prod_{j=1}^{k}
			\left(\frac{\dilp-z_{\sigma(j)}}{1-\nu z_{\sigma(j)}}\right)^{-n_j}$\\
			\hfill$n_i\in\Z$, $n_1\ge \ldots\ge n_k$\hspace{10pt}{\color{blue}($q,\nu,\dilp$)}};
			\draw \qh node[block] (qhahn) {\textbf{$q$-Hahn, \S\S \ref{sec:definition_of_eigenfunctions}--\ref{sec:the_q_mu_nu_boson_process_and_coordinate_bethe_ansatz}}
			\\
			$\displaystyle\sum_{\sigma\in S(k)}\prod_{B<A}
			\frac{z_{\sigma(A)}-qz_{\sigma(B)}}
			{z_{\sigma(A)}-z_{\sigma(B)}}\prod_{j=1}^{k}
			\left(\frac{1-z_{\sigma(j)}}{1-\nu z_{\sigma(j)}}\right)^{-n_j}$\\
			\hfill$n_i\in\Z$, $n_1\ge \ldots\ge n_k$
			\hspace{10pt}{\color{blue}($q,\nu$)}};
			\draw \sv node[sblock] (6vv) {\textbf{Asymmetric six-vertex model, \S \ref{sec:application_to_six_vertex_model}}
			\\
			$\displaystyle\sum_{\sigma\in S(k)}\prod_{B<A}
			\frac{z_{\sigma(B)}-qz_{\sigma(A)}}
			{z_{\sigma(B)}-z_{\sigma(A)}}\prod_{j=1}^{k}
			\left(\frac{\dilp-z_{\sigma(j)}}{1-z_{\sigma(j)}/(q\dilp)}\right)^{-x_j}$\\
			\hfill$x_i\in\Z$, $x_1< \ldots< x_k$\hspace{10pt}{\color{blue}($q,\dilp$)}};
			\draw \asep node[block] (asep) {\textbf{ASEP, stochastic six-vertex model,
			\S \ref{sec:application_to_asep}}
			\\
			$\displaystyle\sum_{\sigma\in S(k)}\prod_{B<A}
			\frac{z_{\sigma(B)}-\tau z_{\sigma(A)}}
			{z_{\sigma(B)}-z_{\sigma(A)}}\prod_{j=1}^{k}
			\left(\frac{1+z_{\sigma(j)}}{1+z_{\sigma(j)}/\tau}\right)^{-x_j}$\\
			\hfill$x_i\in\Z$, $x_1< \ldots< x_k$\hspace{10pt}{\color{blue}($\tau$)}};
			\draw \xxz node[sblock] (xxz) {\textbf{Heisenberg XXZ spin chain,
			\S \ref{sec:application_to_six_vertex_model}}
			\\
			$\displaystyle\sum_{\sigma\in S(k)}\prod_{B<A}
			\frac{z_{\sigma(B)}-\dilp^{-2} z_{\sigma(A)}}
			{z_{\sigma(B)}-z_{\sigma(A)}}\prod_{j=1}^{k}
			\left(\frac{1-\dilp^{-1} z_{\sigma(j)}}{\dilp^{-1}-z_{\sigma(j)}}\right)^{-x_j}$\\
			\hfill$x_i\in\Z$, $x_1< \ldots< x_k$\hspace{10pt}{\color{blue}($\dilp$)}
			};
			\draw \qb node[block] (qboson) {\textbf{$q$-Boson, \S \ref{sub:boson_eigenfunctions}}
			\\
			$\displaystyle\sum_{\sigma\in S(k)}\prod_{B<A}
			\frac{z_{\sigma(A)}-qz_{\sigma(B)}}
			{z_{\sigma(A)}-z_{\sigma(B)}}\prod_{j=1}^{k}
			\left(1-z_{\sigma(j)}\right)^{-n_j}$\\
			\hfill$n_i\in\Z$, $n_1\ge \ldots\ge n_k$
			\hspace{10pt}{\color{blue}($q$)}};
			\draw \hl node[sblock] (hl) {\textbf{Van Diejen's delta Bose gas}\\\hfill
			\textbf{\S\S \ref{ssub:_q_boson_to_van_diejen_s_delta_bose_gas}--\ref{ssub:van_diejen_s_delta_bose_gas_to_kpz}}
			\\
			$\displaystyle\sum_{\sigma\in S(k)}\prod_{B<A}
			\frac{z_{\sigma(A)}-qz_{\sigma(B)}}
			{z_{\sigma(A)}-z_{\sigma(B)}}\prod_{j=1}^{k}
			z_{\sigma(j)}^{-n_j}$\\
			\hfill$n_i\in\Z$, $n_1\ge \ldots\ge n_k$\hspace{10pt}{\color{blue}($q$)}\\
			\textbf{Hall-Littlewood polynomials}};
			\draw \ocy node[sblock] (semid) {\textbf{Semi-discrete delta Bose gas}\\\hfill\textbf{\S \ref{ssub:boson_to_semi_discrete_directed_polymer}}
			\\
			$\displaystyle\sum_{\sigma\in S(k)}\prod_{B<A}
			\frac{z_{\sigma(A)}-z_{\sigma(B)}-c}
			{z_{\sigma(A)}-z_{\sigma(B)}}\prod_{j=1}^{k}
			z_{\sigma(j)}^{-n_j}$\\
			\hfill$n_i\in\Z$, $n_1\ge \ldots\ge n_k$
			\hspace{10pt}{\color{blue}($c$)}};
			\draw \kpz node[sblock] (kpz) {\textbf{Continuous delta Bose gas}, \textbf{\S \ref{ssub:_semi_discrete_directed_polymer_to_kpz}}
			\\
			$\displaystyle\sum_{\sigma\in S(k)}\prod_{B<A}
			\frac{z_{\sigma(A)}-z_{\sigma(B)}-\tilde c}
			{z_{\sigma(A)}-z_{\sigma(B)}}\prod_{j=1}^{k}
			e^{x_jz_{\sigma(j)}}$\\
			\hfill$x_i\in\R$, $x_1\le \ldots\le x_k$
			\hspace{10pt}{\color{blue}($\tilde c$)}};

			\draw[->] (cqh.south) -- (6vv.north) node [midway, miniblock,xshift=-35] {$\nu=1/(q\dilp)$
			\\(or $\nu=q/\dilp$)};

			\draw[->] (cqh.east) -- (qhahn.west) node [midway, yshift=10] {$\dilp=1$};
			\draw[->] (6vv.south) -- ([xshift=-65]asep.north)
			node [midway, yshift=-7, miniblock, xshift=-30] {$\dilp=1$, $q=\tau$};

			\draw[->] ([xshift=-40]qhahn.south)
			to [in=90,out=-135] ([xshift=-20]qboson.west) to [in=120,out=-90]
			([xshift=-40]asep.north)  node [miniblock, yshift=160, xshift=-56]
			{$\nu=1/q=1/\tau$\\(or $\nu=q=1/\tau$)};
			\draw[->] (qhahn.south) -- (qboson.north)
			node [midway,xshift=-22] {$\nu=0$};
			\draw[->] (6vv.south) -- (xxz.north) node [midway, miniblock, xshift=-42]
			{$q=1/\dilp^{2}$\\
			(or $q=\dilp^{2}$)};

			\draw[<->] ([yshift=20]xxz.east) -- ([yshift=20]asep.west)
			node [midway, rotblock, xshift=-27] {$\dilp=1/\sqrt\tau$,\\
			conjugate};
			
			\draw[->,  dashed]
			([xshift=60]qboson.south) to
			[in=70,out=-90] ([yshift=-20,xshift=4]asep.east)
			to [in=90,out=-110] (hl.north)
			node [xshift=5,yshift=45] {$|z_j|\gg1$};
			\draw[->,  dashed]
			([xshift=60]qboson.south) -- ([xshift=40]semid.north)
			node [miniblock, xshift=-45, yshift=17]
			{$q=e^{-\epsilon}\to1$\\
			$1-z_j=O(\epsilon)$};
			\draw[->,  dashed] (hl.south) -- ([xshift=-40]kpz.north)
			node [midway, miniblock, xshift=-48]
			{$q=e^{-\epsilon}\to1$\\
			$1-z_j=O(\epsilon)$\\
			$n_j=O(\epsilon^{-1})$};
			\draw[->, dashed]
		    (semid.south) -- ([xshift=50]kpz.north)
			node[midway, miniblock, xshift=-58, yshift=-8]
			{
			$c=O(\epsilon)$\\
			$z_j=const+O(\epsilon)$\\
			$n_j=O(\epsilon^{-1})$};
			\draw[->, dashed] ([yshift=-18]asep.west)
			to [in=90,out=-135]
			([xshift=-40,yshift=-20]hl.west)
			to [in=180,out=-90]
			(kpz.west)
			node at ([xshift=-33,yshift=55]hl.west) [miniblock]
			{$\tau=e^{-\sqrt \epsilon}$\\
			$z_j=const+O(\sqrt \epsilon)$\\
			$x_j=O(\epsilon^{-1})$};
			\draw[dashed]
			([xshift=-18.1,yshift=-76.1]hl.west)
			to [in=-90,out=130]
			([xshift=-30]xxz.south);
			\node at ([xshift=-33,yshift=-79.5]hl.west) {\S \ref{ssub:asep_and_xxz_to_kpz_}};
		\end{tikzpicture}}
		\end{center}
  	\caption{A hierarchy of eigenfunctions of Hall-Littlewood type
  	(about the name see
  	Remark \ref{rmk:generic_linear_fractional_not_more_general}) possessing a Plancherel theory. All functions arise via coordinate Bethe ansatz for various integrable particle systems.
  	Only the left eigenfunctions are written down.
  	The use of the spatial variables $n_j$ vs. $x_j$
  	reflects literature conventions which are
  	not uniform throughout.
  	Shading of boxes indicates particle systems which are not necessarily stochastic
  	(however, they may be dual to stochastic processes such as the semi-discrete
  	stochastic heat equation or the continuous stochastic heat \mbox{equation / KPZ} equation).
  	Solid arrows mean straightforward degenerations of eigenfunctions, and dashed arrows correspond to scaling limits which are briefly discussed in \S\ref{sub:further_degenerations_of_eigenfunctions}.}
  	\label{fig:big_scheme}
\end{figure}
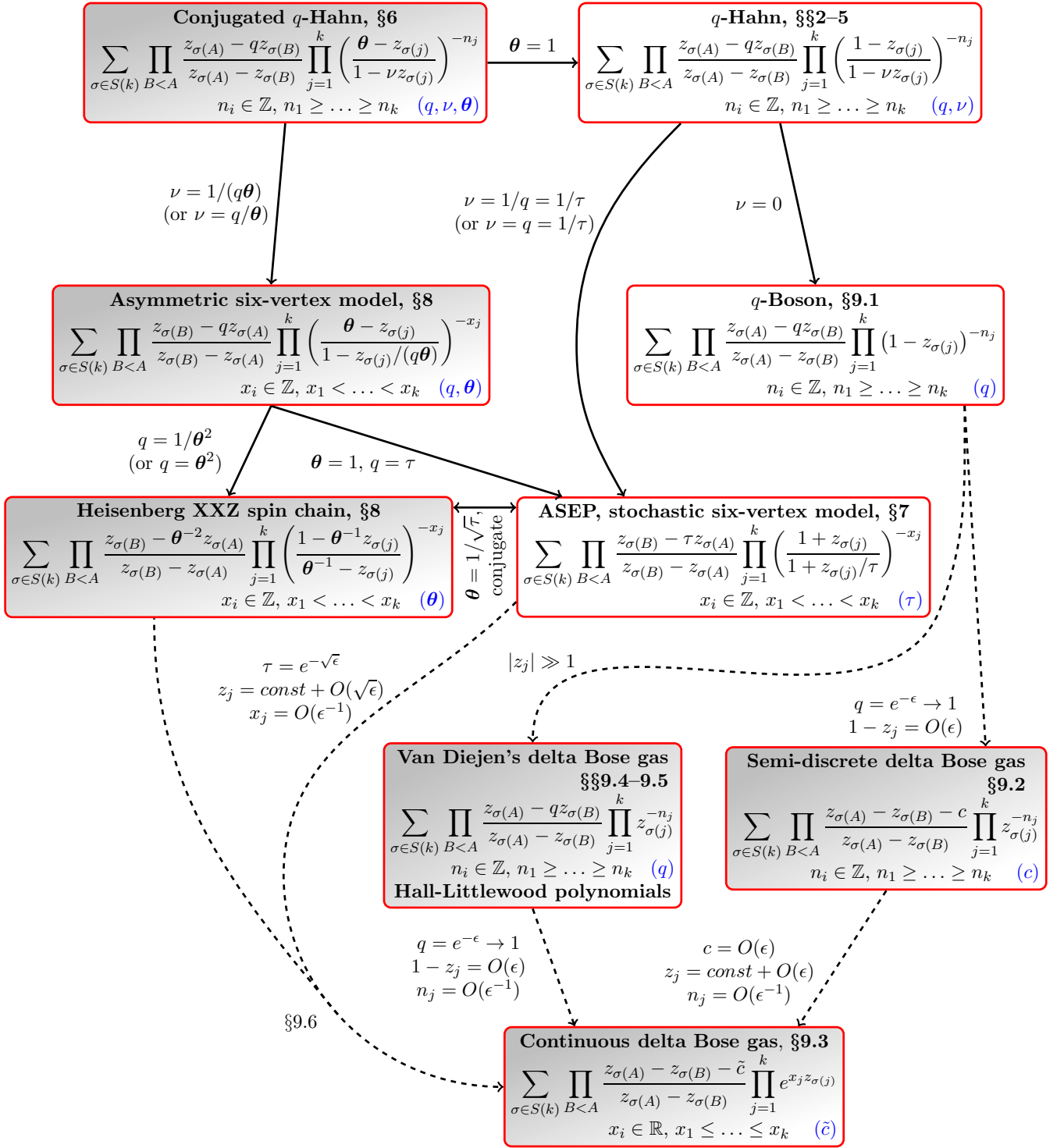

\subsection{Degenerations and limits of $q$-Hahn eigenfunctions} 
\label{sub:degenerations}

After establishing the main results concerning the eigenfunctions \eqref{eigen_intro},
we turn our attention to studying their degenerations, and discuss
various systems which are diagonalized by these.
We explore two ways to degenerate these eigenfunctions (corresponding to two arrows starting
from the ``$q$-Hahn'' block on Fig.~\ref{fig:big_scheme}):

\smallskip

The first way is to set $\nu=0$.
Then the eigenfunctions \eqref{eigen_intro}
turn into the Bethe ansatz eigenfunctions
of the (stochastic) $q$-Boson particle system
introduced in \cite{SasamotoWadati1998}.
This system is dual to the $q$-TASEP of \cite{BorodinCorwin2011Macdonald}
(see \cite{BorodinCorwinSasamoto2012}).
When $\nu=0$, the $q$-Hahn stochastic process itself becomes the system
dual to the
discrete-time geometric $q$-TASEP of \cite{BorodinCorwin2013discrete}.
A continuous-time limit of the latter is the $q$-TASEP.

The spectral theory for the $q$-Boson eigenfunctions
is the main subject of our previous paper \cite{BorodinCorwinPetrovSasamoto2013}.
Though results of the present paper imply the main results of
\cite{BorodinCorwinPetrovSasamoto2013}, the ideas of the proofs
in the two papers differ significantly
(see Remarks \ref{rmk:difference_in_proofs_of_Plancherel} and \ref{rmk:difference_with_BCPS_spectral}
for details).

The $q$-Boson system admits a scaling limit to
a semi-discrete delta Bose gas considered previously by Van Diejen \cite{vanDiejen2004HL},
as well as to another semi-discrete delta Bose gas that describes the evolution of moments of the semi-discrete stochastic heat equation (or equivalently, the O'Connell-Yor semi-discrete directed polymer partition function
\cite{OConnellYor2001}, \cite{Oconnell2009_Toda}).
See \cite{BorodinCorwinPetrovSasamoto2013} for details.
We briefly describe the corresponding limits of the $q$-Boson eigenfunctions
in \S \ref{ssub:boson_to_semi_discrete_directed_polymer} and \S \ref{ssub:_q_boson_to_van_diejen_s_delta_bose_gas}.

\smallskip

The second way is to set $\nu=1/q$.
As we explain in \S \ref{sec:application_to_asep},
under this degeneration the eigenfunctions \eqref{eigen_intro}
turn into the eigenfunctions of the ASEP,
another well-studied stochastic interacting particle system
\cite{Spitzer1970}
in which particles living on the lattice $\Z$
(at most one particle at a site)
randomly jump to the left by one at rate $\qasep$
or to the right by one at rate $\pasep$.
Then the remaining parameter $q$ in \eqref{eigen_intro}
would mean the ratio
$\pasep/\qasep$ (which is sometimes denoted by $\tau$).\footnote{However,
the ASEP itself does not seem to be
a degeneration of either
the $q$-Hahn system or the $q$-Hahn TASEP
(note that the ASEP is self-dual in various senses, cf.
\cite{Liggett1985}, \cite{Schutz1997dualityASEP}, \cite{BorodinCorwinSasamoto2012},
\cite{SasamotoASEPqj2014}).}
The weak ordering of the spatial variables
$\n=(n_1\ge \ldots\ge n_k)$
should be replaced by the strict one because
the constant $\st(\n)\vert_{\nu=1/q}$
in the right eigenfunction
vanishes unless $n_1> \ldots>n_k$. These strictly ordered spatial
variables encode locations of the ASEP particles.

All our main results for the $q$-Hahn eigenfunctions
have analogues for the ASEP eigenfunctions.
However, this degeneration is not always straightforward:
For instance, all nested contours in the inverse transform $\Pli$
and in the spatial Plancherel formula
should be replaced by a single small contour around a singularity.

Some parts of the ASEP spectral theory have been already established.
Namely, the spatial Plancherel formula for the ASEP
is equivalent to Tracy-Widom's solution of the
ASEP master equation \cite{TW_ASEP1} (see also \cite{Schutz1997exact} for two-particle case).
Moreover, the spectral Plancherel formula for the ASEP
implies as corollaries the symmetrization
identities
first obtained by Tracy and Widom
as
\cite[(1.6)]{TW_ASEP1}
(for ASEP with the step initial condition)
and
\cite[(9)]{TW_ASEP4}
(for the step Bernoulli initial condition).\footnote{In fact,
the $q$-Hahn spectral Plancherel formula
also provides certain symmetrization identities
generalizing the Tracy-Widom's ones.
These identities at the $q$-Hahn
level turn out to be significantly more complicated,
and
it remains unclear whether it is possible to use them in any asymptotic analysis.
See, however, the use of similar identities in \cite{K_Lee_qBoson2013}
at the $q$-Boson level
for
finding the
distribution of the leftmost particle's position in the $q$-Boson process.}
These symmetrization identities
served as a
crucial step towards
Fredholm determinantal formulas for the ASEP \cite{tracy2008fredholm}
and ultimately to proving its KPZ universality
\cite{TW_ASEP2}, \cite{TW_total_current2009} (see also
\cite{BCG6V} for a recent application of these symmetrization identities to the
six-vertex model).

In \S \ref{sec:application_to_six_vertex_model}
we consider
two other (non-stochastic) models
with strictly ordered spatial variables,
namely, the (asymmetric) six-vertex model
and the Heisenberg XXZ spin chain \cite{baxter2007exactly}
(both on the infinite lattice $\Z$).
Eigenfunctions of these two models are similar
to the ASEP ones, and, moreover,
arise as degenerations
of eigenfunctions of the $q$-Hahn transition operator
conjugated by the multiplication operator $(const)^{n_1+\ldots+n_k}$
(with a suitable choice of the constant).
This conjugated $q$-Hahn operator is no longer stochastic.
As we explain in \S \ref{sec:spectral_theory_for_the_conjugated_hahn_boson_operator},
the spectral theory in the conjugated case is essentially equivalent
to the one for the stochastic $q$-Hahn transition operator.
In this way we arrive at spectral theory results
for the six-vertex and the XXZ models.
We also comment on
Plancherel formulas for the XXZ model which already
appeared in
the works of Thomas, Babbitt, and Gutkin
\cite{BabbittThomas}, \cite{BabbittGutkin},
\cite{Gutkin}.

\smallskip

Both ways to degenerate the eigenfunctions \eqref{eigen_intro} described
above are unified again at the level of the
continuous delta Bose gas.
The latter system is dual
\cite{BertiniCancrini1995}, \cite[\S6]{BorodinCorwin2011Macdonald}
to the stochastic heat equation, or, via the Hopf-Cole transform, to the KPZ
equation. That is, eigenfunctions of both the
semi-discrete delta Bose gases mentioned above, as well as the ASEP eigenfunctions,
admit scaling limits to the
continuous delta Bose gas eigenfunctions (see \S \ref{ssub:_semi_discrete_directed_polymer_to_kpz},
\S \ref{ssub:van_diejen_s_delta_bose_gas_to_kpz}, and \S \ref{ssub:asep_and_xxz_to_kpz_}).
In this way, our results for the $q$-Hahn eigenfunctions
provide a unification
(at the spectral theory level)
of two discrete-space regularizations of the Kardar--Parisi--Zhang equation / stochastic heat equation, namely, ASEP and $q$-TASEP.


\subsection{Further directions and connections} 
\label{sub:further_directions_and_connections}

Let us mention (very briefly) other connections, as well as
possible directions of further study.

\smallskip

The stochastic $q$-Boson process (which is a continuous-time degeneration
of the $q$-Hahn system) is dual to the $q$-TASEP.
The latter dynamics was invented in \cite{BorodinCorwin2011Macdonald}
is a one-dimensional marginal of a
certain dynamics on interlacing integer arrays related to
Macdonald processes.
The results about Macdonald processes
developed in \cite{BorodinCorwin2011Macdonald}, \cite{BCGS2013}
provide another way of establishing moment
formulas for the $q$-TASEP
besides the duality approach used in \cite{BorodinCorwinSasamoto2012}.
Since then, other dynamics on interlacing arrays
with $q$-TASEP (and other integrable interacting particle systems)
as one-dimensional marginals were introduced, cf.
\cite{OConnellPei2012}, \cite{BorodinPetrov2013NN},
\cite{MatveevPetrov2014}.
See also \cite{BorodinCorwin2013discrete}
for more discussion on intrinsic connections between
$q$-TASEPs and Macdonald processes (and also Macdonald difference operators).

An intriguing problem is to find
suitable dynamics on interlacing arrays (or maybe another dynamical model
on two-dimensional particle configurations)
with the $q$-Hahn TASEP as a one-dimensional marginal.
This could potentially provide a ``$q$-Hahn extension''
of Macdonald processes.

\smallskip

The spectral theory at the $q$-Hahn level
also generalizes the setup of the six-vertex model (and the related ASEP and XXZ models).
The coordinate Bethe ansatz approach to the six-vertex model
initially performed by Lieb \cite{Lieb67}, \cite{baxter2007exactly}
was greatly generalized to
what is now known as the algebraic Bethe ansatz, cf.
\cite{Faddeev_Lectures},
\cite{reshetikhin2010lectures}.
It would be of interest to construct
an algebraic lifting of the Bethe ansatz
for the $q$-Hahn system.
See also \cite{SasamotoWadati1998} on the algebraic
nature of the $q$-Boson Hamiltonian.


\subsection{Outline}

The outline and scope of the paper is represented graphically
on Fig.~\ref{fig:big_scheme}.
In \S \ref{sec:definition_of_eigenfunctions}
we describe the
eigenfunctions of the $q$-Hahn stochastic particle system, and
introduce other necessary notation.
In
\S \ref{sec:main_results}
we present detailed statements of our main results, and
prove the spatial Plancherel formula.
In
\S \ref{sec:spectral_biorthogonality_of_eigenfunctions}
we prove the spectral biorthogonality of the $q$-Hahn eigenfunctions
which implies the spectral Plancherel formula.
In
\S \ref{sec:spectral_theory_for_the_conjugated_hahn_boson_operator}
we briefly describe modifications of our main results
corresponding to a conjugated $q$-Hahn operator
(which is no longer stochastic).
In
\S \ref{sec:application_to_asep}
we describe a spectral theory
for the ASEP eigenfunctions, and
match it to the work of Tracy and Widom on exact formulas for the ASEP
\cite{TW_ASEP1}, \cite{TW_ASEP4}.
In
\S \ref{sec:application_to_six_vertex_model}
(using formulas from \S \ref{sec:spectral_theory_for_the_conjugated_hahn_boson_operator}),
we explain how the eigenfunctions of the conjugated
$q$-Hahn operator degenerate to eigenfunctions
of the (asymmetric) six-vertex model
and of the Heisenberg XXZ spin chain, and produce spectral theory results
for the latter two systems.
In \S\ref{sub:further_degenerations_of_eigenfunctions}
we briefly describe further degenerations of the $q$-Hahn eigenfunctions.

\subsection{Acknowledgments}

This project was initiated at the Simons Symposium on the Kardar--Parisi--Zhang Equation.
The authors appreciate illuminating and helpful comments and discussions
with Donald Babbitt, Vadim Gorin,
Alexander Povolotsky,
Larry Thomas,
Craig Tracy, and Harold Widom.
We are grateful to anonymous reviewers for suggestions
on improving the presentation of our results.
We also thank Yier Lin for pointing out issues in the ASEP reduction of our results
(see Remark~\ref{rmk:update_ASEP_statement} for further discussion).

A.B. was partially supported by the NSF grant DMS-1056390.
I.C. was partially supported by the NSF through DMS-1208998 as well as by Microsoft Research and MIT through the Schramm Memorial Fellowship, by the Clay Mathematics Institute through the Clay Research Fellowship, and by the Institute Henri Poincare through the Poincare Chair.
T.S. was partially supported by KAKENHI 22740054 and Sumitomo Foundation.


\section{Definition of eigenfunctions} 
\label{sec:definition_of_eigenfunctions}

The results in this and the next sections
depend on two parameters $(q,\nu)$, with $|q|<1$ and $0\le \nu<1$,
and on a fixed
integer $k\ge1$.

\subsection{Spatial and spectral variables} 
\label{sub:spatial_and_spectral_variables}

We will deal with $k$-particle configurations on
the lattice $\Z$. Ordered positions of particles are encoded by a vector
$\n\in\Weyl{k}$, where $\Weyl{k}$ is
the Weyl chamber:
\begin{align}\label{Weyl_k}
	\Weyl{k} := \big\{\n = (n_1,\ldots, n_k)\colon n_1\geq \cdots \geq n_k, \ n_i\in\Z\big\}.	
\end{align}
We will refer to $\n=(n_1,\ldots,n_k)$ as to the \emph{spatial variables}.

By $\Wc^{k}$ denote the space of all compactly supported functions $f\colon \Weyl{k}\to\C$.
The space $\Wc^{k}$ carries a natural symmetric bilinear pairing
\begin{align}\label{W_pairing}
	\llangle f,g\rrangle_{\Wc^k}:=\sum_{\n\in\Weyl{k}}f(\n)g(\n).
\end{align}

Define the map $\xinu\colon \C\to\C$ by
\begin{align}\label{xinu}
	\xinu(z):=\frac{1-z}{1-\nu z}.
\end{align}
Note that this is an involution (i.e., $\xinu\big(\xinu(z)\big)=z$)
which swaps pairs of points $0\leftrightarrow1$ and $\nu^{-1}\leftrightarrow\infty$
in $\overline\C:=\C\cup\{\infty\}$. Note that when $\nu=0$ (so $\nu^{-1}=\infty$),
the map $\xinu$ reduces to $z\mapsto 1-z$, which preserves~$\infty$.

For any two points $a\ne b\in\overline\C$, by $\Cc^{k}(a,b)$ denote the space of
symmetric functions $G(z_1,\ldots,z_k)$ on $\C^{k}$
which are Laurent polynomials in
\begin{align*}
	\frac{1-a^{-1}z_1}{1-b^{-1}z_1},\quad\ldots,\quad\frac{1-a^{-1}z_k}{1-b^{-1}z_k}.
\end{align*}
By agreement, if $a=\infty$, then the factor $1-a^{-1}z_j$ is not present;
if $a=0$,
then the same factor should be replaced simply by $(-z_j)$.

We will need two particular cases:
\begin{itemize}
	\item
	the space $\Cc^{k}_{z}:=\Cc^{k}(1,\nu^{-1})$ of symmetric functions
	in variables denoted by $z_1,\ldots,z_k$
	such that these functions are Laurent polynomials
	in the expressions $\dfrac{1-z_1}{1-\nu z_1},\ldots,\dfrac{1-z_k}{1-\nu z_k}$;
	\item
	the space $\Cc^{k}_{\xi}:=\Cc^{k}(0,\infty)$
	of symmetric Laurent polynomials in variables which we will denote by $\xi_1,\ldots,\xi_k$.
\end{itemize}
Note that the involution $\xinu$ interchanges the spaces $\Cc^{k}_z$ and $\Cc^{k}_\xi$ by swapping the variables
$z_j\leftrightarrow \xi_j$.

We will refer to either the variables $\z=(z_1,\ldots,z_k)$ or
$\vxi=(\xi_1,\ldots,\xi_k)$ (related to the corresponding spaces of symmetric Laurent polynomials)
as to the \emph{spectral variables}.


\subsection{Bilinear pairing in spectral variables} 
\label{sub:spectral_bilinear_pairing}

Here we introduce symmetric bilinear pairings on spaces of functions in spectral
variables (both $\z$ and $\vxi$). First, we need to define integration contours and a (complex-valued) Plancherel spectral measure.

\begin{definition}[Contours]\label{def:contours}
	Let $\ga_1,\ldots,\ga_k$ be positively oriented, closed contours chosen so that they all contain $1$
	and do not contain $\nu^{-1}$,
	so that the $\ga_A$ contour contains the image of $q$ times the $\ga_B$
	contour for all $B>A$, and so that $\ga_k$ is a small enough circle around 1
	that does not contain $q$.
	Let $\ga$ be a positively oriented closed contour which
	contains $1$, does not contain $\nu^{-1}$, and also contains its own image under the multiplication by $q$.
	See Fig.~\ref{fig:contours}.

	Clearly, it is possible to
	choose such contours for all $q\in\C$ with $|q|<1$.
	One way to see this possibility for complex $q$ is to consider the spiral
	$\{q^{\varkappa}\}_{\varkappa\in\R_{\ge0}}$, and modify the
	contours on Fig.~\ref{fig:contours} appropriately.
\end{definition}

\begin{figure}[htbp]
	\begin{center}
	\begin{tikzpicture}
		[scale=3]
		\def\pt{0.02}
		\def\q{.6}
		\draw[->, thick] (-.8,0) -- (2,0);
	  	\draw[->, thick] (0,-1.2) -- (0,1.2);
	  	\draw[fill] (1,0) circle (\pt) node [below] {1};
	  	\draw[fill] (\q,0) circle (\pt) node [below,yshift=-1] {$q$};
	  	\draw[fill] (\q*\q,0) circle (\pt) node [below,yshift=3,xshift=2] {$q^{2}$};
	  	\draw[fill] (0,0) circle (\pt) node [below left] {$0$};
	  	\draw[fill] (1.7,0) circle (\pt) node [below, yshift=3,xshift=3] {$\nu^{-1}$};
	  	\draw (1,0) circle (.2) node [below,xshift=9,yshift=16] {$\ga_3$};
	  	\draw[dotted] (\q,0) circle (.2*\q);
	  	\draw[dotted] (\q*\q,0) circle (.2*\q*\q);
	  	\draw (1/2,0) circle (1) node [below,xshift=80,yshift=50] {$\ga$};
	  	\draw[dotted] (1/2*\q,0) circle (\q);
	  	\draw (1/2+\q/2+.06,0) ellipse (.4 and .3) node [below,xshift=-10,yshift=27] {$\ga_{2}$};
	  	\draw (\q+.17,0) ellipse (.55 and .45) node [below,xshift=-20,yshift=35] {$\ga_{1}$};
	\end{tikzpicture}
	\end{center}
  	\caption{A possible choice of integration contours $\ga_1,\ga_2,\ga_3$, and $\ga$ for $k=3$ and $0<q<1$.
  	Contours $q\ga_1$, $q^{2}\ga_1$, and $q\ga$ are shown dotted.}
  	\label{fig:contours}
\end{figure}
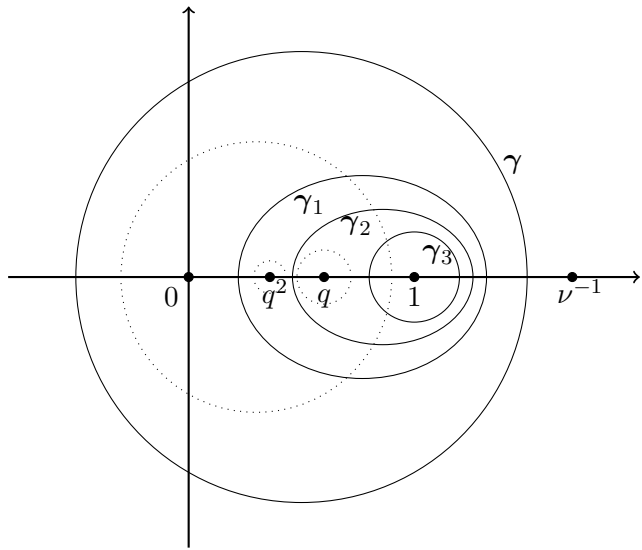

\begin{definition}[Plancherel measure]\label{def:Plancherel_measure}
	For each partition\footnote{That is, $\la=(\la_1\ge\la_2\ge \ldots\ge0)$, $\la_i\in\Z$,
	with $\la_1+\la_2+\ldots=k$.
	The number of nonzero components in $\la$ will be denoted by $\ell(\la)$.}  $\la\vdash k$ define
	\begin{align}\label{Plancherel_measure}
		d\Plm^{(q)}_\la(\w):=
		\frac{(1-q)^{k}(-1)^{k}q^{-\frac{k^2}2}}{m_1!m_2!\ldots}
		\det\left[\frac{1}{w_iq^{\la_i}-w_j}\right]_{i,j=1}^{\ell(\la)}
		\prod_{j=1}^{\ell(\la)}w_j^{\la_j}q^{\frac{\la_j^{2}}2}\frac{dw_j}{2\pi\i},
	\end{align}
	where $\w=(w_1,\ldots,w_{\ell(\la)})\in\C^{\ell(\la)}$, and $m_j$ is the number of components of $\la$
	equal to $j$ (so that $\la=1^{m_1}2^{m_2}\ldots$).

	The object \eqref{Plancherel_measure} may be viewed as a (complex-valued)
	Plancherel measure
	(see Remark \ref{rmk:name_for_Plancherel} below).
\end{definition}

If $\la\vdash k$, we will use the notation
\begin{align}\label{w_circ_la}
	\w\circ\la:=(w_1, qw_1,\ldots, q^{\la_1-1}w_1,w_2, qw_2,\ldots, q^{\la_{2}-1}w_2,\ldots,
	w_{\la_{\ell(\la)}}, qw_{\la_{\ell(\la)}},\ldots,
	q^{\la_{\ell(\la)}-1}w_{\la_{\ell(\la)}})\in\C^{k}.
\end{align}

The bilinear pairing on the space $\Cc^{k}_{z}$ is defined as
\begin{align}\label{Cz_pairing}
	\llangle F,G\rrangle_{\Cc_{z}^k}:=
	\sum_{\la\vdash k}\oint_{\ga_k}\ldots\oint_{\ga_k}
	d\Plm^{(q)}_\la(\w)
	\prod_{j=1}^{\ell(\la)}\frac{1}{(w_j;q)_{\la_j}(\nu w_j;q)_{\la_j}}
	F(\w\circ\la)
	G(\w\circ\la),\qquad F,G\in\Cc^{k}_{z}.
\end{align}
Here $(a;q)_{n}:=(1-a)(1-aq)\ldots(1-aq^{n-1})$ denotes the $q$-Pochhammer symbol as usual.

\begin{remark}\label{rmk:large_contour_Plancherel_measure}
	The same pairing can be rewritten in a simpler form
	in terms of integration over the large contour
	$\ga$ (this follows from Proposition \ref{prop:nesting_unnesting} below):
	\begin{align}\label{Cz_pairing_large}
		\llangle F,G\rrangle_{\Cc_{z}^k}=
		\oint_{\ga}\ldots\oint_{\ga}
		d\Plm^{(q)}_{(1^{k})}(\z)
		\prod_{j=1}^{k}\frac{1}{(1-z_j)(1-\nu z_j)}
		F(\z)
		G(\z),\qquad F,G\in\Cc^{k}_{z}.
	\end{align}

	The Plancherel measure $d\Plm^{(q)}_{(1^{k})}(\z)$ in \eqref{Cz_pairing_large}
	is simplified (with the help of the Cauchy determinant identity) to
	\begin{align}\label{dmu_large}
		d\Plm^{(q)}_{(1^{k})}(\z)=
		\frac{1}{k!}
		\frac{(-1)^{\frac{k(k-1)}{2}}\Vand(\z)^2}{\prod_{i\ne j}(z_i-qz_j)}
		\prod_{j=1}^{k}\frac{dz_j}{2\pi\i},
	\end{align} 	
	where $\Vand(\z)=\prod_{1\le i<j\le k}(z_i-z_j)$ is the Vandermonde determinant.
\end{remark}

There is a corresponding bilinear pairing on the space $\Cc_{\xi}^{k}$ as well. One can define
it by making a change of variables $\xinu$ in \eqref{Cz_pairing_large}:
\begin{align}\label{Cxi_pairing_large}
	\llangle F,G\rrangle_{\Cc_\xi^k}:=
	\frac{1}{k!(2\pi\i)^{k}(\nu-1)^{k}}
	\left(\frac{1-\nu}{1-q \nu}\right)^{k(k-1)}
	\oint_{\ga}\ldots\oint_{\ga}
	\prod_{i\ne j}\frac{\xi_i-\xi_j}{\Sm(\xi_i,\xi_j)}
	\prod_{j=1}^{k}\frac{d\xi_j}{\xi_j}
	F(\vxi)G(\vxi)
	,
\end{align}
where $F,G\in\Cc^{k}_{\xi}$, and
\begin{align}\label{quadratic_S_matrix}
	\Sm(\xi_1,\xi_2):=
	\frac{1-q}{1-q\nu}+\frac{q-\nu}{1-q\nu}\xi_2+
	\frac{\nu(1-q)}{1-q\nu}\xi_1\xi_2-
	\xi_1.
\end{align}
Here we have used
\begin{align*}
	\frac{\xinu(\xi_1)-\xinu(\xi_2)}{\xinu(\xi_1)-q\xinu(\xi_2)}
	=-\frac{1-\nu}{1-q \nu}\frac{\xi_1-\xi_2}{\Sm(\xi_1,\xi_2 )},
	\qquad
	\frac{1}{(1-\xinu(\xi))(1-\nu\xinu(\xi))}
	d\xinu(\xi)=\frac{d\xi}{(\nu-1)\xi}.
\end{align*}
In \eqref{Cxi_pairing_large}, $\ga$ can be taken to be the same contour containing $0$ and $1$ and
not containing $\nu^{-1}$ as in \eqref{Cz_pairing_large}
because $\xinu$ swaps $0\leftrightarrow1$ and
$\nu^{-1}\leftrightarrow\infty$ (note that this also means that
$\xinu$ does not change the orientation of the contour $\ga$).
Our definitions imply
\begin{align*}
	\llangle \xinu F, \xinu G\rrangle_{\Cc_{\xi}^k}=
	\llangle F, G\rrangle_{\Cc_{z}^k},\qquad
	F, G\in\Cc_{z}^{k}.
\end{align*}

In what follows we will use
both forms \eqref{Cz_pairing}
and \eqref{Cz_pairing_large}
of the pairing $\llangle \cdot,\cdot\rrangle_{\Cc_{z}^k}$, or the
equivalent pairing $\llangle \cdot,\cdot\rrangle_{\Cc_{\xi}^k}$
\eqref{Cxi_pairing_large}, depending on
convenience.

\begin{remark}\label{rmk:Povolotsky_parameters}
	The cross-term \eqref{quadratic_S_matrix} can be written in the form
	$\Sm(\xi_1,\xi_2)=\al\xi_1\xi_2+\be\xi_2+\gamma-\xi_1$,
	where $\al=\frac{\nu(1-q)}{1-q\nu}$,
	$\be=\frac{q-\nu}{1-q\nu}$, and $\gamma=\frac{1-q}{1-q\nu}$
	are the parameters
	from \cite{Povolotsky2013} (also used in \cite{Corwin2014qmunu}).
	Note that this means that $\al+\be+\gamma=1$ (see, however, \S \ref{sec:spectral_theory_for_the_conjugated_hahn_boson_operator} below on how to overcome this restriction).
	We do not use the parameters
	$\al,\be$, and $\gamma$ in the notation of the present paper.
\end{remark}


\subsection{Left and right eigenfunctions} 
\label{sub:eigenfunctions}

Here we define certain distinguished functions on $\C^{k}\times \Weyl{k}$
which are crucial to Plancherel isomorphism theorems between the spaces
$\Wc^{k}$ and $\Cc_{z}^{k}$ (Theorems \ref{thm:spatial_Plancherel}
and \ref{thm:spectral_Plancherel} below). These functions are right and
left Bethe ansatz eigenfunctions (sometimes also called Bethe wave functions)
of the $q$-Hahn operator
which depends on parameters $(q,\nu)$,
as well as on an additional parameter $\mu$
(see \S \ref{sec:the_q_mu_nu_boson_process_and_coordinate_bethe_ansatz} below for more details).
They were introduced by Povolotsky \cite{Povolotsky2013}.

We will use the \emph{left eigenfunctions} written in the following form:
\begin{align}\label{Psil}
	\Psil_{\z}(\n)&:=
	\sum_{\sigma\in S(k)}\prod_{1\le B<A\le k}
	\frac{z_{\sigma(A)}-qz_{\sigma(B)}}
	{z_{\sigma(A)}-z_{\sigma(B)}}\prod_{j=1}^{k}
	\left(\frac{1-z_{\sigma(j)}}{1-\nu z_{\sigma(j)}}\right)^{-n_j}.
\end{align}

To write down the right eigenfunctions, we will need additional notation.
For each $k$-tuple $\n=(n_1\ge\ldots\ge n_k)\in\Weyl{k}$, let
$\c(\n)=(c_1,\ldots,c_{M(\n)})$ denote its \emph{cluster sizes}, so that
\begin{align}\label{n_clusters}
	n_1=n_2=\ldots=n_{c_1}>n_{c_1+1}=\ldots=n_{c_1+c_2}
	>\ldots>n_{c_1+\ldots+c_{M(\n)-1}+1}=\ldots=n_{c_1+\ldots+c_{M(\n)}}
\end{align}
(of course, $c_1+\ldots+c_{M(\n)}=k$). For example, if $\n=(5,5,3,-1,-1)$, then
$\c(\n)=(2,1,2)$.

Next, for $\n\in\Weyl{k}$, define
\begin{align}\label{stationary_measure}
	\st(\n):=\prod_{j=1}^{M(\n)}\frac{(\nu;q)_{c_j}}{(q;q)_{c_j}}.
\end{align}
This product is related to the stationary measures
of the
$q$-Hahn stochastic particle system, see
\cite{Povolotsky2013}, \cite{Corwin2014qmunu}.

The \emph{right eigenfunctions} are, by definition,
\begin{align}\label{Psir}
	\Psir_{\z}(\n)&:=
	(-1)^k(1-q)^{k}q^{\frac{k(k-1)}{2}}\st(\n)
	\sum_{\sigma\in S(k)}\prod_{1\le B<A\le k}
	\frac{z_{\sigma(A)}-q^{-1}z_{\sigma(B)}}
	{z_{\sigma(A)}-z_{\sigma(B)}}\prod_{j=1}^{k}
	\left(\frac{1-z_{\sigma(j)}}{1-\nu z_{\sigma(j)}}\right)^{n_j}.
\end{align}
Note that as functions in $\z$, the left and right eigenfunctions
belong to the space $\Cc^{k}_{z}$
of symmetric Laurent polynomials in $\frac{1-z_i}{1-\nu z_i}$.
Indeed, this is because $\Vand(\z)\Psir_{\z}(\n)$
and
$\Vand(\z)\Psil_{\z}(\n)$
are skew-symmetric Laurent polynomials, and thus are
divisible by the Vandermonde determinant $\Vand(\z)$
with the ratios still in $\Cc^{k}_{z}$.

\begin{remark}\label{rmk:constant_to_make_nicer}
	The constant $(-1)^k(1-q)^{k}q^{\frac{k(k-1)}{2}}$ in front of $\Psir_{\z}(\n)$
	(not depending on $\n$)
	helps to make certain formulas below look nicer.
\end{remark}

%

The functions $\Psil_\z$ and $\Psir_\z$ arise as left and right eigenfunctions
of the $q$-Hahn Markov transition operator, see
Propositions \ref{prop:Psibwd_eigenvalue} and
\ref{prop:Psifwd_eigenvalue} below.

The $q$-Hahn stochastic particle system
is PT-invariant,
i.e., invariant under joint space-reflection and time-reversal
(see \S \ref{sub:formal_stationary_vector_pt_symmetry_and_the_right_eigenfunctions}
below).
This PT-symmetry translates into a certain property of the eigenfunctions
which we are about to explain.

Let $\Refl$ be the \emph{reflection operator}
acting on $\Wc^{k}$ as
\begin{align}\label{reflection_operator}
	(\Refl f)(n_1,\ldots,n_k):=f(-n_k,\ldots,-n_1).
\end{align}
This is an involution, i.e., $\Refl ^{-1}=\Refl$.
One can readily check that the eigenfunctions satisfy
\begin{align}\label{Psi_PT_symmetry}
	\big(\Refl \Psil_{\z}\big)(\n)=(-1)^k(1-q)^{-k}\st^{-1}(\n)\Psir_{\z}(\n);\qquad
	\big(\Refl \Psir_{\z}\big)(\n)=(-1)^k(1-q)^k\st(\n)\Psil_{\z}(\n).
\end{align}
Therefore, the operator on $\Wc^{k}$ defined as
\begin{align}\label{Psiswap_operator}
	(\Psiswap f)(\n):=(-1)^{k}(1-q)^{-k}\st^{-1}(\n)
	(\Refl f)(\n)
\end{align}
swaps left and right eigenfunctions:
\begin{align}\label{Psiswap_property}
	\big(\Psiswap \Psir_{\z}\big)(\n)=\Psil_{\z}(\n),
	\qquad
	\big(\Psiswap^{-1} \Psil_{\z}\big)(\n)=\Psir_{\z}(\n).
\end{align}
Observe that
$\st(\n)$ is invariant under the reflection $\Refl$.
Note also that the operators $\Refl$
and $\Psiswap$ are symmetric
with respect to $\llangle\cdot,\cdot\rrangle_{\Wc^k}$
in the sense that
$\llangle f,\Refl g\rrangle_{\Wc^k}=
\llangle \Refl f, g\rrangle_{\Wc^k}$, and the same for $\Psiswap$.

\medskip

For future use, let us record how the eigenfunctions
look in the other spectral variables $\xi_j$
(see~\S \ref{sub:spatial_and_spectral_variables}).
We have
\begin{align}
	\Psil_{\xinu(\vxi)}(\n)&=
	\left(-\frac{1-q\nu}{1-\nu}\right)^{\frac{k(k-1)}{2}}
	\sum_{\sigma\in S(k)}\prod_{1\le B<A\le k}
	\frac{\Sm(\xi_{\sigma(A)},\xi_{\sigma(B)})}
	{\xi_{\sigma(A)}-\xi_{\sigma(B)}}
	\prod_{j=1}^{k}
	\xi_{\sigma(j)}^{-n_j};
	\label{Psil_xi}
	\\
	\Psir_{\xinu(\vxi)}(\n)&=
	(-1)^k(1-q)^{k}\st(\n)
	\left(\frac{1-q\nu}{1-\nu}\right)^{\frac{k(k-1)}{2}}
	\sum_{\sigma\in S(k)}\prod_{1\le B<A\le k}
	\frac{\Sm(\xi_{\sigma(B)},\xi_{\sigma(A)})}
	{\xi_{\sigma(A)}-\xi_{\sigma(B)}}
	\prod_{j=1}^{k}
	\xi_{\sigma(j)}^{n_j}.
	\label{Psir_xi}
\end{align}
Here $\Sm$ is given by \eqref{quadratic_S_matrix}.



\section{Plancherel formulas} 
\label{sec:main_results}

\subsection{Direct and inverse transforms} 
\label{sub:transform_and_inverse_transform}

The \emph{direct $q$-Hahn transform} $\Pld$ takes
functions $f$ from the space $\Wc^{k}$
(of compactly supported functions
on the Weyl chamber $\Weyl{k}$)
to functions $\Pld f\in\Cc^{k}_{z}$ via
the bilinear pairing on $\Wc^{k}$:
\begin{align}\label{Pld_transform}
	(\Pld f)(\z):=\llangle f,\Psir_{\z}\rrangle_{\Wc^k}=
	\sum_{\n\in\Weyl{k}}f(\n)\Psir_{\z}(\n).
\end{align}

The (candidate) \emph{inverse $q$-Hahn transform}
$\Pli$
takes functions $G$ from the space $\Cc^{k}_{z}$
to functions $\Pli G\in\Wc^{k}$, and is defined as
\begin{align}\label{Pli_transform}
	(\Pli G)(\n):=
	\oint_{\ga_1}\frac{dz_1}{2\pi\i}
	\ldots
	\oint_{\ga_k}\frac{dz_k}{2\pi\i}
	\prod_{A<B}\frac{z_A-z_B}{z_A-qz_B}
	\prod_{j=1}^{k}
	\frac{1}{(1-z_j)(1-\nu z_j)}
	\left(\frac{1-z_j}{1-\nu z_j}\right)^{-n_j}
	G(\z).
\end{align}
The integration is performed over nested contours as in Definition \ref{def:contours}.

It is obvious that $\Pld$ maps $\Wc^{k}$ to $\Cc^{k}_{z}$.
To see that
$\Pli$ maps $\Cc^{k}_{z}$ to $\Wc^{k}$ requires simple residue calculus. Indeed,
if $n_k\le -M$ for a sufficiently large $M$ (this bound of course
depends on $G(\z)$), then the integrand
in \eqref{Pli_transform} is regular at $z_k=1$, and thus the integral vanishes.
On the other hand, if $n_1\ge M'$ for some other sufficiently large $M'$,
then the integrand is regular outside the $z_1$ contour
(namely, it is regular at points $z_1=\infty$ and $\nu^{-1}$),
and thus the integral also vanishes.
This implies that $\Pli G$
has a finite support in $\n$.

\begin{remark}
	It is convenient for us to work with the
	function spaces $\Wc^{k}$ and $\Cc^{k}_{z}$,
	because it eliminates
	analytic issues (such as proving convergence in \eqref{Pld_transform}).
	It seems plausible that the maps $\Pld$ and $\Pli$ can be extended to larger
	function spaces, but we do not focus on this question in \S \ref{sec:main_results}
	and \S \ref{sec:spectral_biorthogonality_of_eigenfunctions}.
	See, however, arguments in the proof of Proposition \ref{prop:analytic_argument} below
	which we need for applications of
	our main results.
\end{remark}


\subsection{Operator $\Pli$ and bilinear pairing in $\Cc^{k}_{z}$} 
\label{sub:various_forms_of_contour_integration}

The operator $\Pli$ can be expressed through the bilinear pairing
in the space $\Cc^{k}_{z}$ with the help of the left eigenfunction $\Psil_{\z}(\n)$
(cf. \eqref{Pld_transform}):
\begin{align}\label{Pli_as_bilinear_pairing}
	(\Pli G)(\n)=\llangle G,\Psil(\n)\rrangle_{\Cc_{z}^k},
\end{align}
where $\Psil(\n)$ is viewed as the function in $\Cc^{k}_{z}$
which maps $\z$ to $\Psil_{\z}(\n)$.

Expanding the definition of the bilinear pairing
(given by \eqref{Cz_pairing} or \eqref{Cz_pairing_large}),
we can write
\begin{align}
	\label{Pli_and_Psil_large}
	(\Pli G)(\n)
	&=
	\oint_{\ga}\ldots\oint_{\ga}
	d\Plm^{(q)}_{(1^{k})}(\z)
	\prod_{j=1}^{k}\frac{1}{(1-z_j)(1-\nu z_j)}
	\Psil_{\z}(\n)
	G(\z)\\&=
	\sum_{\la\vdash k}\oint_{\ga_k}\ldots\oint_{\ga_k}
	d\Plm^{(q)}_\la(\w)
	\prod_{j=1}^{\ell(\la)}\frac{1}{(w_j;q)_{\la_j}(\nu w_j;q)_{\la_j}}
	\Psil_{\w\circ\la}(\n)
	G(\w\circ\la).
	\label{Pli_and_Psil_small}
\end{align}

The equivalence of formulas \eqref{Pli_transform},
\eqref{Pli_and_Psil_large}, and \eqref{Pli_and_Psil_small}
for the (candidate) inverse transform $\Pli$
follows from the next proposition:
\begin{proposition}
\label{prop:nesting_unnesting}
	Let $G\colon\C^{k}\to\C$ be a symmetric function.\footnote{In this proposition we do not
	require $G$ to belong to $\Cc^{k}_{z}$.}
	Consider contours $\ga_1,\ldots,\ga_k$ and $\ga$ as in
	Definition \ref{def:contours}.
	\begin{enumerate}[\rm{}(1)]
		\item If there exist deformations $D_j^{\mathrm{large}}$ of
		$\ga_j$ to $\ga$\footnote{That is,
		$D_j^{\mathrm{large}}(t)$ is a family of contours which
		continuously depend on an additional parameter $t\in[0,1]$
		such that $D_j^{\mathrm{large}}(0)=\ga_j$ and
		$D_j^{\mathrm{large}}(1)=\ga$.} so that for all $z_1,\ldots,z_{j-1},z_{j+1},\ldots
		z_k$ with $z_i\in\ga$ for $1\le i<j$ and $z_i\in\ga_i$ for $j<i\le k$,
		the function $z_j\mapsto
		\Vand(\z)G(z_1,\ldots,z_{j},\ldots,z_k)$ is holomorphic in a neighborhood
		of the area swept out by the deformation $D_j^{\mathrm{large}}$, then
		the nested contour integral formula \eqref{Pli_transform}
		for
		$(\Pli G)(\n)$ can be rewritten in the form \eqref{Pli_and_Psil_large}.
		\item
		If there exist deformations $D_j^{\mathrm{small}}$ of
		$\ga_j$ to $\ga_{k}$ so that for all $z_1,\ldots,z_{j-1},z_{j+1},\ldots
		z_k$ with $z_i\in\ga_i$ for $1\le i<j$ and $z_i\in\ga_k$ for $j<i\le k$,
		the function $z_j\mapsto
		\Vand(\z)G(z_1,\ldots,z_{j},\ldots,z_k)$ is holomorphic in a neighborhood
		of the area swept out by the deformation $D_j^{\mathrm{small}}$, then
		the nested contour integral formula \eqref{Pli_transform}
		for
		$(\Pli G)(\n)$ can be rewritten in the form \eqref{Pli_and_Psil_small}.
	\end{enumerate}
\end{proposition}
	Note that if $G\in\Cc^{k}_{z}$,
	then it clearly
	satisfies the conditions of both parts of the proposition.
\begin{proof}
	This proposition is essentially contained in
	\cite{BorodinCorwinPetrovSasamoto2013} except for the additional factors
	$(1-\nu z_j)^{-1}$ in \eqref{Pli_transform} and \eqref{Pli_and_Psil_large},
	and $(\nu w_j;q)_{\la_j}^{-1}$ in \eqref{Pli_and_Psil_small}.
	Since $0\le \nu<1$ and thus the pole at $\nu^{-1}$
 	lies outside all contours $\ga_1,\ldots,\ga_k$, and $\ga$, these new
	factors do not affect the proofs
	given in \cite{BorodinCorwinPetrovSasamoto2013}.
	(Another way to say the same would be to incorporate these new
	factors inside the function $G(\z)$.)

	Part 1 of the proposition is \cite[Lemma 3.3]{BorodinCorwinPetrovSasamoto2013}.
	Its proof is fairly straightforward:
	deform contours $\ga_1,\ga_2,\ldots,\ga_k$
	(in this order) to $\ga$ using deformations $D_j^{\mathrm{large}}$.
	Now the integration is over the same contour $\ga$,
	so the integral is invariant under permutations of the $z_j$'s.
	Rewriting
	\begin{align*}
		\prod_{1\le A<B\le k}\frac{z_A-z_B}{z_A-qz_B}
		=
		\prod_{1\le A\ne B\le k}\frac{z_A-z_B}{z_A-qz_B}
		\prod_{1\le B<A\le k}\frac{z_A-qz_B}{z_A-z_B},
	\end{align*}
	we note that the symmetrization of the integrand over the whole
	symmetric group $S(k)$ gives $G(\z)$ times $\Psil_{\z}(\n)$ \eqref{Psil},
	times the corresponding Plancherel measure, see \eqref{dmu_large}.
	The factor $1/k!$ in \eqref{dmu_large} comes from the symmetrization.

	Part 2 of the proposition is
	\cite[Lemma 3.4]{BorodinCorwinPetrovSasamoto2013}
	(see also \cite[Proposition 7.4]{BorodinCorwinPetrovSasamoto2013}
	and \cite[Proposition 3.2.1]{BorodinCorwin2011Macdonald}). Its proof is not nearly as straightforward
	as that of part 1,
	as it involves keeping track of multiple
	residues arising in the course of
	deformations of the $\ga_j$'s to $\ga_k$.
	We omit the proof here.
\end{proof}

\begin{remark}\label{rmk:direct_inverse_xi}
	Similar direct and inverse $q$-Hahn transforms can be defined
	in the other spectral variables $\vxi=(\xi_1,\ldots,\xi_k)$,
	cf. \S \ref{sub:spatial_and_spectral_variables}.
	These transforms will go between
	the spaces $\Wc^{k}$ and
	$\Cc^{k}_{\xi}$ (recall that the latter
	is isomorphic to $\Cc^{k}_{z}$
	by means of
	$\xinu$).

	The direct transform $\Wc^{k}\to \Cc^{k}_{\xi}$
	looks as $f\mapsto \llangle f, \xinu\Psir_{\z}\rrangle_{\Wc^k}$ (cf. \eqref{Pld_transform}),
	where $\xinu\Psir_{\z}$ is the image of $\Psir$ under $\xinu$ applied in each variable.
	Writing the corresponding inverse transform
	$\Cc^{k}_{\xi}\to\Wc^{k}$ amounts to making the
	change of variables $\xinu$ under the integral in the definition of $\Pli$.
	This can be done similarly to \eqref{Cxi_pairing_large}
	if one uses formula \eqref{Pli_and_Psil_large}
	for $\Pli$.

	Analogues of our main results below
	(Theorems \ref{thm:spatial_Plancherel} and \ref{thm:spectral_Plancherel})
	will clearly hold for these transforms
	relative to the space
	$\Cc^{k}_{\xi}$, too.
	In fact, we will employ the variables $\vxi$
	to prove the spectral Plancherel formula (Theorem \ref{thm:spectral_Plancherel}).
\end{remark}


\subsection{The spatial Plancherel formula} 
\label{sub:spatial_plancherel_formula}

The composition $\Plspatial:=\Pli \Pld$ maps the space $\Wc^{k}$
of compactly supported functions in spatial variables
to itself via
\begin{align}
	\label{Plspatial}
	(\Plspatial f)(\n)&=(\Pli \Pld f)(\n)
	\\&=
	\oint_{\ga_1}\frac{dz_1}{2\pi\i}
	\ldots
	\oint_{\ga_k}\frac{dz_k}{2\pi\i}
	\prod_{1\le A<B\le k}\frac{z_A-z_B}{z_A-qz_B}
	\prod_{j=1}^{k}
	\frac{1}{(1-z_j)(1-\nu z_j)}
	\left(\frac{1-z_j}{1-\nu z_j}\right)^{-n_j}
	\llangle f,\Psir_{\z}\rrangle_{\Wc^k}.
	\nonumber
\end{align}

\begin{theorem}\label{thm:spatial_Plancherel}
	The direct $q$-Hahn transform $\Pld$
	induces an isomorphism between the space $\Wc^{k}$
	and its image inside $\Cc^{k}_{z}$, with the inverse given by $\Pli$. Equivalently,
	$\Plspatial$ \eqref{Plspatial} acts as the identity operator
	on $\Wc^{k}$. Moreover, $\Pld$
	acts on the bilinear pairing as follows:
	\begin{align}
		\llangle f,g\rrangle_{\Wc^k}=
		\llangle \Pld f,\Pld (\Psiswap g)\rrangle_{\Cc_{z}^k}
		\qquad
		\text{for all $f,g\in \Wc^{k}$.}
		\label{Plspatial_isomorphism}
	\end{align}
	Here $\Psiswap$ is the operator
	\eqref{Psiswap_operator}
	in $\Wc^{k}$
	which turns each right eigenfunction into the corresponding left one.
\end{theorem}

The rest of this subsection is devoted to the proof of the
theorem. In \S \ref{sub:completeness}
below
we present two immediate corollaries of this theorem.

\medskip

First, we note that \eqref{Plspatial_isomorphism}
readily follows from the fact that $\Plspatial=\Id$ on $\Wc^{k}$.
Fix $\y\in\Weyl{k}$, and let $g(\n):=\mathbf{1}_{\n=\y}$.
Then it follows from \eqref{Psiswap_property} that
$\big(\Pld (\Psiswap g)\big)(\z)=\llangle \Psir_{\z},\Psiswap g\rrangle_{\Wc^k}
=\llangle \Psiswap\Psir_{\z}, g\rrangle_{\Wc^k}=\Psil_{\z}(\y)$.
If $\Plspatial=\Id$, then we can write using \eqref{Pli_as_bilinear_pairing}:
\begin{align*}
	\llangle f, g\rrangle_{\Wc^k}&=
	f(\y)=
	(\Pli\Pld f)(\y)
	=
	\llangle \Pld f, \Psil(\y)\rrangle_{\Cc_{z}^k}
	=\llangle \Pld f,\Pld (\Psiswap g)\rrangle_{\Cc_{z}^k}.
\end{align*}
The case of general $g\in\Wc^{k}$ follows by linearity of our bilinear
pairings. About the need for the operator $\Psiswap$ in \eqref{Plspatial_isomorphism}
see also Remark \ref{rmk:biorthogonality_vs_orthogonality} below.

\smallskip

Our goal now is to show that $\Plspatial = \Id$ on $\Wc^{k}$.
It suffices to prove this operator identity
on a function $f(\n)=\mathbf{1}_{\n=\y}$
for any fixed $\y\in\Weyl{k}$.
Then in \eqref{Plspatial} we have
$\llangle f,\Psir_{\z}\rrangle_{\Wc^k}=\Psir_{\z}(\y)$,
so we need to show that
\begin{align}\label{Plancherel_identity_representation}
	\oint_{\ga_1}\frac{dz_1}{2\pi\i}
	\ldots
	\oint_{\ga_k}\frac{dz_k}{2\pi\i}
	\prod_{1\le A<B\le k}\frac{z_A-z_B}{z_A-qz_B}
	\prod_{j=1}^{k}
	\frac{1}{(1-z_j)(1-\nu z_j)}
	\left(\frac{1-z_j}{1-\nu z_j}\right)^{-n_j}
	\Psir_{\z}(\y)=\mathbf{1}_{\y=\n}.
\end{align}
We note that using \eqref{Pli_and_Psil_large}--\eqref{Pli_and_Psil_small}, one can rewrite
\eqref{Plancherel_identity_representation} in terms of integration over
$z_1,\ldots,z_k$ belonging to one and
the same contour (large $\ga$ or small $\ga_k$), cf. \eqref{completeness_big_1}--\eqref{completeness_small_2} below.

To prove \eqref{Plancherel_identity_representation},
rewrite the right eigenfunction in the following way:
\begin{align}
	\Psir_{\z}(\y)
	=
	(-1)^k(1-q)^{k}\st(\y)
	\sum_{\sigma\in S(k)}\prod_{1\le A<B\le k}
	\frac{z_{\sigma(A)}-qz_{\sigma(B)}}{z_{\sigma(A)}-z_{\sigma(B)}}
	\prod_{j=1}^{k}
	\left(\frac{1-z_{j}}{1-\nu z_{j}}\right)^{y_{\sigma^{-1}(j)}}.
	\label{Psir_permutation_representation}
\end{align}
We will show that
each summand in \eqref{Plancherel_identity_representation}
corresponding to a fixed permutation
$\sigma\in S(k)$ above vanishes unless $\n=\y$:
\begin{lemma}\label{lemma:Plspatial_permutation_vanishing}
	Let $\n,\y\in\Weyl{k}$.
	Then
	for any fixed $\sigma\in S(k)$, we have
	\begin{align}
		\oint_{\ga_1}\frac{dz_1}{2\pi\i}
		\ldots
		\oint_{\ga_k}\frac{dz_k}{2\pi\i}
		\prod_{A<B}\frac{z_A-z_B}{z_A-qz_B}
		\frac{z_{\sigma(A)}-qz_{\sigma(B)}}{z_{\sigma(A)}-z_{\sigma(B)}}
		\prod_{j=1}^{k}
		\frac{1}{(1-z_j)(1-\nu z_j)}\left(\frac{1-z_j}{1-\nu z_j}\right)^{-n_j+y_{\sigma^{-1}(j)}}
		=0
		\label{Plancherel_lemma}
	\end{align}
	unless $\y=\n$.
\end{lemma}
\begin{proof}
	We will split the proof into several parts.

	\smallskip\par\noindent\textbf{I. Shrinking and expanding contours.}
	Observe the following properties of
	the integral in \eqref{Plancherel_lemma} coming
	from the presence of powers of $(1-z_j)/(1-\nu z_j)$:
	\begin{enumerate}[(1)]
		\item If it is possible to shrink the contour $\ga_j$ to 1 (without picking any residues in the process),
		then
		in order for the left-hand side of
		\eqref{Plancherel_lemma}
		to be nonzero, we must have
		\begin{align*}
			-n_j+y_{\sigma^{-1}(j)}\le 0.
		\end{align*}
		\item If it is possible to expand the contour $\ga_j$ to $\nu^{-1}$
		(without picking any residues in the process; note that the integrand is regular at infinity and thus has to residue there),
		then
		in order for the left-hand side of
		\eqref{Plancherel_lemma}
		to be nonzero, we must have
		\begin{align*}
			-n_j+y_{\sigma^{-1}(j)}\ge 0.
		\end{align*}
	\end{enumerate}

	Which of the contours can be shrunk or expanded is determined by the
	product over inversions in the permutation $\sigma$:
	\begin{align}
		\prod_{A<B}\frac{z_A-z_B}{z_A-qz_B}
		\frac{z_{\sigma(A)}-qz_{\sigma(B)}}{z_{\sigma(A)}-z_{\sigma(B)}}=
		\sgn(\sigma)\prod_{A<B}
		\frac{z_{\sigma(A)}-qz_{\sigma(B)}}{z_A-qz_B}
		=
		\sgn(\sigma)\prod_{A<B\colon \sigma(A)>\sigma(B)}
		\frac{z_{\sigma(A)}-qz_{\sigma(B)}}{z_{\sigma(B)}^{}-qz_{\sigma(A)}^{}}.
		\label{Plancherel_lemma_proof1}
	\end{align}
	Namely:
	\begin{enumerate}[(1)]
		\item If for some $i\in\{1,2,\ldots,k\}$ one has $\sigma(i)>\sigma(1),\sigma(2),\ldots,\sigma(i-1)$
		(that is, if the position $i$ corresponds to the \emph{running maximum} in the permutation word $\sigma$),
		then the contour $\ga_{\sigma(i)}$ can be shrunk to $1$, and so $n_{\sigma(i)}\ge y_i$
		(or else the left-hand side of \eqref{Plancherel_lemma} vanishes).

		Indeed, this condition means that all numbers $s\in\{1,2,\ldots,k\}$ with $s> \sigma(i)$ lie to the right of
		$\sigma(i)$ in the permutation word $\sigma=(\sigma(1),\sigma(2),\ldots,\sigma(k))$.
		Therefore, the product $\prod_{A<B}(z_{\sigma(A)}-qz_{\sigma(B)})$ in the numerator in
		left-hand side
		of \eqref{Plancherel_lemma_proof1} contains all terms of the form $(z_{\sigma(i)}-qz_{\sigma(i)+1}),
		(z_{\sigma(i)}-qz_{\sigma(i)+2}),\ldots$. These terms cancel with the corresponding terms
		in the denominator, and so the integrand in \eqref{Plancherel_lemma} does not have
		poles at $z_{\sigma(i)}=qz_{\sigma(i)+1}, qz_{\sigma(i)+2},\ldots, qz_{k}$. Thus, it is possible to
		shrink the contour $\ga_{\sigma(i)}$ to $1$.

		\item If for some $i\in\{1,2,\ldots,k\}$ one has $\sigma(i)<\sigma(i+1),\sigma(i+2),\ldots,\sigma(k)$
		(that is, if the position $i$ corresponds to the \emph{running minimum} in the permutation word
		read from right to left),
		then the contour $\ga_{\sigma(i)}$ can be expanded to $\nu^{-1}$, and so $n_{\sigma(i)}\le y_i$
		(or else the left-hand side of \eqref{Plancherel_lemma} vanishes).

		The argument is similar to the case above: one can expand the contour $\ga_{\sigma(i)}$
		if the integrand does not have a pole at $z_{\sigma(i)}=q^{-1}z_{s}$ for all $s<\sigma(i)$.
	\end{enumerate}

	\smallskip\par\noindent\textbf{II. Arrow diagram of a permutation.}
	We will represent components of the vectors
	$\n=(n_1\ge n_2\ge \ldots\ge n_k)$
	and
	$\y=(y_1\ge y_2\ge \ldots\ge y_k)$
	graphically at two consecutive levels.
	We will draw an arrow from $a$ to $b$
	if we know that $a\ge b$.
	The condition $\n,\y\in\Weyl{k}$
	can be pictured by the following arrows
	(independently of the permutation $\sigma$):
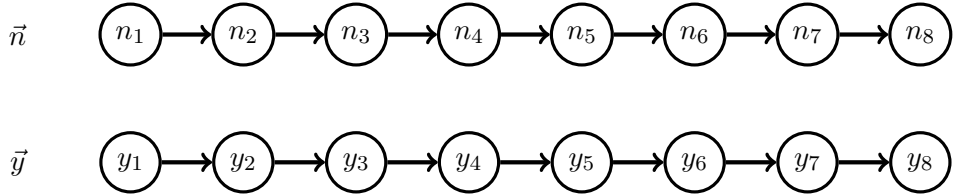
	\begin{align*}
		\parbox{.8\textwidth}{\begin{tikzpicture}
				[scale=1,nodes={draw},very thick]
				\def\h{1.7}
				\def\x{1.5}
				\node[draw=none] at (-1*\x,\h) {$\n$};
				\node[circle] (x1) at (0,\h) {$n_1$};
				\node[circle] (x2) at (\x,\h) {$n_2$};
				\node[circle] (x3) at (2*\x,\h) {$n_3$};
				\node[circle] (x4) at (3*\x,\h) {$n_4$};
				\node[circle] (x5) at (4*\x,\h) {$n_5$};
				\node[circle] (x6) at (5*\x,\h) {$n_6$};
				\node[circle] (x7) at (6*\x,\h) {$n_7$};
				\node[circle] (x8) at (7*\x,\h) {$n_8$};
				\node[draw=none] at (-1*\x,0) {$\y$};
				\node[circle] (y1) at (0,0) {$y_1$};
				\node[circle] (y2) at (\x,0) {$y_2$};
				\node[circle] (y3) at (2*\x,0) {$y_3$};
				\node[circle] (y4) at (3*\x,0) {$y_4$};
				\node[circle] (y5) at (4*\x,0) {$y_5$};
				\node[circle] (y6) at (5*\x,0) {$y_6$};
				\node[circle] (y7) at (6*\x,0) {$y_7$};
				\node[circle] (y8) at (7*\x,0) {$y_8$};
				\draw[ultra thick, ->] (x1)--(x2);
				\draw[ultra thick, ->] (x2)--(x3);
				\draw[ultra thick, ->] (x3)--(x4);
				\draw[ultra thick, ->] (x4)--(x5);
				\draw[ultra thick, ->] (x5)--(x6);
				\draw[ultra thick, ->] (x6)--(x7);
				\draw[ultra thick, ->] (x7)--(x8);
				\draw[ultra thick, ->] (y1)--(y2);
				\draw[ultra thick, ->] (y2)--(y3);
				\draw[ultra thick, ->] (y3)--(y4);
				\draw[ultra thick, ->] (y4)--(y5);
				\draw[ultra thick, ->] (y5)--(y6);
				\draw[ultra thick, ->] (y6)--(y7);
				\draw[ultra thick, ->] (y7)--(y8);
			\end{tikzpicture}}
	\end{align*}
	(Here
	and below we make illustrations for $k=8$.)
	It is convenient to relabel nodes in the upper row by
	$1,2,\ldots,k$, and in the lower row by $\sigma(1),
	\sigma(2),\ldots,\sigma(k)$. There are additional arrows between
	the lower and the upper level
	coming from expanding or shrinking the integration contours.
	These arrows will be drawn as follows (cf.
	Fig.~\ref{fig:arrow_diagram}):
	\begin{enumerate}[(1)]
		\item If $\sigma(i)$ is the running
		maximum at the lower level, connect it
		to $i$ at the upper level by a
		(\emph{red}) \emph{dashed arrow} connecting $i$ on the upper level
		(which is at position $i$) 
		with $i$ on the lower level
		(at position $\sigma(i)$).
		This corresponds to $n_{\sigma(i)}\ge y_i$.

		\item
		If $\sigma(i)$ is the running
		minimum at the lower level when the lower
		level is read from right to left, connect it
		to $i$ at the upper level by a
		(\emph{blue}) \emph{solid arrow} from 
		$i$ on the lower level (at position $\sigma(i)$)
		to $i$ on the upper level (the latter is at position $i$).
		This corresponds to $n_{\sigma(i)}\le y_{i}$.
	\end{enumerate}
	To complete the proof of the lemma, we need
	to show that such an arrow diagram always
	implies that $\n=\y$.
	(It is instructive and not difficult to check that the diagram
	on Fig.~\ref{fig:arrow_diagram}
	indeed implies the desired condition
	for our running example.)

	\begin{figure}[htbp]
		\begin{center}
		\parbox{.8\textwidth}{\begin{tikzpicture}
				[scale=1,nodes={draw},very thick]
				\def\h{3}
				\def\x{1.5}
				\node[draw=none] at (-1*\x,\h) {$\n$};
				\node[circle] (x1) at (0,\h) {$1$};
				\node[circle] (x2) at (\x,\h) {$2$};
				\node[circle] (x3) at (2*\x,\h) {$3$};
				\node[circle] (x4) at (3*\x,\h) {$4$};
				\node[circle] (x5) at (4*\x,\h) {$5$};
				\node[circle] (x6) at (5*\x,\h) {$6$};
				\node[circle] (x7) at (6*\x,\h) {$7$};
				\node[circle] (x8) at (7*\x,\h) {$8$};
				\node[draw=none] at (-1*\x,0) {$\y$};
				\node[circle] (y1) at (0,0) {$3$};
				\node[circle] (y2) at (\x,0) {$2$};
				\node[circle] (y3) at (2*\x,0) {$4$};
				\node[circle] (y4) at (3*\x,0) {$5$};
				\node[circle] (y5) at (4*\x,0) {$1$};
				\node[circle] (y6) at (5*\x,0) {$8$};
				\node[circle] (y7) at (6*\x,0) {$6$};
				\node[circle] (y8) at (7*\x,0) {$7$};
				\draw[ultra thick, ->] (x1)--(x2);
				\draw[ultra thick, ->] (x2)--(x3);
				\draw[ultra thick, ->] (x3)--(x4);
				\draw[ultra thick, ->] (x4)--(x5);
				\draw[ultra thick, ->] (x5)--(x6);
				\draw[ultra thick, ->] (x6)--(x7);
				\draw[ultra thick, ->] (x7)--(x8);
				\draw[ultra thick, ->] (y1)--(y2);
				\draw[ultra thick, ->] (y2)--(y3);
				\draw[ultra thick, ->] (y3)--(y4);
				\draw[ultra thick, ->] (y4)--(y5);
				\draw[ultra thick, ->] (y5)--(y6);
				\draw[ultra thick, ->] (y6)--(y7);
				\draw[ultra thick, ->] (y7)--(y8);
				\draw[line width = 2, ->, dashed, red] (x3) -- (y1);
				\draw[line width = 2, ->, dashed, red] (x4) -- (y3);
				\draw[line width = 2, ->, dashed, red] (x5) -- (y4);
				\draw[line width = 2, ->, dashed, red] (x8) -- (y6);
				\draw[line width = 2, ->, blue] (y5) -- (x1);
				\draw[line width = 2, ->, blue] (y7) -- (x6);
				\draw[line width = 2, ->, blue] (y8) -- (x7);
			\end{tikzpicture}}
			\end{center}
		\caption{Arrow diagram for $\sigma=(3,2,4,5,1,8,6,7)$.}
  		\label{fig:arrow_diagram}
	\end{figure}
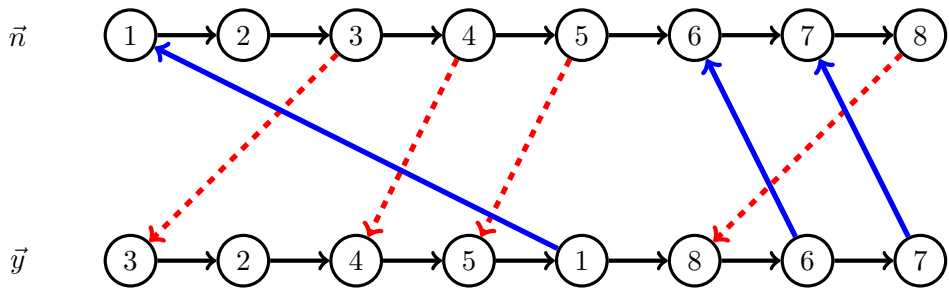

	\smallskip\par\noindent\textbf{III. One independent cycle of $\sigma$ and
	regions $B_\al$.}
	Let us partition $\{1,2,\ldots,k\}=C_1\sqcup C_2\sqcup \ldots \sqcup C_r$
	according to the representation of $\sigma$
	as a product of independent cycles. That is, each $C_\al$
	satisfies $\sigma(C_\al)=C_\al$,
	and is a minimal (i.e., indecomposable) subset with this property.
	
	Fix one of such parts $C_\al$.
	It is convenient to argue in terms of the
	graph of the permutation~$\sigma$.
	Such a graph is simply a collection of points
	$(i,\sigma(i))$, $i=1,\ldots,k$, on the plane
	(see Fig.~\ref{fig:sigma_graph}). 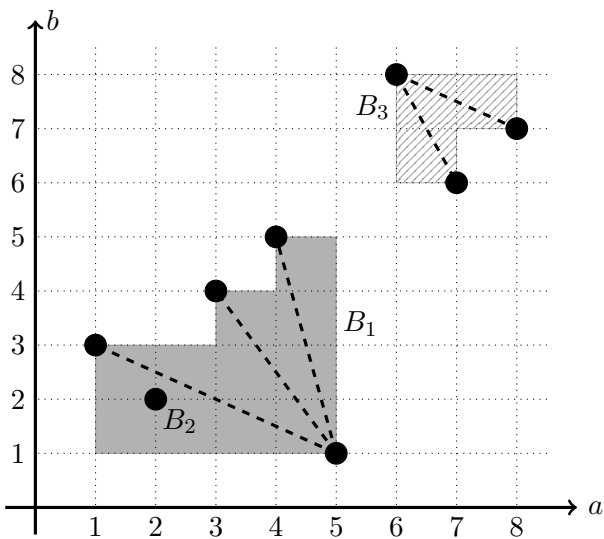
\begin{figure}[htbp]
		\begin{center}
		\begin{tikzpicture}
				[scale=.8]
				\def\h{.9};
				\draw[opacity=.6,gray,fill]
				(1,3*\h)--++(2,0)--++(0,\h)
				--++(1,0)--++(0,\h)
				--++(1,0)--++(0,-4*\h)
				--++(-4,0)--++(0,2*\h)
				--cycle;
				\draw[opacity=.6,gray,pattern=north east lines]
				(6,8*\h)
				--++(2,0)--++(0,-\h)
				--++(-1,0)--++(0,-\h)--++(-1,0)
				--cycle;
				\foreach \ll in {1,2,3,4,5,6,7,8}
				{
					\draw[dotted] (8.5,\ll*\h)--(0,\ll*\h) node[left] {\ll};
					\draw[dotted] (\ll,8.5*\h)--(\ll,0) node[below] {\ll};
				}
				\draw[->,very thick] (-.5,0)--(9,0) node [right] {$a$};
				\draw[->,very thick] (0,-.5*\h)--(0,9*\h)
				node [above right, yshift=-8] {$b$};
				\foreach \pt in {(1,3*\h),(2,2*\h),(3,4*\h),(4,5*\h),(5,1*\h),(6,8*\h),(7,6*\h),(8,7*\h)}
				{
					\draw[fill] \pt circle (.18);
				}
				\node at (5.4,3.4*\h) {$B_1$};
				\node at (2.4,1.6*\h) {$B_2$};
				\node at (5.6,7.4*\h) {$B_3$};
				\draw[dashed, very thick] (5,1*\h) -- (1,3*\h);
				\draw[dashed, very thick] (5,1*\h) -- (3,4*\h);
				\draw[dashed, very thick] (5,1*\h) -- (4,5*\h);
				\draw[dashed, very thick] (6,8*\h) -- (7,6*\h);
				\draw[dashed, very thick] (6,8*\h) -- (8,7*\h);
			\end{tikzpicture}
		\end{center}
	  	\caption{Graph of the permutation $\sigma=(3,2,4,5,1,8,6,7)$.
	  	Its decomposition into independent cycles looks as
	  	$\{1,2,\ldots,8\}=\{1,3,4,5\}\sqcup\{2\}\sqcup\{6,7,8\}.$
	  	The regions $B_1,B_2$, and $B_3$ are shown
	  	(note that $B_2=\{(2,2)\}\subset B_1$).
	  	Dashed lines corresponding to
	  	intersections of arrows in the arrow diagram
	  	are also shown.}
	  	\label{fig:sigma_graph}
	\end{figure}

	Define a region $B_\al$ in the same plane as
	\begin{align*}
		B_\al:=\{\text{points $(a,b)$}\colon
		\text{there exist $i,j\in C_\al$ such that
		$a\ge i$, $b\le \sigma(i)$ and $a\le j$, $b\ge \sigma(j)$}\}.
	\end{align*}
	Note that $B_\al$ can be represented as an intersection of
	two sets: the first set is a union of southeast cones
	attached to each point $(i,\sigma(i))$ with $i\in C_\al$, the second
	set is the analogous union of northwest cones (see Fig.~\ref{fig:sigma_graph}).

	Let us show that set $B_\al$ contains the whole diagonal
	$\{(a,a)\colon \min C_\al\le a\le \max C_\al\}$.
	We need to prove that
	(1) there exists a southeast cone covering $(a,a)$,
	i.e., a number $i\in C_\al$
	such that $i\le a$ and $\sigma(i)\ge a$; and
	(2) there exists a northwest cone covering $(a,a)$,
	i.e., a number $j\in C_\al$
	such that $j\ge a$ and $\sigma(j)\le a$.
	If $a=\min C_\al$, then $i=a$ and $j=\sigma^{-1}(a)$ would suffice,
	and case $a=\max C_\al$ is analogous.
	Assume that $\min C_\al<a<\max C_\al$,
	and that such $i$ (corresponding to a covering southeast
	cone) does not exist. Then for all $i\in C_\al$,
	the fact that $i\le a$ would imply that $\sigma(i)\le a$, too.
	This is not possible because $C_\al$ is an indecomposable subset with the
	property $\sigma(C_\al)=C_\al$. The existence of a covering
	northwest cone is established similarly.

	\smallskip\par\noindent\textbf{IV. Region $B$.}
	Define the region $B$ on the graph of the permutation
	$\sigma$ to be the union
	of all the regions $B_\al$ corresponding to independent
	cycles of $\sigma$. That is,
	\begin{align*}
		B=\{\text{points $(a,b)$}\colon
		\text{there exist $i$ and $j$ such that
		$i\le a$, $\sigma(i)\ge b$ and $j\ge a$, $\sigma(j)\le b$}\}.
	\end{align*}
	One readily sees that running maxima
	$\sigma(i)$
	in the permutation word $\sigma$
	correspond to northwest corners
	$(i,\sigma(i))$
	on the boundary of $B$; and running minima
	$\sigma(j)$
	of word $\sigma$
	read from right to left
	correspond to southeast corners
	$(j,\sigma(j))$ on the boundary
	of $B$. (This description of
	corners on the boundary does not always hold for
	the individual parts $B_\al$.)

	Let $B=\tilde B_1\sqcup \ldots\sqcup \tilde B_s$
	be the decomposition of the region $B$
	into connected subregions.
	Each $\tilde B_l$, $l=1,\ldots,s$
	is a union of some of the $B_\al$'s, $\al\in\{1,\ldots,r\}$
	(hence $s\le r$).
	In more detail, parts $B_\al$ and $B_\be$ belong to the same
	connected subregion $\tilde B_l$ if and only if
	the independent cycles $C_\al$ and $C_\be$ ``intersect''
	in the sense that
	$[\min C_\al,\max C_\al]\cap
	[\min C_\be,\max C_\be]\ne \varnothing$.
	For example, on Fig.~\ref{fig:sigma_graph}
	we have $B=(B_1\cup B_2)\sqcup B_3$.
	
	Since each $\tilde B_l$ is connected, it must
	contain the whole corresponding diagonal
	$\{(a,a)\colon m_l\le a\le M_l\}$, where
	\begin{align}
		m_l:=\min_{(a,b)\in \tilde B_l}a=
		\min_{(a,b)\in \tilde B_l}b,\qquad
		M_l:=
		\max_{(a,b)\in \tilde B_l}a=
		\max_{(a,b)\in \tilde B_l}b.\label{mlml_definition}
	\end{align}

	\smallskip\par\noindent\textbf{V. Dashed lines and intersections of arrows.}
	Let us connect a northwest boundary corner $(i,\sigma(i))$
	to a southeast boundary corner
	$(j,\sigma(j))$ by a \emph{dashed line} if
	$i<j$ and $\sigma(i)>\sigma(j)$, see Fig.~\ref{fig:sigma_graph}.
	Clearly, each dashed line must intersect the
	diagonal.
	While boundary corners
	correspond to arrows between the lower and the upper
	level in the arrow diagram as on Fig.~\ref{fig:arrow_diagram},
	these dashed lines correspond to intersections of
	two such arrows.
	It then readily follows from the arrow diagram
	that if
	a northwest corner $(i,\sigma(i))$ is connected to a
	southeast corner
	$(j,\sigma(j))$ by a dashed line, then
	\begin{align*}
		n_{\sigma(j)}=n_{\sigma(j)+1}=\ldots=n_{\sigma(i)}
		=y_{i}=y_{i+1}=\ldots=y_{j}.
	\end{align*}
	These equalities involve all $y$'s with labels
	belonging to the orthogonal projection of the dashed line onto
	the horizontal coordinate line; and all $n$'s
	with labels
	belonging to the analogous vertical projection.

	Now take any of the connected subregions $\tilde B_l$
	of $B$, where $l\in\{1,\ldots,s\}$. Note that
	the northwest corners $(i,\sigma(i))$ on the boundary
	of $\tilde B_l$
	must
	be located on or above the diagonal,
	while the southeast boundary corners
	$(j,\sigma(j))$ must be located on or below the diagonal.
	Since $\tilde B_l$ is connected and contains the whole corresponding
	diagonal from $m_l$ to $M_l$ \eqref{mlml_definition}, we conclude that
	one can get from any (northwest or southeast) boundary corner
	of $\tilde B_l$
	to \emph{any other} boundary corner along the dashed lines.
	This implies that
	\begin{align*}
		n_{m_l}=n_{m_l+1}=\ldots=n_{M_l}=
		y_{m_l}=y_{m_l+1}=\ldots=y_{M_l}.
	\end{align*}
	Because this has to hold for each
	connected subregion $\tilde B_l$, we see that
	$\n=\y$, and this completes the proof of the lemma.
\end{proof}

Let $\n\in\Weyl{k}$, and $\c=\c(\n)$ denote its cluster sizes
(as in \S \ref{sub:eigenfunctions}).
We say that $\sigma\in S(k)$ \emph{permutes within the clusters} of
$\n$ if $\sigma$ stabilizes
the sets
\begin{align*}
	\{1,2,\ldots,c_1\},\quad
	\{c_1+1,\ldots,c_1+c_2\},\quad
	\ldots,\quad
	\{c_1+\ldots+c_{M(\n)-1}+1,\ldots,k\}.
\end{align*}
Denote the set of all such permutations by $S_{\n}(k)$.
Every $\sigma\in S_\n(k)$ can be represented a product of
$M(\n)$ permutations $\sigma_1, \ldots, \sigma_{M(\n)}$,
where each $\sigma_i$
fixes all elements of $\{1,2,\ldots,k\}$
except those corresponding to the $i$-th cluster of $\n$. We will denote
the indexes within the $i$-th cluster by $C_i(\n)$,
and write $\sigma_i\in S_{\n,i}(k)$.
For example, if $k=5$ and $n_1=n_2=n_3>n_4=n_5$,
the permutation $\sigma\in S_\n(k)$ must stabilize the sets
$C_1(\n)=\{1,2,3\}$ and $C_2(\n)=\{4,5\}$; $S_{\n,1}(k)$ is the
set of all permutations of $\{1,2,3\}$, and $S_{\n,2}(k)$
is the set of all permutations of $\{4,5\}$.

Let us assume that $\y=\n$ in \eqref{Plancherel_lemma}.
By
considering the arrow diagram of $\sigma$
as in the
proof of Lemma \ref{lemma:Plspatial_permutation_vanishing},
one readily sees that the integral in
\eqref{Plancherel_lemma} vanishes unless $\sigma\in S_\n(k)$.
We will now compute the sum of these integrals
over all $\sigma\in S_\n(k)$.

\begin{lemma}\label{lemma:Plspatial_permutation_computation}
	For any $\n\in\Weyl{k}$,
	\begin{align}
		&\sum_{\sigma\in S_\n(k)}
		\oint_{\ga_1}\frac{dz_1}{2\pi\i}
		\ldots
		\oint_{\ga_k}\frac{dz_k}{2\pi\i}
		\prod_{A<B}\frac{z_A-z_B}{z_A-qz_B}
		\frac{z_{\sigma(A)}-qz_{\sigma(B)}}{z_{\sigma(A)}-z_{\sigma(B)}}
		\prod_{j=1}^{k}
		\frac{1}{(1-z_j)(1-\nu z_j)}\left(\frac{1-z_j}{1-\nu z_j}\right)^{-n_j+n_{\sigma^{-1}(j)}}
		\nonumber
		\\&\hspace{140pt}=
		(-1)^{k}(1-q)^{-k}
		\frac{(q;q)_{c_1}\ldots (q;q)_{c_{M(\n)}}}
		{(\nu;q)_{c_1}\ldots (\nu;q)_{c_{M(\n)}}}.
		\label{Plancherel_permuation_computation}
	\end{align}
\end{lemma}
\begin{proof}
	We have
	\begin{align*}
		&
		\sum_{\sigma\in S_\n(k)}
		\prod_{1\le A<B\le k}
		\frac{z_{\sigma(A)}-qz_{\sigma(B)}}{z_{\sigma(A)}-z_{\sigma(B)}}
		\prod_{j=1}^{k}
		\left(\frac{1-z_j}{1-\nu z_j}\right)^{-n_j+n_{\sigma^{-1}(j)}}
		\\&\hspace{60pt}=
		\prod_{1\le i<j\le M(\n)}
		\prod_{\substack{A\in C_i(\n)\\ B\in C_j(\n)}}
		\frac{z_A-qz_B}{z_A-z_B}
		\cdot
		\prod_{i=1}^{M(\n)}
		\sum_{\sigma_i\in S_{\n,i}(k)}
		\prod_{\substack{1\le A<B\le k\\
		A,B\in C_i(\n)}}
		\frac{z_{\sigma_i(A)}-qz_{\sigma_i(B)}}
		{z_{\sigma_i(A)}-z_{\sigma_i(B)}}.
	\end{align*}
	For the second product in the right-hand side above
	we may use the symmetrization identity
	\cite[III.1(1.4)]{Macdonald1995}
	which states that for any $\ell\ge1$,
	\begin{align}\label{Mac_symm_identity}
		\sum_{\omega\in S(\ell)}
		\frac{z_{\omega(A)}-qz_{\omega(B)}}
		{z_{\omega(A)}-z_{\omega(B)}}=\frac{(q;q)_{\ell}}{(1-q)^{\ell}}.
	\end{align}
	Thus, the left-hand side of \eqref{Plancherel_permuation_computation}
	takes the form
	\begin{align*}
		\frac{(q;q)_{c_1}\ldots (q;q)_{c_{M(\n)}}}
		{(1-q)^{k}}
		\oint_{\ga_1}\frac{dz_1}{2\pi\i}
		\ldots
		\oint_{\ga_k}\frac{dz_k}{2\pi\i}
		\prod_{i=1}^{M(\n)}
		\prod_{\substack{1\le A<B\le k\\
		A,B\in C_i(\n)}}\frac{z_A-z_B}{z_A-qz_B}
		\prod_{j=1}^{k}
		\frac{1}{(1-z_j)(1-\nu z_j)}.
	\end{align*}
	Now in the integral the integration variables corresponding
	to different clusters $C_i(\n)$ are independent,
	so the integral reduces to a product of
	$M(\n)$ smaller nested contour integrals of the same form.
	Each of these integrals is computed
	as follows:
	\begin{align*}
		\oint_{\ga_1}\frac{dz_1}{2\pi\i}
		\ldots
		\oint_{\ga_\ell}\frac{dz_\ell}{2\pi\i}
		\prod_{1\le A<B\le \ell}\frac{z_A-z_B}{z_A-qz_B}
		\prod_{j=1}^{k}
		\frac{1}{(1-z_j)(1-\nu z_j)}=\frac{(-1)^{\ell}}{(\nu;q)_{\ell}},
	\end{align*}
	this is a particular case $\n=(0,0,\ldots,0)$ and $c=1$
	of Proposition \ref{prop:TW_qnu} below.
	In this way we arrive at the desired right-hand side of
	\eqref{Plancherel_permuation_computation}.
\end{proof}

The next proposition is an explicit evaluation of a certain nested contour
integral. A particular case of it is used in the proof of Lemma \ref{lemma:Plspatial_permutation_computation} above. In its full generality this statement is employed in connection with the Tracy-Widom symmetrization
identities in \S \ref{sub:tracy_widom_symmetrization_identities} below (see also \S \ref{sub:symmetrization_formula}).

\begin{proposition}\label{prop:TW_qnu}
	For any $c\in\C\setminus\{\nu^{-1}, q^{-1}\nu^{-1},
	\ldots, q^{-(k-1)}\nu^{-1}\}$
	and any $\n\in\Weyl{k}$, we have
	\begin{align*}
		&\oint_{\ga_1}\frac{dz_1}{2\pi\i}
		\ldots
		\oint_{\ga_k}\frac{dz_k}{2\pi\i}
		\prod_{1\le A<B\le k}\frac{z_A-z_B}{z_A-qz_B}
		\prod_{j=1}^{k}\left(\frac{1-z_j}{1-\nu z_j}\right)^{-n_j}
		\frac{1}{(1-z_j)(1-c\nu  z_j)}
		\\&\hspace{220pt}=
		\frac{(-1)^{k}\nu^{n_1+n_2+\ldots+n_k}}{(c\nu;q)_{k}}\prod_{j=1}^{k}\left(\frac{1-cq^{j-1}}{1-c\nu q^{j-1}}\right)^{n_j}
		\mathbf{1}_{n_k\ge0}.
	\end{align*}
	Here the integration contours
	$\ga_1,\ldots,\ga_k$ are as in Definition \ref{def:contours},
	with an additional condition that they do
	not contain $c^{-1}\nu^{-1}$ (this is possible because of our restrictions on
	$c$; this could mean that each contour is a union of several disjoint
	simple contours).
\end{proposition}
By the very definition
of the inverse $q$-Hahn transform
\eqref{Pli_transform},
the left-hand side of the above identity can be interpreted as $(\Pli G)(\n)$,
where $G(\z)=\prod_{j=1}^{k}\frac{1-\nu z_j}{1-c\nu z_j}$.
Note however that this function $G$ does not belong to $\Cc^{k}_{z}$
(cf. the argument in Proposition \ref{prop:analytic_argument} below).
\begin{proof}
	First, note that
	the integrand has zero residue at $z_k=1$
	unless
	$n_k\ge0$, so let us assume $n_k\ge0$ (hence $n_j\ge0$ for all $j$).
	This implies that there is no pole at $z_j=\nu^{-1}$ for all $j$.
	
	Thus, the only pole of the integrand in $z_1$ outside the integration
	contour is at
	$z_1=c^{-1}\nu^{-1}$.\footnote{If $c\ne 0$, the
	integrand is regular at infinity.
	If $c=0$, then this only pole will be at $c^{-1}\nu^{-1}=\infty$.}
	So, we may evaluate the integral over $z_1$
	by taking minus
	the residue at $z_1=c^{-1}\nu^{-1}$:
	\begin{align*}
		&\oint_{\ga_1}\frac{dz_1}{2\pi\i}
		\ldots
		\oint_{\ga_k}\frac{dz_k}{2\pi\i}
		\prod_{1\le A<B\le k}\frac{z_A-z_B}{z_A-qz_B}
		\prod_{j=1}^{k}\left(\frac{1-z_j}{1-\nu z_j}\right)^{-n_j}
		\prod_{j=1}^{k}\frac{1}{(1-z_j)(1-c\nu z_j)}
		\\&\hspace{20pt}=
		-
		\left(-\frac{1}{c\nu(1-c^{-1}\nu^{-1})}\right)
		\left(\frac{1-c}{1-c\nu}\right)^{n_1}\nu^{n_1}
		\oint_{\ga_2}\frac{dz_2}{2\pi\i}
		\ldots
		\oint_{\ga_k}\frac{dz_k}{2\pi\i}
		\prod_{j=2}^{k}\frac{1-c\nu z_j}{1-q c \nu z_j}
		\\&\hspace{140pt}\times
		\prod_{2\le A<B\le k}\frac{z_A-z_B}{z_A-qz_B}
		\prod_{j=2}^{k}\left(\frac{1-z_j}{1-\nu z_j}\right)^{-n_j}
		\prod_{j=2}^{k}\frac{1}{(1-z_j)(1-c\nu z_j)}
		\\&\hspace{20pt}=
		-\frac{1}{1-c\nu}
		\left(\frac{1-c}{1-c\nu}\right)^{n_1}\nu^{n_1}
		\\&\hspace{60pt}\times\oint_{\ga_2}\frac{dz_2}{2\pi\i}
		\ldots
		\oint_{\ga_k}\frac{dz_k}{2\pi\i}
		\prod_{2\le A<B\le k}\frac{z_A-z_B}{z_A-qz_B}
		\prod_{j=2}^{k}\left(\frac{1-z_j}{1-\nu z_j}\right)^{-n_j}
		\prod_{j=2}^{k}\frac{1}{(1-z_j)(1-qc\nu z_j)}.
	\end{align*}
	We obtained a constant times a similar integral in $k-1$ variables
	$z_2,z_3,\ldots,z_k$, but with $c$ replaced by $qc$.
	Continuing an inductive evaluation of the integrals, we get the desired result.
\end{proof}

Now, looking at
\eqref{Psir_permutation_representation} and
\eqref{stationary_measure}, we see that
Lemmas \ref{lemma:Plspatial_permutation_vanishing} and \ref{lemma:Plspatial_permutation_computation}
imply identity
\eqref{Plancherel_identity_representation},
and thus the proof of
Theorem \ref{thm:spatial_Plancherel} is completed.

\begin{remark}\label{rmk:difference_in_proofs_of_Plancherel}
	Our proof of the spatial (direct) Plancherel
	formula differs from the one given for $\nu=0$
	in \cite[\S3.2]{BorodinCorwinPetrovSasamoto2013} (and it is not clear
	if the proof from \cite{BorodinCorwinPetrovSasamoto2013}
	can be directly adapted to the case $\nu\ne 0$).
	In the latter proof, a statement analogous to Lemma
	\ref{lemma:Plspatial_permutation_vanishing} was established
	via a so-called contour shift argument (also used,
	in particular, in the work of Heckman and Opdam
	\cite{HeckmannOpdam1997},
	and dating back to Helgason \cite{Helgason66}).
	The contour shift argument in \cite{BorodinCorwinPetrovSasamoto2013}
	employed the PT-symmetry property of the underlying
	$q$-Boson particle system
	(this property replaces the Hermitian symmetry used
	in contour shift arguments in earlier works).
	Instead of such an argument, above we have presented
	a direct combinatorial argument which
	took care of summands corresponding to each $\sigma\in S(k)$
	in \eqref{Psir_permutation_representation}
	separately. Note that our proof of
	Lemma \ref{lemma:Plspatial_permutation_vanishing}
	did not involve the PT-symmetry.
	On the other hand, the computation of the constant
	we performed in Lemma \ref{lemma:Plspatial_permutation_computation}
	is done somewhat in the spirit of \cite{BorodinCorwinPetrovSasamoto2013}.
\end{remark}


\subsection{The spectral Plancherel formula} 
\label{sub:spectral_plancherel_formula}

In this subsection, we assume in addition that
$0<q<1$.

The composition $\Plspectral:=\Pld \Pli$ of the
$q$-Hahn transforms (the order is reversed
comparing to $\Plspatial=\Pli\Pld$, cf.
\S \ref{sub:spatial_plancherel_formula})
maps the space $\Cc^{k}_{z}$ (\S \ref{sub:spatial_and_spectral_variables}) to itself.
It acts as
\begin{align}\label{Plspectral}
	&(\Plspectral G)(\w)=(\Pld \Pli G)(\w)\\&
	\hspace{10pt}=
	\sum_{\n\in\Weyl{k}}
	\Psir_{\w}(\n)
	\oint_{\ga_1}\frac{dz_1}{2\pi\i}
	\ldots
	\oint_{\ga_k}\frac{dz_k}{2\pi\i}
	\prod_{A<B}\frac{z_A-z_B}{z_A-qz_B}
	\prod_{j=1}^{k}
	\frac{1}{(1-z_j)(1-\nu z_j)}
	\left(\frac{1-z_j}{1-\nu z_j}\right)^{-n_j}
	G(\z).
	\nonumber
\end{align}

\begin{theorem}\label{thm:spectral_Plancherel}
	The inverse $q$-Hahn transform $\Pli$ induces an isomorphism between the
	space $\Cc^{k}_{z}$ and its image inside $\Wc^{k}$, with the inverse given by $\Pld$.
	Equivalently, $\Plspectral$ acts as the identity operator
	on $\Cc^{k}_{z}$. Moreover, $\Pli$
	acts on the bilinear pairing
	as follows:
	\begin{align}
		\llangle F,G \rrangle_{\Cc_{z}^k}=
		\llangle \Pli F,\Psiswap^{-1}(\Pli G)\rrangle_{\Wc^k}
		\qquad\text{for all $F,G\in\Cc^{k}_{z}$}.
		\label{Plspectral_isomorphism}
	\end{align}
	Here $\Psiswap$ is the operator
	\eqref{Psiswap_operator}
	in $\Wc^{k}$
	which turns each right eigenfunction into the corresponding left one.
\end{theorem}
\begin{proof}
	First, if $\Plspectral=\Id$ on $\Cc_{z}^{k}$, then one can obtain \eqref{Plspectral_isomorphism}
	in a way similar to \eqref{Plspatial_isomorphism}
	in the proof of the spatial Plancherel formula (Theorem \ref{thm:spatial_Plancherel}).

	\smallskip

	To establish the main part of the theorem, rewrite the nested contour
	integral in \eqref{Plspectral} in terms of integration over the large contour
	$\ga$:\footnote{This is possible because $G\in\Cc^{k}_{z}$,
	see \eqref{Pli_and_Psil_large} and Proposition \ref{prop:nesting_unnesting}.
	See also
	\eqref{dmu_large} for an explicit expression for the Plancherel
	measure corresponding to the large contour.}
	\begin{align*}
		&(\Plspectral G)(\w)
		=
		\sum_{\n\in\Weyl{k}}
		\Psir_{\w}(\n)
		\oint_{\ga}\ldots\oint_{\ga}
		\frac{d\z}{(2\pi\i)^{k}}
		\frac{1}{k!}
		\frac{(-1)^{\frac{k(k-1)}{2}}\Vand(\z)^2}{\prod_{i\ne j}(z_i-qz_j)}
		\Psil_{\z}(\n)G(\z)
		\prod_{j=1}^{k}\frac{1}{(1-z_j)(1-\nu z_j)}.
	\end{align*}
	To show that the right-hand side is equal to $G(\w)$, it suffices to establish
	the following integrated version of the above identity:
	For any $F\in\Cc^{k}_{z}$,
	\begin{align*}
		&
		\oint_{\ga}\ldots\oint_{\ga}
		\frac{d\w}{(2\pi\i)^{k}}
		F(\w)G(\w)
		=
		\oint_{\ga}\ldots\oint_{\ga}
		\frac{d\w}{(2\pi\i)^{k}}
		F(\w)
		\frac{(-1)^{\frac{k(k-1)}{2}}}{k!}
		\sum_{\n\in\Weyl{k}}
		\Psir_{\w}(\n)
		\\&\hspace{160pt}\times
		\oint_{\ga}\ldots\oint_{\ga}
		\frac{d\z}{(2\pi\i)^{k}}
		\frac{\Vand(\z)^2}{\prod_{i\ne j}(z_i-qz_j)}
		\Psil_{\z}(\n)G(\z)
		\prod_{j=1}^{k}\frac{1}{(1-z_j)(1-\nu z_j)}.
	\end{align*}
	It is possible to interchange the summation and the integration
	in the $\w$ variables (because of the finitely many nonzero terms in the
	sum). Thus, we must show that
	\begin{align}
		\nonumber
		&
		\oint_{\ga}\ldots\oint_{\ga}
		\frac{d\w}{(2\pi\i)^{k}}
		F(\w)G(\w)=
		\frac{(-1)^{\frac{k(k-1)}{2}}}{k!}
		\sum_{\n\in\Weyl{k}}
		\oint_{\ga}\ldots\oint_{\ga}
		\frac{d\w}{(2\pi\i)^{k}}
		\oint_{\ga}\ldots\oint_{\ga}
		\frac{d\z}{(2\pi\i)^{k}}
		\frac{\Vand(\z)^2}{\prod_{i\ne j}(z_i-qz_j)}
		\\&\hspace{180pt}\times
		\Psir_{\w}(\n)
		\Psil_{\z}(\n)
		F(\w)
		G(\z)
		\prod_{j=1}^{k}\frac{1}{(1-z_j)(1-\nu z_j)}.
		\label{spectral_Plancherel_via_biorthogonality_proof}
	\end{align}
	This statement follows from the spectral biorthogonality of the
	eigenfunctions, which we prove independently as Theorem
	\ref{thm:spectral_biorthogonality} below. Indeed, applying that theorem with the following two test functions\footnote{These functions do not belong to $\Cc^{k}_{z}$, but they satisfy the conditions of Theorem \ref{thm:spectral_biorthogonality}:
	the second function times
	$\Vand(\w)\prod_{j=1}^{k}(1-w_j)^{-M}(1-\nu w_j)^M$ is holomorphic
	in the closed exterior of the contour $\ga$ (including $\infty$);
	and the first function times $\Vand(\z)$
	is holomorphic between $\ga$ and $\ga'$
	(about the latter contour see Definition~\ref{def:contour_gap}).}
	\begin{align*}
		G(\z)\Vand(\z)
		\frac1{\prod_{i\ne j}(z_i-qz_j)}
		\prod_{j=1}^{k}\frac{1}{(1-z_j)(1-\nu z_j)}
		\qquad\text{and}\qquad
		F(\w)\frac{1}{\Vand(\w)},
	\end{align*}
	one can rewrite the right-hand side of \eqref{spectral_Plancherel_via_biorthogonality_proof}
	as
	\begin{align*}
		\frac{(-1)^{\frac{k(k-1)}{2}}}{k!}(-1)^{\frac{k(k-1)}2}
		\oint_{\ga}\ldots\oint_{\ga}
		\frac{d\w}{(2\pi\i)^{k}}
		G(\w)
		\sum_{\sigma\in S(k)}
		\sgn(\sigma)
		F(\sigma \w)\frac{\Vand(\w)}{\Vand(\sigma \w)},
	\end{align*}
	where $\sigma\w=(w_{\sigma(1)},\ldots,w_{\sigma(k)})$.
	The ratio of Vandermonde determinants gives $\sgn(\sigma)$,
	and so we arrive at the left-hand side of \eqref{spectral_Plancherel_via_biorthogonality_proof}.
	This completes the proof of the spectral Plancherel formula.
\end{proof}


\subsection{Completeness and spatial biorthogonality} 
\label{sub:completeness}

Here we record two immediate corollaries
of the spatial Plancherel formula
(Theorem \ref{thm:spatial_Plancherel}).
The first is the
\emph{completeness}
of the coordinate Bethe ansatz for the $q$-Hahn stochastic particle system
(discussed in
\S \ref{sec:the_q_mu_nu_boson_process_and_coordinate_bethe_ansatz} below).

\begin{corollary}\label{cor:completeness}
	Any function $f\in\Wc^{k}$ can be expanded as
	\begin{align}\label{completeness_big_1}
		f(\n)&=\sum_{\la\vdash k}
		\oint_{\ga_k}\ldots\oint_{\ga_k}
		d\Plm^{(q)}_\la(\w)
		\prod_{j=1}^{\ell(\la)}\frac{1}{(w_j;q)_{\la_j}(\nu w_j;q)_{\la_j}}
		\Psil_{\w\circ\la}(\n)
		\llangle f,\Psir_{\w\circ\la}\rrangle_{\Wc^k}
		\\&=
		\label{completeness_big_2}
		\sum_{\la\vdash k}
		\oint_{\ga_k}\ldots\oint_{\ga_k}
		d\Plm^{(q)}_\la(\w)
		\prod_{j=1}^{\ell(\la)}\frac{1}{(w_j;q)_{\la_j}(\nu w_j;q)_{\la_j}}
		\Psir_{\w\circ\la}(\n)
		\llangle \Psil_{\w\circ\la},f\rrangle_{\Wc^k}.
	\end{align}
	We also have expansions
	\begin{align}
		\label{completeness_small_1}
		f(\n)&=
		\oint_{\ga}\ldots\oint_{\ga}
		d\Plm^{(q)}_{(1^{k})}(\z)
		\prod_{j=1}^{k}\frac{1}{(1-z_j)(1-\nu z_j)}
		\Psil_{\z}(\n)
		\llangle f,\Psir_{\z}\rrangle_{\Wc^k}
		\\&=
		\oint_{\ga}\ldots\oint_{\ga}
		d\Plm^{(q)}_{(1^{k})}(\z)
		\prod_{j=1}^{k}\frac{1}{(1-z_j)(1-\nu z_j)}
		\Psir_{\z}(\n)
		\llangle \Psil_{\z},f\rrangle_{\Wc^k}.
		\label{completeness_small_2}
	\end{align}
\end{corollary}
\begin{proof}
	The expansion \eqref{completeness_big_1} follows by applying
	$\Plspatial$ to the function $f$ due to Theorem \ref{thm:spatial_Plancherel}
	and formula \eqref{Pli_and_Psil_small}.
	
	The swapping of the left and right eigenfunctions
	in \eqref{completeness_big_2} can be achieved
	by writing \eqref{completeness_big_1} for the function $\Psiswap f$
	(for the definition of $\Psiswap$, see \S \ref{sub:eigenfunctions}),
	which gives
	$\llangle \Psiswap f,\Psir_{\w\circ\la}\rrangle_{\Wc^k}=
	\llangle f,\Psiswap \Psir_{\w\circ\la}\rrangle_{\Wc^k}
	=\llangle f,\Psil_{\w\circ\la}\rrangle_{\Wc^k}$
	in the left-hand side. Then \eqref{completeness_big_2} follows
	by applying $\Psiswap^{-1}$
	to both sides of that identity for $\Psiswap f$.

	Expansions \eqref{completeness_small_1} and
	\eqref{completeness_small_2} follow from formula
	\eqref{Pli_and_Psil_large}.
\end{proof}

\begin{remark}
	The spatial Plancherel formula
	in the form \eqref{Plancherel_identity_representation}
	(or \eqref{completeness_small_1} with
	$f(\n)=\mathbf{1}_{\n=\x}$)
	is equivalent to the statement conjectured
	in \cite[\S 4.1]{Povolotsky2013}.
\end{remark}

\begin{remark}\label{rmk:name_for_Plancherel}
 	One can say that expansions
 	\eqref{completeness_big_1}--\eqref{completeness_big_2}
 	correspond to integrating against a (complex-valued)
 	measure which is supported on a disjoint sum of subspaces
 	(or contours and strings of specializations, cf. \eqref{w_circ_la}).
 	Since such an integration corresponds to the inverse transform $\Pli$
 	(see \S \ref{sub:various_forms_of_contour_integration}),
 	this measure may be called the \emph{Plancherel measure}.
	Note that this measure does not depend on $\nu$.
	This object has already appeared as a Plancherel measure in
	the treatment of the $\nu=0$ case in
	\cite{BorodinCorwinPetrovSasamoto2013}.

	This should be compared to other models with Hermitian Hamiltonians
	such as
	the XXZ spin chain \cite{BabbittThomas}, \cite{BabbittGutkin},
	\cite{Gutkin}
	(see also \S \ref{sec:application_to_six_vertex_model} below)
	and the continuous delta Bose gas \cite{Oxford1979},
	\cite{HeckmannOpdam1997},
	where the corresponding Plancherel measures are positive
	on suitably chosen contours.
\end{remark}

The second corollary
is the following \emph{spatial biorthogonality} of eigenfunctions
with respect to the bilinear pairing
$\llangle\cdot,\cdot\rrangle_{\Cc_{z}^k}$ (so that the spatial
variables play the role of labels of these functions).
\begin{corollary}\label{cor:C_biorthogonality}
	For any $\n,\vec m\in\Weyl{k}$, regarding $\Psil(\n)$ and
	$\Psir(\vec m)$ as elements of $\Cc^{k}_{z}$
	(i.e., $\Psil(\n)$ acts as $\z\mapsto\Psil_{\z}(\n)$, and \
	same for $\Psir(\vec m)$), we have
	\begin{align*}
		&\llangle \Psil(\n),\Psir(\vec m)\rrangle_{\Cc_{z}^k}
		\\&\hspace{20pt}=
		\sum_{\la\vdash k}\oint_{\ga_k}\ldots\oint_{\ga_k}
		d\Plm^{(q)}_\la(\w)
		\prod_{j=1}^{\ell(\la)}\frac{1}{(w_j;q)_{\la_j}(\nu w_j;q)_{\la_j}}
		\Psil_{\w\circ\la}(\n)
		\Psir_{\w\circ\la}(\vec m)
		\\&\hspace{20pt}=
		\oint_{\ga}\ldots\oint_{\ga}
		d\Plm^{(q)}_{(1^{k})}(\z)
		\prod_{j=1}^{k}\frac{1}{(1-z_j)(1-\nu z_j)}
		\Psil_{\z}(\n)
		\Psir_{\z}(\vec m)
		\\&\hspace{20pt}=\mathbf{1}_{\vec m=\n}.
	\end{align*}
\end{corollary}
\begin{proof}
	This immediately follows by taking $f(\x)=\mathbf{1}_{\x=\n}$ and
	$g(\x)=\mathbf{1}_{\x=\vec m}$ in \eqref{Plspatial_isomorphism}.
	Note that 
	the contour integral expressions
	for the pairing
	$\llangle \Psil(\n),\Psir(\vec m)\rrangle_{\Cc_{z}^k}$
	come directly from its definition, cf. \S \ref{sub:spectral_bilinear_pairing}. 
\end{proof}

\begin{remark}\label{rmk:biorthogonality_vs_orthogonality}
	Corollary \ref{cor:C_biorthogonality}
	implies that the left and right eigenfunctions
	are biorthogonal as elements
	of the space $\Cc^{k}_{z}$.
	On the other hand, each $\Psir(\vec m)$ is the
	image of $f(\n)=\mathbf{1}_{\n=\vec m}$ (viewed as an element
	of $\Wc^{k}$) under the direct transform $\Pld$
	\eqref{Pld_transform}. The indicator functions
	are orthogonal with respect to the bilinear pairing
	in $\Wc^{k}$. Thus, it is natural
	that to map
	the pairing $\llangle\cdot,\cdot\rrangle_{\Wc^k}$
	to the pairing $\llangle\cdot,\cdot\rrangle_{\Cc_z^k}$
	we need a twisting operator
	$\Psiswap$ as in
	\eqref{Plspatial_isomorphism} and
	\eqref{Plspectral_isomorphism}.
	(One could avoid using the operator $\Psiswap$
	in \eqref{Plspatial_isomorphism} and
	\eqref{Plspectral_isomorphism}
	by taking a twisted bilinear pairing
	$\llangle f,g\rrangle_{\Wc^k}^{\sim}
	=\llangle f,\Psiswap^{-1}g \rrangle_{\Wc^k}$
	in the space $\Wc^{k}$ instead.)
\end{remark}



\section{Spectral biorthogonality of eigenfunctions} 
\label{sec:spectral_biorthogonality_of_eigenfunctions}

In this section we establish a statement
about the
\emph{spectral biorthogonality}
of the left and right eigenfunctions,
when spectral variables are treated as labels of the
eigenfunctions
(Theorem~\ref{thm:spectral_biorthogonality}).
We prove the spectral biorthogonality independently, as
the spectral Plancherel formula
(Theorem \ref{thm:spectral_Plancherel})
was deduced from Theorem \ref{thm:spectral_biorthogonality}.

In this section we work
under an additional assumption
$0<q<1$ (the same assumption was made
in \S \ref{sub:spectral_plancherel_formula}).
We also fix
$0\le \nu<1$,
and
integer $k\ge1$.

\subsection{Contour $\ga'$} 
\label{sub:contour_}

\begin{definition}\label{def:contour_gap}
	Recall the contour $\ga$
	from Definition \ref{def:contours}.
	Let $\ga'$ be a positively oriented contour containing $\ga$,
	contained inside $q^{-1}\ga$,
	and not containing $\nu^{-1}$,
	such that for all
	$z\in\ga$ and all $w\in\ga'$ one has
	\begin{align}\label{contour_gap_condition}
		\left|\frac{1-z}{1-\nu z}\right|<
		\left|\frac{1-w}{1-\nu w}\right|.
	\end{align}
	See Fig.~\ref{fig:gamma_prime} for an example of such a contour.
\end{definition}
\begin{figure}[htbp]
	\begin{center}
	\begin{tabular}{cc}
	\begin{tikzpicture}
		[scale=2.4]
		\def\pt{0.02}
		\def\q{.75}
		\def\nuu{0.58}
		\def\eps{0.3}
		\def\r{.93}
		\draw[->, thick] (-.7,0) -- (2.06,0);
	  	\draw[->, thick] (0,-1.7) -- (0,1.7);
	  	\node at (-.3,1) {$z$};
	  	\draw[fill] (1,0) circle (\pt) node [below] {1};
	  	\draw[fill] (\q,0) circle (\pt) node [below,yshift=-1] {$q$};
	  	\draw[fill] (0,0) circle (\pt) node [below right] {$0$};
	  	\draw[fill] (1/\nuu,0) circle (\pt) node [below, yshift=3,xshift=3] {$\nu^{-1}$};
	  	\draw (-\eps/2+\r/2+.5/\nuu-.5*\r/\nuu,0) circle (\eps/2+\r/2+.5/\nuu-.5*\r/\nuu) node [below,xshift=25,yshift=36] {$\ga$};
	  	\draw[dotted] (-\eps/2/\q+\r/2/\q+.5/\nuu/\q-.5*\r/\nuu/\q,0) circle (\eps/2/\q+\r/2/\q+.5/\nuu/\q-.5*\r/\nuu/\q);
	  	\draw (0.47,0) circle (.813) node [xshift=58,yshift=20] {$\ga'$};
	\end{tikzpicture}
	&
	\begin{tikzpicture}
		[scale=2.4]
		\def\pt{0.02}
		\def\q{.6}
		\def\nuu{0.58}
		\def\eps{0.12}
		\def\r{.6}
	  	\node at (1.1,1.3) {$\xi$};
		\draw[->, thick] (-2.2,0) -- (1.4,0);
	  	\draw[->, thick] (0,-1.7) -- (0,1.7);
	  	\draw[fill] (1,0) circle (\pt) node [below] {1};
	  	\draw[fill] (0,0) circle (\pt) node [below right] {$0$};
	  	\draw (0.49,0) circle (0.615) node [below,xshift=15,yshift=35] {$\xinu(\ga)$};
	  	\draw[dotted] (-0.48,0) circle (1.62);
	  	\draw (0,0) circle (1.12) node [below,xshift=-45,yshift=87] {$\xinu(\ga')$};
	\end{tikzpicture}
	\end{tabular}
	\end{center}
  	\caption{A possible choice of the integration contours $\ga$
  	and $\ga'$ in the $z$ variables (left), and the corresponding contours
  	$\xinu(\ga)$ and $\xinu(\ga')$ in the $\xi$ variables (right).
  	Contours $q^{-1}\ga$ and $\xinu(q^{-1}\ga)$ are shown dotted.}
  	\label{fig:gamma_prime}
\end{figure}
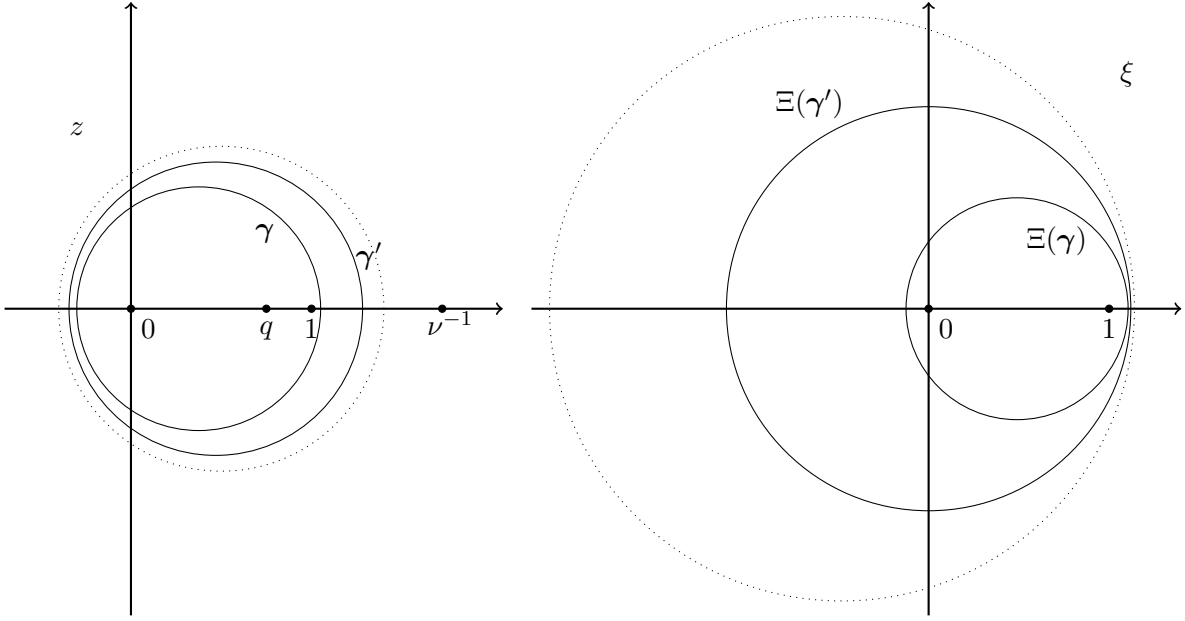

\begin{lemma}
	Contours $\ga$ and $\ga'$ satisfying Definition \ref{def:contour_gap}
	exist.
\end{lemma}
\begin{proof}
	The case $\nu=0$ is simpler (and can be obtained
	in the limit $\nu\searrow 0$), so in the proof we
	assume that $0<\nu<1$.

	Take the contour $\ga$ to be
	the positively oriented closed
	circle which intersects the real line
	at points $-\epsilon$ and $r+\nu^{-1}(1-r)$,
	where $\epsilon>0$ and $0<r<1$ are
	parameters of the contour.
	The map $\xinu$ \eqref{xinu}
	turns $\ga$ into a positively oriented
	closed circle which intersects the
	real line at points
	\begin{align*}
		c_1=-\frac{1-r}{r\nu}
		\qquad
		\text{and}
		\qquad
		c_2=\frac{1+\epsilon}{1+\nu \epsilon}.
	\end{align*}
	Next, the intersection points of the contour
	$\xinu(q^{-1}\ga)$
	with the real axis are
	\begin{align*}
		d_1=\frac{1-q\nu-r(1-\nu)}{(1-q-r(1-\nu))\nu}
		\qquad
		\text{and}
		\qquad
		d_2=
		\frac{q+\epsilon}{q+\nu \epsilon}.
	\end{align*}

	First, assume that $q\le \nu$.
	Then we will construct the contours $\ga$, $\ga'$
	corresponding to $r=1- \epsilon$, and $\epsilon$
 	small enough. Then
 	\begin{align*}
 		&
 		c_1=-\nu^{-1} \epsilon+O(\epsilon^{2}),
 		&c_2=1+(1-\nu)\epsilon+O(\epsilon^{2}),
 		\\&d_1=\frac{1-q}{\nu-q}-q \frac{(1-\nu)^{2}}{(q-\nu)^{2}\nu}\epsilon
 		+O(\epsilon^{2}),
 		&d_2=1+(1-\nu)q^{-1}\epsilon+O(\epsilon^{2}).
 	\end{align*}
 	Thus, $|c_1|<|c_2|<|d_2|<|d_1|$.
 	(If $q=\nu$, then $d_1$ behaves as $const+\epsilon^{-1}$, and the inequalities
 	continue to hold.)
 	Take $\xinu(\ga')$ to be the
 	circle which intersects the real line at points
 	$\pm(c_2+\epsilon^{2})$.
 	This contour does not intersect $\xinu(q^{-1}\ga)$,
 	and $|\xi|<|\om|$ for all $\xi\in\xinu(\ga)$ and $\om\in\xinu(\ga')$.
 	Thus, the contours $\ga$ and $\ga'$ satisfy Definition~\ref{def:contour_gap} because under $\xinu$ condition $|\xi|<|\om|$
 	turns into~\eqref{contour_gap_condition}.

 	For $q>\nu$, we will construct the desired contours with
 	$r=1/(1+\nu)+\epsilon$
 	and $\epsilon>0$ small enough.
 	We have
 	\begin{align*}
 		&
 		c_1=
 		-1+\left(\nu +\nu^{-1}+2\right) \epsilon
 		+O(\epsilon^{2}),
 		&c_2=1+(1-\nu)\epsilon+O(\epsilon^{2}),
 		\\&d_1=\frac{q(1+\nu)-2}{q(1+\nu)-2\nu}+
 		\frac{\left(\nu ^2-1\right)^2 q }{\nu  (q(1+\nu)-2\nu)^2}\epsilon
 		+O(\epsilon^{2}),
 		&d_2=1+(1-\nu)q^{-1}\epsilon+O(\epsilon^{2}).
 	\end{align*}
 	If $q<2\nu/(1+\nu)$, one can check that $d_1>1$, and so
 	$|c_1|<|c_2|<|d_2|<|d_1|$.
 	If $q>2\nu/(1+\nu)$, we similarly have
 	$d_1<-1$, and so $|c_1|<|c_2|<|d_2|<|d_1|$.
 	If $q=2\nu/(1+\nu)$, then $d_1$ behaves as $const
 	-2(1+\nu)^{-1}\epsilon^{-1}$, and so the desired inequalities
 	hold. In all cases,
 	as $\xinu(\ga')$ we take the circle intersecting the
 	real line at points $\pm(c_2+\epsilon^{2})$,
 	and thus we obtain the desired contours.
\end{proof}


\subsection{Formulations of spectral biorthogonality} 
\label{sub:formulations}

We will prove spectral biorthogonality
for test functions from a wider class than
$\Cc_{z}^{k}$.
In particular, this is needed
in the proof of the spectral Plancherel
formula (Theorem \ref{thm:spectral_Plancherel}).

\begin{theorem}\label{thm:spectral_biorthogonality}
	Let the function $F(\z)$ be
	such that for $M$ large enough,
	\begin{align*}
		\Vand(\z)F(\z)\prod_{j=1}^{k}\left(\frac{1-z_j}{1-\nu z_j}\right)^{-M}
	\end{align*}
	is holomorphic in the closed exterior of the contour $\ga$
	(including $\infty$).
	Let $G(\w)$ be such that $\Vand(\w)G(\w)$
	is holomorphic in a neighborhood of the closed region between the contours
	$\ga$ and $\ga'$ (Definition \ref{def:contour_gap}).
	Then
	\begin{align}
		\nonumber
		&\sum_{\n\in\Weyl{k}}
		\left(
		\oint_{\ga}\ldots\oint_{\ga}
		\frac{d\z}{(2\pi\i)^{k}}
		\Psir_{\z}(\n)\Vand(\z)F(\z)\right)
		\left(
		\oint_{\ga}\ldots\oint_{\ga}
		\frac{d\w}{(2\pi\i)^{k}}
		\Psil_{\w}(\n)\Vand(\w)G(\w)
		\right)
		\\&\hspace{17pt}=
		\oint_{\ga}\ldots\oint_{\ga}
		\frac{d\z}{(2\pi\i)^{k}}
		(-1)^{\frac{k(k-1)}{2}}
		\prod_{j=1}^{k}(1-z_j)(1-\nu z_j)
		\prod_{A\ne B}(z_A-qz_B)\sum_{\sigma\in S(k)}
		\sgn(\sigma)
		F(\z)G(\sigma \z).
		\label{spectral_biorthogonality}
	\end{align}
	Here $\sigma\z=(z_{\sigma(1)},\ldots,z_{\sigma(k)})$.
	(See Proposition \ref{prop:convergence_of_biorthogonality}
	below about the convergence of the series in $\n$
	in the left-hand side.)
\end{theorem}

To make certain formulas and arguments below shorter,
let us rewrite Theorem \ref{thm:spectral_biorthogonality} using the
other spectral variables $\xi_j$
(see \S \ref{sub:spatial_and_spectral_variables}).
Changing the variables in $\Psil_{\z}$,
$\Psir_{\z}$, and multiplying them by Vandermonde
determinants,
denote
\begin{align}
	\Phil_{\vxi}(\n)&:=
	\sum_{\sigma\in S(k)}
	\sgn(\sigma)\prod_{1\le B<A\le k}\Sm(\xi_{\sigma(A)},\xi_{\sigma(B)})
	\prod_{j=1}^{k}\xi_{\sigma(j)}^{-n_j};
	\label{Phil_xi}
	\\
	\Phir_{\vxi}(\n)&:=
	\st(\n)\left(\frac{1-q}{1-\nu}\right)^{k}
	\sum_{\sigma\in S(k)}
	\sgn(\sigma)\prod_{1\le B<A\le k}\Sm(\xi_{\sigma(B)},\xi_{\sigma(A)})
	\prod_{j=1}^{k}\xi_{\sigma(j)}^{n_j},
	\label{Phir_xi}
\end{align}
where $\Sm$ is given in \eqref{quadratic_S_matrix}.

\begin{theorem}\label{thm:spectral_biorthogonality_xi}
	For any functions\footnote{Slightly abusing the notation,
	we are using the same letters for
	test functions
	in both \eqref{spectral_biorthogonality} and \eqref{spectral_biorthogonality_xi}.} $F(\vxi)$ and $G(\vom)$ such that
	for $M$ large enough,
	$\Vand(\vxi)F(\vxi)\prod_{j=1}^{k}\xi_j^{-M}$
	is holomorphic in the closed exterior of $\xinu(\ga)$,
	and such that $\Vand(\vom)G(\vom)$
	is holomorphic in a neighborhood of the closed region between
	$\xinu(\ga)$ and
	$\xinu(\ga')$, we have
	\begin{align}
		\begin{array}{ll}
			&\displaystyle\sum_{\n\in\Weyl{k}}
			\left(
			\oint_{\xinu(\ga)}\ldots\oint_{\xinu(\ga)}
			\frac{d\vxi}{(2\pi\i)^{k}}
			\Phir_{\vxi}(\n)F(\vxi)\right)
			\left(
			\oint_{\xinu(\ga)}\ldots\oint_{\xinu(\ga)}
			\frac{d\vom}{(2\pi\i)^{k}}
			\Phil_{\vom}(\n)G(\vom)
			\right)
			\\&\displaystyle\hspace{50pt}=
			\oint_{\xinu(\ga)}\ldots\oint_{\xinu(\ga)}
			\frac{d\vxi}{(2\pi\i)^{k}}
			\prod_{j=1}^{k}\xi_j
			\prod_{1\le A\ne B\le k}\Sm(\xi_A,\xi_B)\sum_{\sigma\in S(k)}
			\sgn(\sigma)
			F(\vxi)G(\sigma \vxi).
		\end{array}
		\label{spectral_biorthogonality_xi}
	\end{align}
	(See Proposition \ref{prop:convergence_of_biorthogonality}
	below about the convergence of the series in $\n$
	in the left-hand side.)

	Identity \eqref{spectral_biorthogonality_xi} is
	equivalent to \eqref{spectral_biorthogonality}.
\end{theorem}

\par\noindent\textit{Proof of equivalence of
Theorems \ref{thm:spectral_biorthogonality} and
\ref{thm:spectral_biorthogonality_xi}.}
	From \eqref{Psil_xi}--\eqref{Psir_xi} we readily have
	\begin{align}\label{spectral_biorthogonality_xi_proof}
		\Vand(\xinu(\vxi))\Vand(\xinu(\vom))\Psir_{\xinu(\vxi)}(\n)
		\Psil_{\xinu(\vom)}(\n)=
		\frac{(-1)^{\frac{k(k+1)}{2}}(1-\nu)^{k}
		\left(1-q\nu\right)^{{k(k-1)}}}{\prod_{j=1}^{k}
		(1-\nu \xi_j)^{k-1}(1-\nu \om_j)^{k-1}}
		\Phir_{\vxi}(\n)
		\Phil_{\vom}(\n).
	\end{align}
	Here we have used
	\begin{align*}
		\frac{\Vand(\xinu(\vxi))}{\Vand(\vxi)}=\frac{(\nu-1)^{\frac{k(k-1)}{2}}}{\prod_{j=1}^{k}
		(1-\nu \xi_j)^{k-1}}.
	\end{align*}
	
	Next, the change of variables $\z=\xinu(\vxi)$,
	$\w=\xinu(\vom)$ in the left-hand side of \eqref{spectral_biorthogonality} yields
	\begin{align}\label{spectral_biorthogonality_xi_proof_1}
		d\z d\w=
		\frac{(\nu-1)^{2k}}{\prod_{j=1}^{k}(1-\nu \xi_j)^{2}(1-\nu\om_j)^{2}}d\vxi d\vom,
	\end{align}
	and in the right-hand side the change of variables gives
	\begin{align}\nonumber
		&(-1)^{\frac{k(k-1)}{2}}
		\prod_{j=1}^{k}(1-z_j)(1-\nu z_j)
		\prod_{A\ne B}(z_A-qz_B)
		d\z
		\\&\hspace{50pt}=
		(-1)^{\frac{k(k-1)}{2}}
		(\nu-1)^{3k}
		(1-q\nu)^{k(k-1)}
		\prod_{j=1}^{k}
		\frac{\xi_j}{(1-\nu \xi_j)^{2k+2}}
		\prod_{A\ne B}\Sm(\xi_A,\xi_B)
		d\vxi.
		\label{spectral_biorthogonality_xi_proof_2}
	\end{align}
	Clearly, this change of variables turns test functions
	in \eqref{spectral_biorthogonality} into test functions in
	\eqref{spectral_biorthogonality_xi}.
	Next, by multiplying $G$ by a suitable power of
	$\prod_{j=1}^{k}(1-\nu \om_j)$ (which preserves its
	necessary properties), one can readily match the
	left-hand side
	(product of
	\eqref{spectral_biorthogonality_xi_proof} and \eqref{spectral_biorthogonality_xi_proof_1})
	with the right-hand side \eqref{spectral_biorthogonality_xi_proof_2}.
\qed

\begin{proposition}\label{prop:convergence_of_biorthogonality}
	Both infinite series in $\n$
	in left-hand sides of \eqref{spectral_biorthogonality}
	and \eqref{spectral_biorthogonality_xi} converge
	for the corresponding choices of test functions $F$ and $G$.
\end{proposition}
\begin{proof}
	Let us illustrate why the statement holds in the $\vxi$ variables
	(the statement in the $\z$ variables is completely analogous).
	
	Assume that $n_k<-M$, where $M$ is a sufficiently large positive integer
	(which depends on the test function $F$).
	Then in the left-hand side of \eqref{spectral_biorthogonality_xi}
	the integral of
	$(\Phir_{\vxi}(\n)\Vand(\vxi)^{-1})\Vand(\vxi)F(\vxi)$
	vanishes. Indeed, this is because
	each summand in
	the Laurent
	polynomial
	$\Phir_{\vxi}(\n)\Vand(\vxi)^{-1}$
	contains a factor of the form $\xi_j^{-m}$ for some $1\le j\le k$ and
	$m$ sufficiently large, and thus the expression
	\begin{align*}
		\xi_j^{-m}\Vand(\vxi)F(\vxi)=
		(\xi_1 \ldots\xi_{j-1}\xi_{j+1}\ldots\xi_k)^m
		(\xi_1 \ldots\xi_k)^{-m}\Vand(\vxi)F(\vxi)
	\end{align*}
	is holomorphic in the variable $\xi_j$ in the exterior of the integration
	contour $\xinu(\ga)$. Performing the integration in $\xi_j$
	first, we see that
	the integral of
	this summand (coming from $\Phir_{\vxi}(\n)\Vand(\vxi)^{-1}$) vanishes.

	This implies that it is possible to perform the summation over $\n$
	in the left-hand side of \eqref{spectral_biorthogonality_xi}
	such that $n_k\ge -M$.
	Next, let us deform the $\vom$
	contours from $\xinu(\ga)$
	to $\xinu(\ga')$. Because of the properties
	of $G(\vom)$, this will not change the
	integral over $\vom$
	in the left-hand side of \eqref{spectral_biorthogonality_xi}.
	Now we can interchange the summation over $\n$
	with the integration over $\vxi$ and $\vom$
	because the sum over $\n$
	inside the integral will converge
	uniformly, thanks to \eqref{contour_gap_condition}. Therefore, the sum
	converged before the interchange,
	too.
	This implies the desired claim.
\end{proof}

\begin{remark}\label{rmk:spectral_biorthogonality_with_deltas}
	The identity \eqref{spectral_biorthogonality}
	can be formally rewritten as
	\begin{align}
		\llangle
		\Psir_{\z},\Psil_{\w}
		\rrangle_{\Wc^k}
		\Vand(\z)\Vand(\w)=
		(-1)^{\frac{k(k-1)}{2}}
		\prod_{j=1}^{k}(1-z_j)(1-\nu z_j)
		\prod_{A\ne B}(z_A-qz_B)
		\det[\delta({z_i-w_j})]_{i,j=1}^{k},
		\label{spectral_biorthogonality_deltas}
	\end{align}
	where $\delta(z)$ is the Dirac delta
	function supported at $z=0$.
	Similarly, \eqref{spectral_biorthogonality_xi}
	corresponds to the following formal identity:
	\begin{align}\label{spectral_biorthogonality_xi_deltas}
		\llangle
		\Phir_{\vxi},\Phil_{\vom}
		\rrangle_{\Wc^k}=
		\prod_{j=1}^{k}\xi_j
		\prod_{A\ne B}\Sm(\xi_A,\xi_B)
		\det[\delta({\xi_i-\om_j})]_{i,j=1}^{k}.
	\end{align}
\end{remark}

\begin{remark}\label{rmk:spectral_biorthogonality_over_1_nuin}
	If one restricts the class of test functions
	in
	\eqref{spectral_biorthogonality}
	to (not necessarily symmetric) Laurent polynomials in variables $\frac{1-z_j}{1-\nu z_j}$
	and $\frac{1-w_j}{1-\nu w_j}$, respectively,
	then the only singularities of the integrands in
	both sides would be $1$ and $\nu^{-1}$.
	Thus, the contour $\ga$ of integration
	in \eqref{spectral_biorthogonality}
	can be changed to either a positively
	oriented small circle around $1$
	(we denoted this contour by $\ga_k$, see Definition
	\ref{def:contours}),
	or to a negatively oriented small circle around~$\nu^{-1}$.

	Likewise, if one restricts \eqref{spectral_biorthogonality_xi}
	to (not necessarily symmetric)
	Laurent polynomials $\xi_1,\ldots,\xi_k$
	and $\om_1,\ldots,\om_k$,
	then the integration contours $\xinu(\ga)$
	can be replaced by small positively
	oriented circles around~$0$.
\end{remark}

We prove Theorem
\ref{thm:spectral_biorthogonality} (in the equivalent form of
Theorem \ref{thm:spectral_biorthogonality_xi}) in \S \ref{sub:proof_of_the_spectral_orthogonality} below.


\subsection{Proof of the spectral biorthogonality} 
\label{sub:proof_of_the_spectral_orthogonality}

Denote the left-hand side of \eqref{spectral_biorthogonality_xi}
by $\LHS{F(\vxi)G(\vom)}$. It is helpful to regard
$\LHS{\cdots}$ as a linear functional acting on functions
in $\vxi$ and $\vom$ (in other words, a generalized function,
or a distribution),
cf. the notation in \eqref{spectral_biorthogonality_deltas}--\eqref{spectral_biorthogonality_xi_deltas}.

\begin{lemma}\label{lemma:spectral_biorthogonality_symmetric_poly}
	For any symmetric polynomial $p(\xi_1,\ldots,\xi_k)$,
	\begin{align*}
		\LHS{p(\vxi)F(\vxi)G(\vom)}=
		\LHS{F(\vxi)p(\vom)G(\vom)}.
	\end{align*}
\end{lemma}
\begin{proof}
	We will use the fact that the functions $\Phir_{\vxi}(\n)$
	and $\Phil_{\vom}(\n)$ are eigenfunctions
	of the $q$-Hahn operator
	$\Hbwd$ with an \emph{additional free parameter} $\mu$.
	Their eigenvalues are (see \S \ref{sec:the_q_mu_nu_boson_process_and_coordinate_bethe_ansatz} below; note that here
	we use the spectral variables $\vxi$)
	\begin{align*}
		\ev(\xinu(\vxi))=\prod_{j=1}^{k}
		\frac{1-\mu-\xi_j(\nu-\mu)}{1-\nu},
	\end{align*}
	because  $\Phir_{\vxi}(\n)$
	and $\Phil_{\vom}(\n)$ differ from
	$\Psir_{\z}(\n),\Psil_{\w}(\n)$ by a change of variables and
	by multiplicative factors not depending on the spatial
	variables $\n$.

	Using the notation
	with Dirac delta functions
	\eqref{spectral_biorthogonality_xi_deltas}, we have
	\begin{align}
		\ev(\xinu(\vom))\llangle
		\Phir_{\vxi},\Phil_{\vom}
		\rrangle_{\Wc^k}
		=
		\llangle
		\Phir_{\vxi},\Hbwd\Phil_{\vom}
		\rrangle_{\Wc^k}
		=\llangle
		(\Hbwd)^{\text{transpose}}\Phir_{\vxi},\Phil_{\vom}
		\rrangle_{\Wc^k}
		=
		\ev(\xinu(\vxi))\llangle
		\Phir_{\vxi},\Phil_{\vom}
		\rrangle_{\Wc^k}.
		\label{spectral_biorthogonality_symmetric_poly_proof}
	\end{align}
	This is readily extended to a rigorous formula
	like in Theorem \ref{thm:spectral_biorthogonality_xi}
	because
	with our test functions
	$F$ and $G$ the corresponding sum over $\n$ in
	\eqref{spectral_biorthogonality_xi} converges
	(see Proposition \ref{prop:convergence_of_biorthogonality}).
	
	The eigenvalues depend on an additional free parameter $\mu$.
	Extracting coefficients
	of various powers of $\mu$
	in
	\eqref{spectral_biorthogonality_symmetric_poly_proof}, we see that the statement of the lemma holds when $p$
	is an elementary symmetric polynomial
	\begin{align*}
		p(\xi_1,\ldots,\xi_k)=e_m(\xi_1,\ldots,\xi_k)
		=\sum_{i_1<\ldots<i_m}\xi_{i_1}\ldots\xi_{i_m},\qquad
		m=1,\ldots,k
		.
	\end{align*}
	Writing \eqref{spectral_biorthogonality_symmetric_poly_proof}
	for products of eigenvalues with different free parameters
	$\mu_i$,
	we get the same statement for any products of elementary symmetric polynomials.
	Finally, an arbitrary symmetric polynomial is a linear combination
	of such products, cf. \cite[I.2]{Macdonald1995}.
\end{proof}

\begin{lemma}\label{lemma:spectral_biorthogonality_zero}
	For any $j=1,\ldots,k$,
	\begin{align*}
		\LHS{(\om_j-\xi_1)\ldots(\om_j-\xi_k)F(\vxi)G(\vom)}=0,
		\qquad
		\LHS{F(\vxi)(\xi_j-\om_1)\ldots(\xi_j-\om_k)G(\vom)}=0.
	\end{align*}
\end{lemma}
\begin{proof}
	Apply Lemma \ref{lemma:spectral_biorthogonality_symmetric_poly}
	with the symmetric polynomial
	$(\om_j-\xi_1)\ldots(\om_j-\xi_k)$
	in $\vxi$ (here $\om_j$ is viewed as a parameter
	of the polynomial). Under $\LHS{\cdots}$,
	this polynomial in $\vxi$ can be replaced by
	the polynomial
	$(\om_j-\om_1)\ldots(\om_j-\om_k)\equiv0$
	in $\vom$,
	which proves the first claim.
	The second claim is analogous.
\end{proof}

\begin{lemma}\label{lemma:power_shifting}
	For any test functions $F(\vxi)$, $G(\vom)$
	as in Theorem \ref{thm:spectral_biorthogonality_xi},
	integers $1\le j\le k$ and $N\ge1$,
	there exists $N_1\ge N$ such that
	\begin{align*}
		\LHS{F(\vxi)G(\vom)}
		=
		\sum_{m= N}^{N_1}\LHS{F^{(m)}(\vxi)G(\vom)\om_j^{m}},
	\end{align*}
	where $F^{(m)}(\vxi)$ are finite linear combinations
	of functions of the form $F(\vxi)\xi_1^{-i_1}\ldots\xi_k^{-i_k}$,
	where $i_1,\ldots,i_k\ge0$.
\end{lemma}
\begin{proof}
	Lemma \ref{lemma:spectral_biorthogonality_zero}
	written as
	\begin{align*}
		\LHS{
		(1-\om_j/\xi_1)\ldots(1-\om_j/\xi_k)F(\vxi)G(\vom)}=0
	\end{align*}
	allows to express
	$\LHS{F(\vxi)G(\vom)}$
	as a linear combination of
	quantities of the form
	\begin{align*}
		\LHS{F(\vxi)\xi_1^{-i_1}\ldots\xi_k^{-i_k}G(\vom)\om_j^{m}},
	\end{align*}
	with $m\ge 1$ and $i_1,\ldots,i_k\ge0$.
	A repeated application of this procedure allows to
	shift the power of $\om_j$
	past any $N\ge1$.
\end{proof}

By Runge's theorem, it is possible to
approximate
any function $G(\vom)$
satisfying the condition of
Theorem \ref{thm:spectral_biorthogonality_xi}
by
finite linear combinations of powers
$\om_1^{x_1}\ldots \om_k^{x_k}$
(uniformly in $\om_1,\ldots,\om_k$ belonging
to the closed region between
the contours $\xinu(\ga)$ and $\xinu(\ga')$).
Thus,
we may
prove the spectral biorthogonality by computing
$\LHS{F(\vxi)\om_1^{x_1}\ldots \om_k^{x_k}}$.
Moreover, since $\Phil_{\vom}$ involved in
\eqref{spectral_biorthogonality_xi} is
skew-symmetric in $\om_1,\ldots,\om_k$, we may assume that $x_1\ge \ldots\ge x_k$.
Finally, by using Lemma \ref{lemma:power_shifting}
we may shift powers of $\om_k,\om_{k-1},\ldots,\om_1$
(in this order) so that
all gaps between $x_i$ and $x_{i+1}$ ($i=1,\ldots,k-1$)
are arbitrarily large.
Thus, it remains to compute the left-hand side of \eqref{spectral_biorthogonality_xi}
in the following case:
\begin{lemma}
	Let $x_1>x_2>\ldots>x_k$, and $x_i-x_{i+1}>2k$ for all $i=1,\ldots,k-1$.
	Then the identity \eqref{spectral_biorthogonality_xi} holds for
	arbitrary $F(\vxi)$
	(satisfying the condition of Theorem \ref{thm:spectral_biorthogonality_xi}) and for $G(\vom)=\om_1^{x_1}\ldots \om_k^{x_k}$.
\end{lemma}
\begin{proof}
	First, note that
	for this choice of $G$,
	the sum over $\n$ in the left-hand of \eqref{spectral_biorthogonality_xi}
	is finite.
	We have
	\begin{align*}
		&\sum_{\n\in \Weyl{k}}
		\oint_{\xinu(\ga)}\ldots\oint_{\xinu(\ga)}
		\frac{d\vom}{(2\pi\i)^{k}}
		\om_1^{x_1}\ldots \om_k^{x_k}
		\Phir_{\vxi}(\n)\Phil_{\vom}(\n)
		\\&\hspace{50pt}=
		\sum_{\sigma\in S(k)}\sgn(\sigma)\sum_{\omega\in S(k)}
		\sgn(\omega)
		\sum_{\n\in \Weyl{k}}
		\st(\n)\left(\frac{1-q}{1-\nu}\right)^{k}
		\oint_{\xinu(\ga)}\ldots\oint_{\xinu(\ga)}
		\frac{d\vom}{(2\pi\i)^{k}}
		\\
		&\hspace{190pt}
		\times
		\prod_{B<A}\Sm(\om_{\sigma(A)},\om_{\sigma(B)})
		\Sm(\xi_{\omega(B)},\xi_{\omega(A)})
		\prod_{j=1}^{k}\om_{j}^{-n_{\sigma^{-1}(j)}+x_j}
		\xi_{\omega(j)}^{n_j}.
	\end{align*}
	The product $\prod_{B<A}\Sm(\om_{\sigma(A)},\om_{\sigma(B)})$
	is a polynomial in each of the variables $\om_j$
	of degree $k-1$. In order for the above integral in $\om_j$'s
	not to vanish, all powers $x_j-n_{\sigma^{-1}(j)}$, $j=1,\ldots,k$,
	must range from $-k$ to $-1$. Because of the large gaps between the $x_j$'s
	and the inequalities $n_1\ge \ldots\ge n_k$,
	this implies that $\sigma$ must be the identity permutation
	(i.e., the contribution of all other permutations is zero).
	This also implies that the $n_j$'s must be distinct.
	For distinct $n_j$'s,
	$\st(\n)\left(\frac{1-q}{1-\nu}\right)^{k}=1$
	(see \eqref{stationary_measure}).

	After restricting to $\sigma=id$,
	the conditions $n_1\ge n_2\ge \ldots\ge n_k$
	will be automatic (otherwise the integral vanishes),
	so we will be able to drop the assumption
	$\n\in\Weyl{k}$ and sum over all $\n\in\Z^{k}$.
	Therefore, continuing the above formula, we have
	\begin{align*}
		&
		=\sum_{\omega\in S(k)}
		\sgn(\omega)
		\prod_{B<A}\Sm(\xi_{\omega(B)},\xi_{\omega(A)})
		\\&
		\hspace{120pt}\times
		\sum_{\n\in \Z^{k}}
		\prod_{j=1}^{k}\xi_{\omega(j)}^{n_j}
		\oint_{\xinu(\ga)}\ldots\oint_{\xinu(\ga)}
		\frac{d\vom}{(2\pi\i)^{k}}
		\prod_{B<A}\Sm(\om_{A},\om_{B})
		\prod_{j=1}^{k}\om_{j}^{-n_{j}+x_j}
		.
	\end{align*}
	The integration over $\vom$
	may be interpreted as taking the coefficient
	of
	$\om_1^{n_1}\ldots\om_k^{n_k}$
	in the product
	\begin{align}\label{spectral_biorthogonality_proof_good}
		\om_1^{x_1+1}\ldots \om_k^{x_k+1}\prod_{B<A}\Sm(\om_{A},\om_{B}).
	\end{align}
	Then each such coefficient
	of $\om_1^{n_1}\ldots\om_k^{n_k}$
	is
	multiplied by
	$\xi_{\omega(1)}^{n_1}\ldots \xi_{\omega(k)}^{n_k}$,
	and the summation over all possible powers $\n\in\Z^{k}$
	is performed. Therefore, the result of this summation is a substitution
	$\om_j\to\xi_{\omega(j)}$ in \eqref{spectral_biorthogonality_proof_good}.
	Therefore, we have
	\begin{align*}
		&\sum_{\n\in \Weyl{k}}
		\oint_{\xinu(\ga)}\ldots\oint_{\xinu(\ga)}
		\frac{d\vom}{(2\pi\i)^{k}}
		\om_1^{x_1}\ldots \om_k^{x_k}
		\Phir_{\vxi}(\n)\Phil_{\vom}(\n)
		\\&
		\hspace{80pt}
		=
		\sum_{\omega\in S(k)}
		\sgn(\omega)
		\prod_{B<A}\Sm(\xi_{\omega(B)},\xi_{\omega(A)})
		\prod_{B<A}\Sm(\xi_{\omega(A)},\xi_{\omega(B)})
		\,\xi_{\omega(1)}^{x_1+1}\ldots
		\xi_{\omega(k)}^{x_k+1}.
	\end{align*}
	After multiplying both sides
	of the above identity
	by
	$F(\vxi)$ and integrating over the variables $\xi_j$,
	we readily get \eqref{spectral_biorthogonality_xi}.
\end{proof}
This completes the proof of the spectral biorthogonality (Theorems
\ref{thm:spectral_biorthogonality} and \ref{thm:spectral_biorthogonality_xi}).

\begin{remark}\label{rmk:difference_with_BCPS_spectral}
	The sum over $\n$ in the spectral biorthogonality
	statement \eqref{spectral_biorthogonality_xi_deltas}
	can be written as
	\begin{align}\label{spectral_biorthogonality_xi_remark}
		\sum_{\n\in\Weyl{k}}
		C(\n)\Phil_{\vom}(\n)(\Refl\Phil_{\vxi})(\n),
	\end{align}
	where $\Refl$ is the space reflection operator
	\eqref{reflection_operator}, and
	$C(\n)$ is a certain constant not depending on $\n$
	times $\st(\n)$ \eqref{stationary_measure}.

	It may seem strange that
	our proof of the spectral biorthogonality
	(in other words, the computation of the sum
	\eqref{spectral_biorthogonality_xi_remark})
	does not involve the explicit value of the constant
	$C(\n)$ (except in the case when all $n_j$'s are distinct).
	On the other hand,
	the proof heavily relies on
	Lemma \ref{lemma:spectral_biorthogonality_zero},
	a statement made possible by the presence of a free parameter
	$\mu$ in the operator and not in the eigenfunctions.
	
	It seems plausible (and was checked for $k=2$)
	that
	Lemma \ref{lemma:spectral_biorthogonality_zero} \emph{alone}
	determines the value of the constant $C(\n)$
	for all $\n$'s (up to an overall factor not depending on $\n$).

	Note also that the proof
	of the spectral biorthogonality for $\nu=0$ given in
	\cite[\S 6]{BorodinCorwinPetrovSasamoto2013}
	is very different, and it
	employs the explicit value of the corresponding constant $C(\n)$.
\end{remark}


\subsection{Degenerate spectral biorthogonality} 
\label{sub:more_degenerate_spectral_biorthogonality}

A spectral biorthogonality statement
similar to Theorems \ref{thm:spectral_biorthogonality}
and \ref{thm:spectral_biorthogonality_xi}
should hold for
all eigenfunctions which
appear in completeness results \eqref{completeness_big_1}--\eqref{completeness_big_2}
of Corollary \ref{cor:completeness}.
Namely, these eigenfunctions
are indexed by spectral variables which form
strings
$\z=\z'\circ\la$, $\w=\w'\circ\varkappa$,
where $\la$ and $\varkappa$ are two partitions of $k$
(recall the definition of strings
\eqref{w_circ_la}).\footnote{Note that if $\la\ne(1^{k})$, then $\z=\z'\circ\la$
depends on (independent)
spectral variables $z'_1,\ldots,z'_{\ell(\la)}$
whose number is strictly less than $k$.}
We conjecture that for such $\z$ and $\w$,
Theorem \ref{thm:spectral_biorthogonality}
should reduce to
\begin{align}
	\label{degenerate_biorthogonality_deltas}
	\llangle
	\Psir_{\z},\Psil_{\w}
	\rrangle_{\Wc^k}
	\Vand(\z)\Vand(\w)=
	\mathbf{1}_{\la=\varkappa}\cdot
	(-1)^{\frac{k(k-1)}{2}}
	\prod_{j=1}^{k}(1-z_j)(1-\nu z_j)
	\prod_{A\ne B}^{\sim}(z_A-qz_B)
	\det[\delta({z_i-w_j})]_{i,j=1}^{k},
\end{align}
where $\prod\limits_{A\ne B}^{\sim}(z_A-qz_B)$
means that we omit factors which are identically
zero
by the very definition of $\z$.
Formal identity \eqref{degenerate_biorthogonality_deltas}
should be understood in an integral sense similar to Theorem
\ref{thm:spectral_biorthogonality}. (However, it is not clear which
test functions or integration contours should be used.)
The change of variables $\xinu$ (see \S \ref{sub:spatial_and_spectral_variables})
would turn \eqref{degenerate_biorthogonality_deltas}
into a conjectural degenerate version of Theorem \ref{thm:spectral_biorthogonality_xi}
(i.e., into a statement in the other spectral variables $\vxi$ and $\vom$).

Using ideas similar to the proof of Theorem \ref{thm:spectral_biorthogonality_xi} in
\S \ref{sub:proof_of_the_spectral_orthogonality}, one could easily show that if $\la\ne \varkappa$,
the left-hand side of \eqref{degenerate_biorthogonality_deltas} vanishes.
We have checked \eqref{degenerate_biorthogonality_deltas}
for $k=2$ and $\la=\varkappa=(2)$, which yields
\begin{align*}
	&\sum_{n_1\ge n_2}
	\Psir_{z', qz'}(n_1,n_2)
	\Psil_{w', qw'}(n_1,n_2)
	(z'-qz')(w'-qw')
	\\&\hspace{150pt}=
	-
	(1-z')(1-qz')
	(1-\nu z')(1-\nu q z')
	(z'-q^{2}z')\delta({z'-w'}).
\end{align*}
Note that in this case we have simplified
the determinant in the right-hand side of \eqref{degenerate_biorthogonality_deltas}
as
$$\det[\delta({z_i-w_j})]_{i,j=1}^{2}
=\delta({z'-w'})\delta({qz'- qw'})-\delta({z'- qw'})\delta({qz'-w'})
=\delta({z'-w'}),$$ because the second summand always vanishes.
This formal identity for $k=2$ was checked on
test functions which are
Laurent polynomials in $\xinu(z')$ and $\xinu(w')$,
and the integration contours in $z'$ and $w'$
were small positively oriented circles around $1$
(cf. Remark \ref{rmk:spectral_biorthogonality_over_1_nuin}).

\smallskip

We note that in the physics literature such degenerate
orthogonality has appeared in the context of
continuous delta Bose gas
\cite{CalabreseCaux2007},
\cite{DotsenkoTW_onedim}.
See also
appendices A and B and references in section II
of
\cite{Dotsenko2010_universal}. In the case of attractive delta Bose gas,
one needs to use all eigenfunctions corresponding to all various
strings (indexed by partitions $\la$) for the completeness of the
Bethe ansatz. In our situation
due to the other completeness results \eqref{completeness_small_1}--\eqref{completeness_small_2},
it is possible to work only with eigenfunctions
corresponding to $\la=(1^{k})$, whose spectral orthogonality
is Theorems \ref{thm:spectral_biorthogonality} and \ref{thm:spectral_biorthogonality_xi}.




\section{The $q$-Hahn system and coordinate Bethe ansatz} 
\label{sec:the_q_mu_nu_boson_process_and_coordinate_bethe_ansatz}

In this section we discuss connections of
the eigenfunctions to the $q$-Hahn stochastic particle
system
(about the name see \S \ref{sub:hopping_distribution_and_q_hahn_orthogonality_weights} below)
introduced by
Povolotsky
\cite{Povolotsky2013}.
Namely,
we briefly recall the coordinate Bethe ansatz computations
leading to the eigenfunctions described in \S \ref{sub:eigenfunctions}.
Our Plancherel formulas in \S \ref{sub:spatial_plancherel_formula}
and \S \ref{sub:spectral_plancherel_formula}
give rise to moment formulas for the $q$-Hahn TASEP
(also introduced in \cite{Povolotsky2013})
with an arbitrary initial condition (\S \ref{sub:solving_backward_and_forward_evolution_equations}),
and also to
certain new symmetrization identities
(\S \ref{sub:symmetrization_formula}).
In this section we will sometimes call the $q$-Hahn particle system
the \emph{$q$-Hahn zero-range process}
to distinguish it from the $q$-Hahn TASEP.

Definitions, constructions and results in this
section depend on
parameters $0<q<1$ and
$0\le \nu\le \mu<1$.

\subsection{The $q$-Hahn system} 
\label{sub:_q_hahn_boson_process}

The ($k$-particle) $q$-Hahn system's state space is the set
$\Weyl{k}$ \eqref{Weyl_k} of spatial variables $\n=(n_1\ge \ldots\ge n_k)$.
Here each $n_j$ is the location of the $j$-th rightmost particle on $\Z$.
We will need another encoding of such $k$-particle configurations.
Namely, let $y_i\ge0$ denote the number of particles at the site $i\in\Z$.
If $\n\in\Weyl{k}$, then for $\y=\y(\n)=(y_i)_{i\in\Z}$
we have $\sum_{i\in\Z}y_i=k$, and only finitely many of the
coordinates of $\y$ are nonzero. Denote by $\Yb{k}$ the space of all such vectors $\y$.
We thus have a bijection between $\Yb{k}$ and $\Weyl{k}$.

For $s_i\in\{0,1,\ldots,y_i\}$, denote
\begin{align*}
	&
	\y_{i,i-1}^{s_i}:=(y_0,y_1,\ldots,y_{i-1}+s_i,y_{i}-s_i,y_{i+1},\ldots,y_N).
	\\&
	\y_{i,i+1}^{s_i}:=(y_0,y_1,\ldots,y_{i-1},y_{i}-s_i,y_{i+1}+s_i,\ldots,y_N).
\end{align*}
That is, the configuration
$\y_{i,i-1}^{s_i}$ is obtained by taking
$s_i$ particles from the site $i$ and moving them to
the site $i-1$ (in $\y_{i,i+1}^{s_i}$, $s_i$
particles are moved to the site $i+1$).

We will also use the following notation for $\n\in\Weyl{k}$:
For $I\subseteq \{1,\ldots, k\}$, let $\n^{\pm}_{I}$ denote the the vector $\n$ with $n_i$ replaced by $n_i\pm1$ for all $i\in I$. The modified vectors $\n^{\pm}_{I}$
do not necessarily belong to $\Weyl{k}$.

\begin{definition}[A deformed Binomial distribution]
\label{def:binomial}
	For $|q|<1$, $0\le \nu\le \mu<1$
	and integers $0\le j\le m$, define
	\begin{align*}
		\phi_{q,\mu,\nu}(j\mid m):=\mu^{j}\frac{(\nu/\mu;q)_{j}(\mu;q)_{m-j}}{(\nu;q)_{m}}\frac{(q;q)_{m}}{(q;q)_{j}(q;q)_{m-j}}.
	\end{align*}
	This definition can be also extended to include the case $m=+\infty$:
	\begin{align*}
		\phi_{q,\mu,\nu}(j\mid +\infty):=
		\mu^{j}\frac{(\nu/\mu;q)_{j}}{(q;q)_j}
		\frac{(\mu;q)_{\infty}}{(\nu;q)_{\infty}}.
	\end{align*}
\end{definition}

It was shown in \cite{Povolotsky2013} (see also \cite{Corwin2014qmunu},
\cite{Barraquand_qhahn_2014} and \cite[(3.6.2)]{Koekoek1996},
cf. \S \ref{sub:hopping_distribution_and_q_hahn_orthogonality_weights} below)
that
for $m=\{0,1,\ldots\}\cup\{+\infty\}$,
\begin{align*}
	\sum_{j=0}^{m}\phi_{q,\mu,\nu}(j\mid m)=1.
\end{align*}
Clearly, for the parameters $(q,\mu,\nu)$ satisfying
conditions stated at the beginning of this section, we have
$\phi_{q,\mu,\nu}(j\mid m)\ge0$ for all meaningful $j$ and $m$.

\begin{definition}[The $q$-Hahn system \cite{Povolotsky2013}]\label{def:q_hahn_process}
	The $k$-particle $q$-Hahn stochastic particle system
	on the lattice $\Z$
	is a discrete-time Markov chain on
	$\Yb{k}$ (equivalently, on $\Weyl{k}$) defined as follows.
	At each time step $t\to t+1$
	and independently at each site $i\in \Z$,
	select $s_i\in\{0,1,\ldots,y_i\}$
	particles at random with probability $\phi_{q,\mu,\nu}(s_i\mid y_i)$.
	After this selection is made at every $i\in\Z$
	(note that $s_i=y_i=0$ for
	all except a finite number of $i$),
	transfer $s_i$ particles from each side $i$ to
	its left neighbor $i-1$. See Fig.~\ref{fig:qBoson}, left.
\end{definition}

The (backward) Markov transition operator of the $q$-Hahn system can be written as follows.
Define for each $i\in\Z$ the one-site transition operator as
\begin{align*}
	\left(\big[\Abwd\big]_if\right)(\y):=
	\sum_{s_i=0}^{y_i}\phi_{q,\mu,\nu}(s_i\mid y_i)f(\y^{s_i}_{i,i-1}).
\end{align*}
Then the global transition operator
is
\begin{align*}
	\left(\Hbwd f\right)(\y):=\ldots\big[\Abwd\big]_{-1}\big[\Abwd\big]_{0}
	\big[\Abwd\big]_{1}\big[\Abwd\big]_{2}\ldots f(\y).
\end{align*}
The operators $\big[\Abwd\big]_i$ and $\Hbwd$
act on functions on
$\Yb{k}$
which correspond to compactly supported functions
on $\Weyl{k}$ (via the bijection between $\Yb{k}$ and $\Weyl{k}$).
Recall that the latter functions constitute the space $\Wc^{k}$,
cf. \S \ref{sub:spatial_and_spectral_variables}.
For such functions $f$, the action of
the operators $\big[\Abwd\big]_i$ is trivial for all but a finite number of $i\in\Z$.

\begin{remark}
	The order of operators $\big[\Abwd\big]_i$
	in the definition of $\Hbwd$ above
	should be understood according to
	\begin{align*}
		\big[\Abwd\big]_{i}\big[\Abwd\big]_{i+1} f(\y)
		=
		\Big(\big[\Abwd\big]_{i}\big(\big[\Abwd\big]_{i+1} f\big)\Big)(\y).
	\end{align*}
	Note that
	the order of the
	$\big[\Abwd\big]_i$'s
	matters because the one-site operators do not act independently.
\end{remark}


\subsection{Particle hopping distribution and $q$-Hahn orthogonality weights} 
\label{sub:hopping_distribution_and_q_hahn_orthogonality_weights}

The particle hopping distribution
$\phi_{q,\mu,\nu}$
(\S \ref{sub:_q_hahn_boson_process})
is related to the orthogonality weight for the
$q$-Hahn orthogonal polynomials
(about the latter objects, e.g., see \cite[\S3.6]{Koekoek1996} and references therein).
The orthogonality weight for the $q$-Hahn polynomials is defined
as
\begin{align*}
	\qhwt_{q,\al,\be}(x\mid N)=(\al\be q)^{-x}\frac{(\al q;q)_{x}(q^{-N};q)_{x}}
	{(q;q)_x(\be^{-1}q^{-N};q)_{x}},
	\qquad
	x=0,1,\ldots,N.
\end{align*}
These $q$-Hahn weights
sum to
\begin{align*}
	\sum_{x=0}^{N}\qhwt_{q,\al,\be}(x\mid N)=(\al q)^{-N}\frac{(q^{2}\al\be;q)_{N}}{(q\be;q)_{N}}.
\end{align*}
Indeed, this follows from the orthogonality relation
for the $q$-Hahn polynomials \cite[(3.6.2)]{Koekoek1996} for $m=n=0$.

One can readily write down the following
relation between $\phi_{q,\mu,\nu}$ (Definition \ref{def:binomial})
and the $q$-Hahn weight function:
\begin{align}\label{phi_qhahn}
	\phi_{q,\mu,\nu}(j\mid m)=
	\frac{\mu^{m}}{\nu^{m-j}}
	\frac
	{(\nu/\mu;q)_{m}}{(\nu;q)_{m}}
	\qhwt_{q^{-1}, q\mu^{-1}, q\mu\nu^{-1}}(m-j\mid m),\qquad
	j=0,1,\ldots,m.
\end{align}

The above understanding
that the
two-parameter
family of particle
hopping distributions $\phi_{q,\mu,\nu}$
(depending on $\mu,\nu$)
is essentially equivalent to the two-parameter
family of orthogonality
weights
$\qhwt_{q,\al,\be}$
(depending on $\al,\be$)
for the $q$-Hahn orthogonal
polynomials
is the reason for the name ``$q$-Hahn''
for the stochastic particle system
of Definition \ref{def:q_hahn_process}.

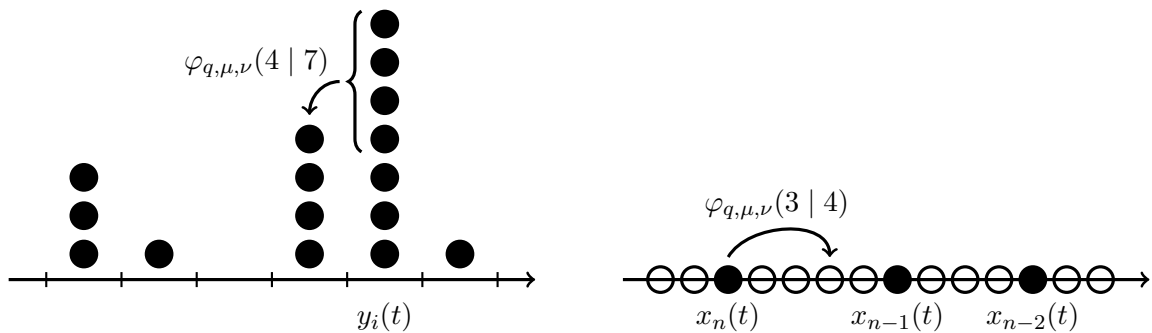
\begin{figure}[htbp]
	\begin{center}
	\begin{tabular}{cc}
		\begin{tikzpicture}
			[scale=1,very thick]
			\def\pt{.17}
			\def\ee{.1}
			\def\h{.34}
			\draw[->] (-.5,0) -- (6.5,0);
			\foreach \ii in {0,1,2,3,4,5,6}
			{
				\draw[thick] (\ii,-\ee) -- (\ii,\ee);
			}
			\foreach \ii in {(.5,\h),(.5,2.5*\h), (.5,4*\h),(1.5,\h),(3.5,\h),(3.5,2.5*\h),(3.5,4*\h), (3.5,5.5*\h), (4.5,\h), (4.5,2.5*\h),(4.5,4*\h), (4.5,5.5*\h), (4.5,7*\h),(4.5,8.5*\h),(4.5,10*\h), (5.5,\h)}
			{
				\draw[fill] \ii circle(\pt);
			}
			\node at (4.5,-5*\ee) {$y_i(t)$};
			\draw [decorate,decoration={brace,amplitude=5pt,raise=4pt},xshift=-19pt] (5,+5*\h) -- (5,10*\h+\pt);
		    \draw[->, very thick] (3.9,7.5*\h+\pt/2) to [in=90, out=-180] (3.5,2.19) node [xshift=-20,yshift=20] {$\phi_{q,\mu,\nu}(4\mid 7)$};
		\end{tikzpicture}
		&\qquad
		\begin{tikzpicture}
			[scale=1,very thick]
			\def\pt{.17}
			\def\ee{.1}
			\def\h{.45}
			\draw[->] (-.5,0) -- (6.5,0);
			\foreach \ii in {(0,0),(\h,0),(3*\h,0),(4*\h,0),(5*\h,0),(6*\h,0) ,(8*\h,0),(10*\h,0),(9*\h,0),(12*\h,0),(13*\h,0)}
			{
				\draw \ii circle(\pt);
			}
			\foreach \ii in {(2*\h,0),(7*\h,0),(11*\h,0)}
			{
				\draw[fill] \ii circle(\pt);
			}
			\node at (2*\h,-5*\ee) {$x_n(t)$};
			\node at (7*\h,-5*\ee) {$x_{n-1}(t)$};
			\node at (11*\h,-5*\ee) {$x_{n-2}(t)$};
		    \draw[->, very thick] (2*\h,.3) to [in=90, out=60] (5*\h,.3)
		    node [xshift=-20,yshift=20] {$\phi_{q,\mu,\nu}(3\mid 4)$};
		\end{tikzpicture}
	\end{tabular}
	\end{center}
  	\caption{A possible transition
  	at one site of the $q$-Hahn zero-range process (left) and the $q$-Hahn TASEP (right)
  	processes.}
  	\label{fig:qBoson}
\end{figure}


\subsection{The $q$-Hahn TASEP and Markov duality} 
\label{sub:_q_hahn_tasep_and_markov_duality}

Here we recall the definition of the $q$-Hahn TASEP
process (given in \cite{Povolotsky2013})
and its Markov duality relation to the
$q$-Hahn zero-range process (established in \cite{Corwin2014qmunu}, see also \cite{Barraquand_qhahn_2014}).

\begin{definition}[The $q$-Hahn TASEP]
	Fix an integer $N\ge1$.
	The $q$-Hahn TASEP is a discrete-time Markov process system
	with state space
	\begin{align*}
		\Xbn{N}:=\{\x=(x_0,x_1,\ldots,x_N)\colon +\infty=x_0>x_1>x_2>\ldots>x_N,\; x_j\in\Z\}.
	\end{align*}
	By agreement, $x_0=+\infty$ is a virtual particle, introduced to
	simplify the notation.
	
	At each time step $t\to t+1$,
	each of the particles $x_n(t)$, $1\le n\le N$,
	independently jumps to the right
	by $j_n$, where $j_n\in \{0,1,\ldots,x_{n-1}(t)-x_n(t)-1\}$,
	with probability
	$\phi_{q,\mu,\nu}(j_n\mid x_{n-1}(t)-x_n(t)-1)$.
	See Fig.~\ref{fig:qBoson}, right.
\end{definition}

\begin{remark}
	We consider the $q$-Hahn TASEP with
	a finite number of particles.
	However, since the dynamics of
	each particle $x_n$
	depends only on the dynamics of
	$x_{n-1},\ldots,x_{1}$,
	our considerations can be extended to
	dynamics on semi-infinite configurations
	(having a rightmost particle).
	In doing so, one should restrict
	the class of distributions
	allowed as initial conditions for the $q$-Hahn TASEP.
	We do not investigate such an extension here.
\end{remark}

Let $\Ybn{N}$ denote the subset of $\bigsqcup_{k\ge1}\Yb{k}$
consisting of configurations such that $y_i=0$ unless
$1\le i\le N$.
Define the function
$H\colon \Xbn{N}\times \bigsqcup_{k\ge1}\Yb{k}\to\R$ by
\begin{align*}
	H(\x,\y):=\begin{cases}
		\prod_{i=1}^{N}q^{y_i(x_i+i)},&\y\in\Ybn{N};\\
		0,&\text{otherwise}.
	\end{cases}
\end{align*}

\begin{theorem}[\cite{Corwin2014qmunu}]\label{thm:qHahn_duality}
	The function $H(\x,\y)$ serves as a Markov duality functional
	between the $q$-Hahn TASEP and the $q$-Hahn zero-range process
	in the following sense:
	\begin{align}\label{qHahn_duality}
		\mathbb{E}^{\text{\rm{}TASEP}}_{\x_0}
		H(\x(t),\y_0)=
		\mathbb{E}^{{\text{\rm{}ZRP}}}_{\y_0}
		H(\x_0,\y(t)),\qquad
		\x_0\in\Xbn{N},\quad
		\y_0\in\Ybn{N},
		\quad
		 t=0,1,2,\ldots.
	\end{align}
	Here in the left-hand side the
	expectation is taken over
	the $N$-particle $q$-Hahn TASEP started from the
	initial configuration $\x_0$, and
	in the right-hand side the expectation
	is taken over the $q$-Hahn zero-range process
	started from the initial configuration $\y_0$.
\end{theorem}

This result generalizes known dualities for
the continuous and discrete time $q$-TASEPs
\cite{BorodinCorwinSasamoto2012}, \cite{BorodinCorwin2013discrete}.
In \S \ref{sub:solving_backward_and_forward_evolution_equations}
below we rewrite Theorem \ref{thm:qHahn_duality}
in a form of $q$-Hahn evolution equations, and
use our spatial Plancherel formula (Theorem \ref{thm:spatial_Plancherel}) to solve these equations.


\subsection{Coordinate Bethe ansatz integrability and the left (=~backward) eigenfunctions} 
\label{sub:coordinate_bethe_ansatz_integrability_and_the_left_eigenfunctions}

It is possible to write the action of the
$q$-Hahn transition operator
in terms of a free operator subject to certain two-body boundary conditions.
This was done in \cite{Povolotsky2013}, see also \cite{Corwin2014qmunu}.
This and the next subsection are essentially citations from these two papers.

Denote by $\nbwd$ the operator acting on functions in one variable as
\begin{align*}
	\nbwd f(n):=\frac{\mu-\nu}{1-\nu}f(n-1)+\frac{1-\mu}{1-\nu}f(n).
\end{align*}
Also, let the \emph{free backward operator}
be defined by
\begin{align*}
	\left(\Lbwd u\right)(\n):=\prod_{i=1}^{k}\big[\nbwd\big]_{i}u(\n).
\end{align*}
Here $\big[\nbwd\big]_{i}$ denotes the operator $\nbwd$ acting on the variable $n_i$.
This operator acts on compactly supported functions on the whole
$\Z^{k}$ (so we drop the assumption
that the coordinates $n_j$ are ordered).
Note that the operators $\big[\nbwd\big]_{i}$ for different $i$ act independently, so the
order of their product in the definition of $\Lbwd$ does not matter.

\begin{definition}\label{def:backward_two_body}
	We say that a function $u$ on $\Z^{k}$ satisfies
	the $(k-1)$ \emph{backward two-body boundary conditions} if
	\begin{align}\label{backward_two_body}
		\left.\Big(\nu(1-q) u(\n^{-}_{i,i+1})+(q-\nu) u(\n^{-}_{i+1})+(1-q) u(\n)-(1-q\nu)u(\n^{-}_{i})
		\Big)\right\vert_{\n\in\Z^{k}\colon n_i=n_{i+1}}=0
	\end{align}
	for all $1\le i\le k-1$.
	Denote by $\Bbwd$ the operator in the
	left-hand side of \eqref{backward_two_body}.
\end{definition}

The notion of integrability of the $q$-Hahn particle
system listed in the next proposition dates back
to Bethe \cite{Bethe1931}.

\begin{proposition}[Coordinate Bethe ansatz integrability]\label{prop:CBA_qHahn}
	If $u\colon\Z^{k}\to\C$ satisfies the $(k-1)$ backward boundary conditions
	\eqref{backward_two_body},
	then for all $\n\in\Weyl{k}$,
	\begin{align*}
		(\Lbwd u)(\n)=(\Hbwd u)(\n).
	\end{align*}
\end{proposition}

This proposition allows to construct eigenfunctions
for the backward operator $\Hbwd$ by employing the coordinate
Bethe ansatz. See, e.g., \cite[\S2.3]{BorodinCorwinPetrovSasamoto2013}
for a detailed description of this procedure.

As a first step for the Bethe ansatz, we find eigenfunctions
for the one-particle free operators $\big[\nbwd\big]_{i}$. These clearly are
given by power functions, which we will write in the following
form ($n\in\Z$, $z\in\overline\C\setminus\{1,\nu^{-1}\}$):
\begin{align*}
	\fbwd_{z}(n):=\left(\frac{1-z}{1-\nu z}\right)^{-n},
	\qquad
	\nbwd \fbwd_z=\frac{1-\mu z}{1-\nu z}\fbwd_z.
\end{align*}
Any linear combination of these functions of the form
\begin{align*}
	\Psi^{\textrm{bwd}}_{\z}(\n)=\sum_{\sigma\in S(k)}L_\sigma(\z)
	\prod_{j=1}^{k}\fbwd_{z_{\sigma(j)}}(n_j)
\end{align*}
is an eigenfunction of the free operator
$\Lbwd$ with the eigenvalue
\begin{align}\label{ev_S5}
	\ev(\z):=\prod_{j=1}^{k}\frac{1-\mu z_j}{1-\nu z_j}.
\end{align}

The next step is to select
(among all possible
linear combinations $\Psi^{\textrm{bwd}}_{\z}(\n)$)
the
eigenfunctions which satisfy
the boundary conditions of Definition \ref{def:backward_two_body}.
This can be done by taking (cf. \cite[Lemma 2.8]{BorodinCorwinPetrovSasamoto2013})
\begin{align}
	L_\sigma(\z)=\sgn(\sigma)\prod_{1\le B<A\le k}
	S(z_{\sigma(A)},z_{\sigma(B)}),
	\qquad
	S(z_1,z_2):=
	\frac{\big(\Bbwd (\fbwd_{z_1}\otimes \fbwd_{z_2})\big)(n,n)}
	{\big(\fbwd_{z_1}\otimes \fbwd_{z_2}\big)(n,n)}.
	\label{CBA}
\end{align}
A direct computation yields
\begin{align*}
	\big(\Bbwd (\fbwd_{z_1}\otimes \fbwd_{z_2})\big)(n_1,n_2)=
	\frac{(1-\nu)^2 }{(1-\nu z_1)(1-\nu z_2)}
	(z_1-qz_2)\fbwd_{z_1}(n_1)\fbwd_{z_2}(n_2),\qquad
	n_1,n_2\in\Z.
\end{align*}
Note that if $n_1=n_2$, then $z_1-qz_2$ is the only part of the above expression
which is not symmetric in $z_1\leftrightarrow z_2$.

Therefore, we can take
\begin{align*}
	L_\sigma(\z)=
	\prod_{1\le B<A\le k}\frac{z_{\sigma(A)}-qz_{\sigma(B)}}{z_{\sigma(A)}-z_{\sigma(B)}}.
\end{align*}
Note that this expression differs from the one in
\eqref{CBA} by a factor depending
only on $\z$ and not on $\n$, and the sign
$\sgn(\sigma)$ in \eqref{CBA} is absorbed into the
Vandermonde in the denominator.
With the above choice of $L_\sigma$,
we arrive at the following
\emph{backward eigenfunctions}
($\n\in\Z^{k}$, $z_1,\ldots,z_k\in\overline\C\setminus\{1,\nu^{-1}\}$):
\begin{align*}
	\Psi^{\textrm{bwd}}_{\z}(\n):=
	\sum_{\sigma\in S(k)}\prod_{1\le B<A\le k}
	\frac{z_{\sigma(A)}-qz_{\sigma(B)}}
	{z_{\sigma(A)}-z_{\sigma(B)}}\prod_{j=1}^{k}
	\left(\frac{1-z_{\sigma(j)}}{1-\nu z_{\sigma(j)}}\right)^{-n_j}.
\end{align*}
\begin{proposition}[{\cite[\S4]{Povolotsky2013}}]\label{prop:Psibwd_eigenvalue}
	The function $\Psibwd_{\z}$ is an eigenfunction of the backward free operator $\Lbwd$
	with the eigenvalue
	$\ev(\z)$ \eqref{ev_S5}.
	Moreover, it satisfies the $(k-1)$ two-body backward boundary conditions \eqref{backward_two_body}.
	Consequently,
	$\Psibwd_{\z}(\n)$ restricted to $\n\in\Weyl{k}$
	is an eigenfunction of the backward true $q$-Hahn operator $\Hbwd$, with the eigenvalue
	$\ev(\z)$.
\end{proposition}
The eigenfunctions $\Psi^{\text{\rm{}bwd}}_{\z}$
constructed above
coincide with the left eigenfunctions
$\Psil_{\z}$ \eqref{Psil}.

It seems remarkable that the
eigenfunctions depend only on two of the parameters
$q$, $\mu$, and $\nu$ entering the operator $\Hbwd$.\footnote{Recall
that this property greatly helped us to derive the
spectral biorthogonality of the eigenfunctions
(Theorem \ref{thm:spectral_biorthogonality}).}
In other words,
all operators $\Hbwd$ with the same parameters $(q,\nu)$
have the same eigenfunctions.
This observation points to
a similarity of the
$q$-Hahn stochastic particle system with other
integrable models, cf.
the presence of
commuting transfer matrices
in various solvable lattice models~\cite{baxter2007exactly}.


\subsection{PT-symmetry and the right (=~forward) eigenfunctions} 
\label{sub:formal_stationary_vector_pt_symmetry_and_the_right_eigenfunctions}

The $q$-Hahn transition operator $\Hbwd$
is not Hermitian symmetric. To obtain
eigenfunctions of its transpose,
another property called PT-symmetry should be employed.
A manifestation of PT-symmetry
is the relation between the left and right eigenfunctions
listed in the end of \S \ref{sub:eigenfunctions}.

Let
\begin{align*}
	\Hfwd:=(\Hbwd)^{\text{transpose}},
	\qquad
	\Hcfwd:=\st^{-1}\Hfwd \st,
\end{align*}
where
$\st$ is the diagonal operator
of multiplication by $\st(\n)$ given by \eqref{stationary_measure}.
Recall also the space reflection operator $\Refl$
\eqref{reflection_operator}.

\begin{proposition}[PT-symmetry]\label{prop:PT}
	We have
	\begin{align*}
		\Hcfwd=\Refl \Hbwd \Refl^{-1}\qquad
		\text{and}\qquad
		\Hfwd=(\st \Refl)\Hbwd (\st\Refl)^{-1}.
	\end{align*}
\end{proposition}
\begin{proof}
	See \cite[\S3]{Povolotsky2013}, and also \cite{PovolotskyPriezzhev2006}.
\end{proof}

This means that the action of the conjugated forward operator
$\Hcfwd$
is the same as that of
$\Hbwd$, but in the opposite space (i.e., lattice) direction.
It thus suffices to construct the Bethe ansatz eigenfunctions for the conjugated
forward operator $\Hcfwd$. This is done very similarly to the
case of the operator $\Hbwd$, see \S \ref{sub:coordinate_bethe_ansatz_integrability_and_the_left_eigenfunctions} above. Let us record the necessary modifications. The
one-particle operator and the free operator (corresponding to $\Hcfwd$) are
\begin{align*}
	\nfwd f(n):=\frac{\mu-\nu}{1-\nu}f(n+1)+\frac{1-\mu}{1-\nu}f(n),
	\qquad
	\left(\Lcfwd u\right)(\n):=\prod_{i=1}^{k}\big[\nfwd\big]_{i}u(\n).
\end{align*}

\begin{definition}\label{def:forward_two_body}
	We say that a function $u$ on $\Z^{k}$ satisfies
	the $(k-1)$ \emph{forward two-body boundary conditions} if
	\begin{align}\label{forward_two_body}
		\left.\Big(\nu(1-q) u(\n^{+}_{i,i+1})+(q-\nu) u(\n^{+}_{i})+(1-q) u(\n)-(1-q\nu)u(\n^{+}_{i+1})
		\Big)\right\vert_{\n\in\Z^{k}\colon n_i=n_{i+1}}=0
	\end{align}
	for all $1\le i\le k-1$.
\end{definition}

\begin{proposition}[Coordinate Bethe ansatz integrability]
	If $u\colon\Z^{k}\to\C$ satisfies the $(k-1)$ forward boundary conditions
	\eqref{forward_two_body},
	then for all $\n\in\Weyl{k}$,
	\begin{align*}
		(\Lcfwd u)(\n)=(\Hcfwd u)(\n).
	\end{align*}
\end{proposition}

Continuing to apply the coordinate Bethe ansatz similarly to \S \ref{sub:coordinate_bethe_ansatz_integrability_and_the_left_eigenfunctions}, we arrive at the eigenfunctions
\begin{align*}
	\Psicfwd_{\z}(\n)&:=
	\sum_{\sigma\in S(k)}\prod_{1\le B<A\le k}
	\frac{z_{\sigma(A)}-q^{-1}z_{\sigma(B)}}
	{z_{\sigma(A)}-z_{\sigma(B)}}\prod_{j=1}^{k}
	\left(\frac{1-\nu z_{\sigma(j)}}{1-z_{\sigma(j)}}\right)^{-n_j}
\end{align*}
which satisfy:
\begin{proposition}\label{prop:Psifwd_eigenvalue}
	The function $\Psicfwd_{\z}$ is an eigenfunction
	of the free operator $\Lcfwd$
	with the eigenvalue
	$\ev(\z)$ \eqref{ev_S5}.
	Moreover, it satisfies the $(k-1)$ two-body forward boundary conditions \eqref{forward_two_body}.
	Consequently,
	$\Psicfwd_{\z}(\n)$ restricted to $\n\in\Weyl{k}$
	is an eigenfunction of $\Hcfwd$, with the eigenvalue
	$\ev(\z)$.
\end{proposition}
Returning to the forward operator $\Hfwd$, we get its eigenfunctions
\begin{align*}
	\Psifwd_{\z}(\n)&:=(-1)^k(1-q)^{k}q^{\frac{k(k-1)}{2}}\st(\n)\Psicfwd_{\z}(\n).
\end{align*}
It then follows from Proposition \ref{prop:PT} that
$\Hfwd\Psifwd_\z=\ev(\z)\Psifwd_\z$.
The eigenfunctions $\Psifwd_\z$ coincide with the right
eigenfunctions $\Psir_\z$ \eqref{Psir}
(the constant not depending on $\n$ in front of
$\Psifwd_\z$
is chosen to make certain formulas simpler).


\subsection{Solving backward and forward evolution equations for the $q$-Hahn system} 
\label{sub:solving_backward_and_forward_evolution_equations}

It is possible to solve the backward and forward Kolmogorov equations
(with compactly supported initial conditions)
associated to the $q$-Hahn system $\n(t)$.
The backward equation governs the evolution of expectations of
deterministic
functions in the spatial variables $\n(t)$, and
the forward equation corresponds to the evolution of measures
on the state space $\Weyl{k}$.

\begin{proposition}[Backward evolution equation]\label{prop:bwd_eqn}
	For any $f_0\in\Wc^k$, the backward equation with initial condition
	\begin{align}
		\left\{\begin{array}{rclc}
			f(t+1;\n)&=&(\Hbwd f)(t;\n),
			&\qquad t\in\Z_{\ge0};
			\\\rule{0pt}{14pt}
			f(0,\n)&=&f_0(\n),
		\end{array}\right.
		\label{bwd_eqn_problem}
	\end{align}
	has a unique solution given by
	\begin{align}
		\nonumber
		f(t;\n)&=\big((\Hbwd)^t f_0\big)(\n)
		\\&=
		\sum_{\la\vdash k}
		\oint_{\ga_k}\ldots\oint_{\ga_k}
		d\Plm^{(q)}_\la(\w)
		\prod_{j=1}^{\ell(\la)}\frac{1}{(w_j;q)_{\la_j}(\nu w_j;q)_{\la_j}}
		\big(\ev(\w\circ\la)\big)^{t}\Psil_{\w\circ\la}(\n)
		\llangle f_0,\Psir_{\w\circ\la}\rrangle_{\Wc^k}
		\nonumber\\
		\label{bwd_solution}
		&=
		\oint_{\ga_1}\frac{dz_1}{2\pi\i}
		\ldots
		\oint_{\ga_k}\frac{dz_k}{2\pi\i}
		\prod_{1\le A<B\le k}\frac{z_A-z_B}{z_A-qz_B}
		\prod_{j=1}^{k}
		\left(\frac{1-\mu z_j}{1-\nu z_j}\right)^{t}
		\\&\hspace{160pt}\times\prod_{j=1}^{k}
		\frac{1}{(1-z_j)(1-\nu z_j)}
		\left(\frac{1-z_j}{1-\nu z_j}\right)^{-n_j}
		\llangle f_0,\Psir_{\z}\rrangle_{\Wc^k}.
		\nonumber
	\end{align}

	Consequently, \eqref{bwd_solution} provides an expression for the expectations
	\begin{align}
		\label{f_as_expectation}
		f(t;\n)=\mathbb{E}\left[f_0\big(\n(t)\big)\mid \n(0)=\n\right]
	\end{align}
	for any $f_0\in\Wc^{k}$.
\end{proposition}

\begin{proposition}[Forward evolution equation]\label{prop:fwd_eqn}
	For any $f_0\in\Wc^k$, the forward equation with initial condition
	\begin{align}\label{fwd_eqn_problem}
		\left\{\begin{array}{rclc}
			f(t+1;\n)&=&(\Hfwd f)(t;\n),
			&\qquad t\in\Z_{\ge0};
			\\\rule{0pt}{14pt}
			f(0,\n)&=&f_0(\n),
		\end{array}\right.
	\end{align}
	has a unique solution given by
	\begin{align}
		\nonumber
		f(t;\n)&=\big((\Hfwd)^t f_0\big)(\n)
		\\&=
		\sum_{\la\vdash k}
		\oint_{\ga_k}\ldots\oint_{\ga_k}
		d\Plm^{(q)}_\la(\w)
		\prod_{j=1}^{\ell(\la)}\frac{1}{(w_j;q)_{\la_j}(\nu w_j;q)_{\la_j}}
		\big(\ev(\w\circ\la)\big)^{t}\Psir_{\w\circ\la}(\n)
		\llangle \Psil_{\w\circ\la},f_0\rrangle_{\Wc^k}
		\nonumber\\
		\label{fwd_solution}&=
		(-1)^{k}(1-q)^{k}\st(\n)
		\oint_{\ga_1}\frac{dz_1}{2\pi\i}
		\ldots
		\oint_{\ga_k}\frac{dz_k}{2\pi\i}
		\prod_{1\le A<B\le k}\frac{z_A-z_B}{z_A-qz_B}
		\prod_{j=1}^{k}
		\left(\frac{1-\mu z_j}{1-\nu z_j}\right)^{t}
		\\&\hspace{160pt}\times\prod_{j=1}^{k}
		\frac{1}{(1-z_j)(1-\nu z_j)}
		\left(\frac{1-z_j}{1-\nu z_j}\right)^{n_{k+1-j}}
		\llangle \Psil_{\z},f_0\rrangle_{\Wc^k}.
		\nonumber
	\end{align}

	Consequently, the transition probabilities
	of the $q$-Hahn particle system
	\begin{align*}
		P_\y(t;\x):=\mathrm{Prob}\left\{
		\n(t)=\x\mid \n(0)=\y
		\right\},\qquad t=0,1,\ldots,
	\end{align*}
	are given by
	any of the two formulas in the right-hand side of \eqref{fwd_solution}, with
	$\n=\x$ and
	$f_0(\n)=\mathbf{1}_{\n=\y}$.
\end{proposition}

\par\noindent\textit{Proof of Propositions \ref{prop:bwd_eqn} and \ref{prop:fwd_eqn}.}
	The uniqueness in both cases it evident
	because both operators $\Hbwd$ and $\Hfwd$
	are triangular. The rest follows by applying
	$\Hbwd$ or $\Hfwd$, respectively,
	to appropriate identities
	of Corollary \ref{cor:completeness}.
	To rewrite integral expressions
	with the large contours $\ga$
	in nested form, we have used Proposition
	\ref{prop:nesting_unnesting}.
	To get the second formula in Proposition
	\ref{prop:fwd_eqn}, we also employed the swapping operator
	$\Psiswap$ discussed in \S \ref{sub:eigenfunctions}.
\qed

\medskip

For the formulas \eqref{bwd_solution}--\eqref{fwd_solution} to
be more explicit, one would need to express the initial
conditions $f_0(\n)$
through the eigenfunctions, which will allow to
effectively compute the bilinear pairings
$\llangle\cdot,\cdot\rrangle_{\Wc^k}$
in \eqref{bwd_solution}--\eqref{fwd_solution}.
We do not expect that this can be done explicitly
for generic $f_0(\n)$. We will discuss
this issue in more detail for particular choices of
$f_0(\n)$ in \S \ref{sub:nested_contour_integral_formulas_for_the_q_hahn_tasep}
below.


\subsection{Nested contour integral formulas for the $q$-Hahn TASEP} 
\label{sub:nested_contour_integral_formulas_for_the_q_hahn_tasep}

Recall the $N$-particle $q$-Hahn TASEP discussed in \S \ref{sub:_q_hahn_tasep_and_markov_duality},
$N\ge1$.
Utilizing its duality with the
$q$-Hahn stochastic particle system (Theorem \ref{thm:qHahn_duality}),
it is possible to provide moment formulas for the $q$-Hahn TASEP
started from an arbitrary initial configuration.
The argument here essentially repeats the one in
\cite[\S4.3]{BorodinCorwinPetrovSasamoto2013}.

Let us fix $k\ge1$.
Define $\Weyl{k}_N$ to be the
(finite) subset of $\Weyl{k}$
determined by the inequalities
$N\ge n_1\ge \ldots\ge n_k\ge1$.\footnote{Under the bijection
of $\Weyl{k}$ with $\Yb{k}$ (\S \ref{sub:_q_hahn_boson_process}),
$\Weyl{k}_N$ corresponds to $\Ybn{N}\cap\Yb{k}$ (the set $\Ybn{N}$ is defined
in \S \ref{sub:_q_hahn_tasep_and_markov_duality}).}
Let us
define for (possibly random) initial data
$\x(0)$
for the
$q$-Hahn TASEP,
\begin{align}\label{q_moments_h_t_n_initial_data}
	h_0(\n):=\mathbb{E}\left[\prod_{i=1}^{k}q^{x_{n_i}(0)+n_i}\right]
\end{align}
if $\n\in\Weyl{k}_N$, and $h_0(\n)=0$ otherwise.
Above the expectation is taken with respect to possibly random $\x(0)$,
and we assume that it exists.
Then it follows from Theorem \ref{thm:qHahn_duality}
that the function
\begin{align}\label{q_moments_h_t_n}
	h(t;\n):=\mathbb{E}\left[\prod_{i=1}^{k}q^{x_{n_i}(t)+n_i}\right]
\end{align}
(with the expectation taken with respect to the $q$-Hahn TASEP,
as well as with respect to the possibly random $\x(0)$),
solves the backward evolution equation \eqref{bwd_eqn_problem}
with initial condition $h(0;\n)=h_0(\n)$.
Proposition \ref{prop:bwd_eqn} then immediately implies the following:

\begin{corollary}\label{cor:moment_formulas}
	For the $q$-Hahn TASEP with initial data $\x(0)$,
	the $q$-moments \eqref{q_moments_h_t_n} at any time $t=0,1,2,\ldots$
	are given by \eqref{bwd_solution} with $f_0(\n)=h_0(\n)$
	defined in \eqref{q_moments_h_t_n_initial_data}.
\end{corollary}

The nested contour integral formulas \eqref{bwd_solution}
involve the pairing
$\llangle h_0,\Psir_{\z}\rrangle_{\Wc^{k}}$
of the initial condition
with the right eigenfunction.
An immediate application of the spectral
Plancherel formula (Theorem \ref{thm:spectral_Plancherel})
shows that
when $h_0(\n)$
itself can be written as $(\Pli G)(\n)$ for some $G\in\Cc^{k}_{z}$,
then this pairing has the form
$\llangle h_0,\Psir_{\z}\rrangle_{\Wc^{k}}=G(\z)$.
So, for such special initial conditions
one would get simpler formulas for the $q$-moments of the $q$-Hahn TASEP
\eqref{bwd_solution}, which could be useful for further asymptotic analysis.

Let us apply this idea to the half-stationary (and, in particular, to the
step) initial condition in the $q$-Hahn TASEP.
Let us fix $\rho\in[0,1)$, and consider independent
random variables $X_i$
with the common distribution
\begin{align}\label{full_q_geometric_distribution}
	\mathrm{Prob}(X_i=n)=\rho^{n}
	\frac{(\nu;q)_{n}}{(q;q)_{n}}\frac{(\rho;q)_{\infty}}
	{(\rho\nu;q)_{\infty}},\qquad n=0,1,2,\ldots.
\end{align}
The fact that these probabilities sum to one follows from
the $q$-Binomial theorem.
This distribution corresponds to the product stationary
measure of the $q$-Hahn particle system, cf.
\cite{Povolotsky2013},
\cite{Corwin2014qmunu}.

The half-stationary initial condition is, by definition,
\begin{align*}
	\text{$x_1(0)=-1-X_1$,\quad{}and\quad{}
	$x_i(0)=x_{i-1}(0)-1-X_i$\quad{}for\quad{}$i>1$.}
\end{align*}
When $\rho=0$, the half stationary initial condition
reduces to the step initial condition $x_i(0)=-i$, $i=1,2,\ldots$.
\begin{proposition}\label{prop:IC_computation}
	Let $0\le \rho<q^{k}$.
	Then for any $\n\in\Weyl{k}$,
	we have
	for the half-stationary initial data:
	\begin{align*}
		h_0^{\mathrm{half}}(\n)=\prod_{j=1}^{k}\mathbf{1}_{0<n_j\le N}
		\left(\frac{1-\rho\nu q^{-j}}{1-\rho q^{-j}}\right)^{n_j}.
	\end{align*}
	In particular, the step initial data corresponds to
	$h_0^{\mathrm{step}}(\n)=\prod_{j=1}^{k}\mathbf{1}_{0<n_j\le N}$.
\end{proposition}
\begin{proof}
	First, note that if $X$ has distribution \eqref{full_q_geometric_distribution},
	then by the $q$-Binomial theorem
	\begin{align*}
		\mathbb{E}\left[q^{-m X}\right]=
		\sum_{n=0}^{\infty}
		(\rho/q^{m})^{n}
		\frac{(\nu;q)_{n}}{(q;q)_{n}}
		\frac{(\rho;q)_{\infty}}{(\rho\nu;q)_{\infty}}
		=
		\frac{(\rho;q)_{\infty}}{(\rho\nu;q)_{\infty}}
		\frac{(\rho\nu q^{-m};q)_{\infty}}{(\rho q^{-m};q)_{\infty}}
		=
		\prod_{j=1}^{m}\frac{1-\rho\nu q^{-j}}{1-\rho q^{-j}}
	\end{align*}
	for all $m=0,1,2,\ldots$ such that $\rho<q^{m}$
	(otherwise the above series obviously diverges). Next, observe that
	for the half-stationary initial data, we have
	(with the understanding that $n_{k+1}=0$)
	\begin{align*}
		\prod_{i=1}^{k}q^{x_{n_i}(0)+n_i}=
		\prod_{i=1}^{k}\prod_{m=n_{i+1}+1}^{n_i}q^{-i X_m}.
	\end{align*}
	The random variables $X_m$ in the above product are independent, so
	taking the expectation as above yields the desired result.
\end{proof}

In fact, we will extend the functions $h_0^{\mathrm{half}}(\n)$
and $h_0^{\mathrm{step}}(\n)$
obtained in the previous proposition
to
(non-compactly supported)
$h_0^{\mathrm{step}}(\n):=\prod_{j=1}^{k}\mathbf{1}_{n_j>0}$,
and similarly for
$h_0^{\mathrm{half}}(\n)$.
By the triangular nature of the
$q$-Hahn operator $\Hbwd$,
this will not affect the values of the $q$-moments
$h^{\mathrm{step}}(t;\n)$
and $h^{\mathrm{half}}(t;\n)$ if $\n\in\Weyl{k}_N$.

Our next goal is to employ
the idea discussed after Corollary \ref{cor:moment_formulas},
and write down the moment formulas for the $q$-Hahn
TASEP with half-stationary initial condition.
Note that formulas corresponding to the step initial
data (see Corollary \ref{cor:TASEP_moments} below)
were obtained in \cite{Corwin2014qmunu}.

\begin{proposition}\label{prop:analytic_argument}
	Consider any $0\le \rho<q^{k}$.
	There exist positively oriented closed contours
	$\ga_1,\ldots,\ga_k$ and $\tilde\ga_1,\ldots,\tilde\ga_k$
	which both satisfy the conditions of Definition \ref{def:contours},
	do not include $\rho/q$, and, moreover,
	for all $z\in\ga_i$ and $w\in\tilde\ga_i$ one has
	$|\xinu(z)|<|\xinu(w)|$ (recall the map $\xinu$ \eqref{xinu}).
	
	Moreover, for any $\z\in\C^{k}$ such that $z_i\in\ga_i$ ($1\le i\le k$),
	\begin{align}
		(\Pld h_0^{\mathrm{half}})(\z)=
		(-1)^{k}q^{\frac{k(k-1)}2}
		(\rho\nu q^{-1};q^{-1})_{k}\prod_{j=1}^{k}
		\frac{1-z_j}{z_j-\rho/q}.
		\label{pairing_h0_right}
	\end{align}
\end{proposition}
\begin{proof}
	The existence of such contours can be readily obtained
	similarly to \S \ref{sub:contour_}.

	Assume that $\rho>0$ as the case $\rho=0$ is simpler and can be obtained
	by analogy.
	Let $G(\z)$ denote the right-hand side of \eqref{pairing_h0_right}.
	Proposition \ref{prop:TW_qnu} with
	$c=q/(\nu\rho)$ after a shift of $n_j$'s by one
	implies that
	\begin{align}\label{pairing_h0_right_proof}
		h_0^{\mathrm{half}}(\n)=(\Pli G)(\n).
	\end{align}
	Since this function $G$ does not belong to $\Cc^{k}_{z}$,
	above we have extended the
	definition of $\Pli$
	by requiring that
	the contours $\ga_1,\ldots,\ga_k$
	must not include the point $c^{-1}\nu^{-1}=\rho/q$.

	We show that \eqref{pairing_h0_right_proof} implies
	\eqref{pairing_h0_right} using the spectral Plancherel
	formula (Theorem \ref{thm:spectral_Plancherel}) plus an approximation argument.
	
	Let $G_m\in\Cc^{k}_{z}$ be a sequence
	of functions so that as $m\to\infty$,
	$G_m(\z)$ converges to $G(\z)$
	uniformly over $\z$ with $z_i\in\ga_i$. As $G_m$ one could simply take
	truncations
	of the Laurent series for $G(\z)$ in variables $\xinu(\z)$.
	Theorem \ref{thm:spectral_Plancherel} implies that
	\begin{align}
		\label{pairing_h0_right_proof2}
		G_m(\z)=(\Pld\Pli G_m)(\z)=\sum_{\n\in\Weyl{k}}\Psir_{\z}(\n)(\Pli G_m)(\n).
	\end{align}
	for all $\z$ with $z_i\in\ga_i$ (in fact, the sum goes only over $\n$
	with all parts nonnegative because $G$ and $G_m$ do not have a pole at $z_k=1$,
	cf. the statement of Proposition \ref{prop:TW_qnu}).
	Consider the term corresponding to $\n$ in the summation above.
	Call $\w$ the integration variables
	involved in the definition of $(\Pli G_m)(\n)$.
	Clearly, these integration variables $\w$ can be chosen so that $w_i\in\tilde\ga_i$.
	With this choice of the variables $\w$ and $\z$, the
	$\n$ term can be bounded uniformly in $m$
	by a constant times $\delta^{n_1+\ldots+n_k}$ for some $\delta<1$
	(which is the maximal ratio of $|\xinu(z_i)/\xinu(w_i)|$
	over $z_i\in\ga_i$, $w_i\in\tilde\ga_i$ for all $1\le i\le k$).
	This bound implies that we can take $m\to\infty$
	in \eqref{pairing_h0_right_proof2}, and thus $(\Pld\Pli G)(\z)=G(\z)$
	as desired.
\end{proof}

\begin{corollary}\label{cor:TASEP_moments}
	For the half-stationary initial data with $0<\rho<q^{k}$,
	\begin{align}\label{half_moments}
		\begin{array}{>{\displaystyle}l>{\displaystyle}l}
		&
		\mathbb{E}^{\mathrm{half}}\left[\prod_{i=1}^{k}q^{x_{n_i}(t)+n_i}\right]=
		(-1)^{k}q^{\frac{k(k-1)}2}
		(\rho\nu q^{-1};q^{-1})_{k}
		\oint_{\ga_1}\frac{dz_1}{2\pi\i}
		\ldots
		\oint_{\ga_k}\frac{dz_k}{2\pi\i}
		\prod_{1\le A<B\le k}\frac{z_A-z_B}{z_A-qz_B}
		\\&\hspace{160pt}\times\prod_{j=1}^{k}
		\left(\frac{1-\mu z_j}{1-\nu z_j}\right)^{t}
		\prod_{j=1}^{k}
		\frac{1}{(z_j-\rho/q)(1-\nu z_j)}
		\left(\frac{1-z_j}{1-\nu z_j}\right)^{-n_j}.
		\end{array}
	\end{align}
	Here the integration contours are as in Definition \ref{def:contours}
	with an additional restriction that they do not include $\rho/q$
	(this is possible due to the restriction on $\rho$).
	
	In particular, for the step initial data (case $\rho=0$ of the above formula),
	\begin{align}\label{step_moments}
		\begin{array}{>{\displaystyle}l>{\displaystyle}l}
		&
		\mathbb{E}^{\mathrm{step}}\left[\prod_{i=1}^{k}q^{x_{n_i}(t)+n_i}\right]=
		(-1)^{k}q^{\frac{k(k-1)}2}
		\oint_{\ga_1}\frac{dz_1}{2\pi\i}
		\ldots
		\oint_{\ga_k}\frac{dz_k}{2\pi\i}
		\prod_{1\le A<B\le k}\frac{z_A-z_B}{z_A-qz_B}
		\\&\hspace{160pt}\times\prod_{j=1}^{k}
		\left(\frac{1-\mu z_j}{1-\nu z_j}\right)^{t}
		\prod_{j=1}^{k}
		\frac{1}{z_j(1-\nu z_j)}
		\left(\frac{1-z_j}{1-\nu z_j}\right)^{-n_j}.
		\end{array}
	\end{align}
	Here the integration contours must not include $0$.
\end{corollary}
\begin{proof}
	This immediately follows from
	Markov duality (Theorem \ref{thm:qHahn_duality})
	combined with Proposition \ref{prop:bwd_eqn}
	allowing to solve backward evolution equations
	for the $q$-Hahn stochastic particle system.
\end{proof}

\begin{remark}
	Formula \eqref{step_moments}
	serves as a starting point
	towards a Fredholm determinantal expression
	for a $q$-Laplace transform of the
	position of any particle $x_n(t)$
	in the $q$-Hahn TASEP started from the step initial condition
	\cite{Corwin2014qmunu}.
	To get a Fredholm determinant, one must use moments \eqref{step_moments}
	for all $k=1,2,\ldots$.
	Thus, in the half-stationary case, the restriction $\rho<q^{k}$
	in \eqref{half_moments} (i.e., the lack of finiteness of moments)
	presents an impediment
	to getting a Fredholm determinantal formula.
\end{remark}


\subsection{Symmetrization identities from the spectral Plancherel theorem} 
\label{sub:symmetrization_formula}

The spectral Plancherel
formula (Theorem \ref{thm:spectral_Plancherel})
implies a family of nontrivial identities
involving the
summation over the
symmetric group $S(k)$. (This symmetrization comes from the definition of the
eigenfunctions, cf. \eqref{Psil}, \eqref{Psir}.)

\begin{proposition}\label{prop:general_summation}
	For any $G\in\Cc^{k}_z$ and $\vxi=(\xi_1,\ldots,\xi_k)\in(\C\setminus\{0\})^{k}$,
	\begin{align}
		\label{Gxi_symmetrization_S5}
		\sum_{\sigma\in S(k)}
		\prod_{B<A}
		\frac{\Sm(\xi_{\sigma(B)},\xi_{\sigma(A)})}
		{\xi_{\sigma(A)}-\xi_{\sigma(B)}}
		\Bigg(\sum_{\n\in\Weyl{k}}
		\st(\n)
		(\Pli G)(\n)
		\prod_{j=1}^{k}
		\xi_{\sigma(j)}^{n_j}\Bigg)
		=
		(q-1)^{-k}
		\left(\frac{1-\nu}{1-q\nu}\right)^{\frac{k(k-1)}{2}}G(\xinu(\vxi)),
	\end{align}
	where $\xinu$, $\Sm$, and $\st$ are given in
	\eqref{xinu}, \eqref{quadratic_S_matrix}, and \eqref{stationary_measure}, respectively.
\end{proposition}
Note that the sum over $\n$ in the
left-hand side of \eqref{Gxi_symmetrization_S5} is finite.
\begin{proof}
	Theorem \ref{thm:spectral_Plancherel}
	states that the operator
	$\Plspectral=\Pld\Pli$
	\eqref{Plspectral}
	acts as the identity operator
	in the space $\Cc^{k}_{z}$ of symmetric Laurent polynomials
	in the variables $\frac{1-z_j}{1-\nu z_j}$.
	Expanding this statement, we can write
	\begin{align*}
		G(\z)&=
		(\Pld \Pli G)(\z)
		=
		\sum_{\n\in\Weyl{k}}
		\Psir_{\z}(\n)(\Pli G)(\n)
		\\&=\nonumber
		(-1)^k(1-q)^{k}q^{\frac{k(k-1)}{2}}
		\sum_{\sigma\in S(k)}
		\prod_{B<A}
		\frac{z_{\sigma(A)}-q^{-1}z_{\sigma(B)}}
		{z_{\sigma(A)}-z_{\sigma(B)}}
		\sum_{\n\in\Weyl{k}}
		\st(\n)
		(\Pli G)(\n)
		\prod_{j=1}^{k}
		\left(\frac{1-z_{\sigma(j)}}{1-\nu z_{\sigma(j)}}\right)^{n_j}.
	\end{align*}
	The change of variables $\xinu$ \eqref{xinu}
	yields the desired formula.
\end{proof}
Let us illustrate Proposition \ref{prop:general_summation}
with two examples.
For the first (simplest possible)
example, take
$G(\z)\equiv1$. Then
one can readily check (similarly to the proof of
Proposition \ref{prop:TW_qnu}) that
\begin{align*}
	(\Pli G)(\n)=
	\frac{(-1)^{k}}{(\nu;q)_{k}}
	\mathbf{1}_{n_1=\ldots=n_k=0}.
\end{align*}
Thus, the summation over $\n$ in the
left-hand side of \eqref{Gxi_symmetrization_S5} reduces to
a single term, and the whole identity becomes
\begin{align*}
	\sum_{\sigma\in S(k)}
	\prod_{B<A}
	\frac{\Sm(\xi_{\sigma(B)},\xi_{\sigma(A)})}
	{\xi_{\sigma(A)}-\xi_{\sigma(B)}}
	=
	\frac{(q;q)_{k}}{(1-q)^{k}}
	\left(\frac{1-\nu}{1-q\nu}\right)^{\frac{k(k-1)}{2}},
\end{align*}
which
(under the change of variables $\xinu$)
is the classical symmetrization identity listed as \eqref{Mac_symm_identity} above.

The second example is given in the following corollary:
\begin{corollary}\label{cor:summation_qnu}
	For any $k\ge1$, any $c\in\C\setminus\{\nu^{-1}, q^{-1}\nu^{-1},\ldots, q^{-(k-1)}\nu^{-1}\}$,
	and any $\xi_1,\ldots,\xi_k\in\C\setminus\{0\}$
	with $|\xi_i|<\left|\frac{1-c\nu q^{j-1}}{\nu(1-cq^{j-1})}\right|$ for $1\le i,j\le k$,
	\begin{align}\label{q_nu_symmetrization_identity}
		\begin{array}{ll}
			&\displaystyle
			\sum_{\sigma\in S(k)}
			\prod_{B<A}
			\frac{\Sm(\xi_{\sigma(B)},\xi_{\sigma(A)})}
			{\xi_{\sigma(A)}-\xi_{\sigma(B)}}
			\Bigg(\sum_{n_1\ge \ldots\ge n_k\ge0}
			\st(\n)
			\prod_{j=1}^{k}
			\left(\nu\xi_{\sigma(j)}\frac{1-cq^{j-1}}{1-c\nu q^{j-1}}\right)^{n_j}
			\Bigg)
			\\&\displaystyle\hspace{100pt}\rule{0pt}{25pt}=
			(c\nu;q)_{k}
			\left(\frac{1-\nu}
			{1-q}\right)^{k}
			\left(\frac{1-\nu}{1-q\nu}\right)^{\frac{k(k-1)}{2}}
			\prod_{j=1}^{k}
			\frac{1}{(1-c \nu)-\nu\xi_j(1-c)}.
		\end{array}
	\end{align}

	In particular, for $c=\infty$,
	\begin{align}\label{q_nu_symmetrization_identity_step}
		\begin{array}{ll}
			&\displaystyle
			\sum_{\sigma\in S(k)}
			\prod_{B<A}
			\frac{\Sm(\xi_{\sigma(B)},\xi_{\sigma(A)})}
			{\xi_{\sigma(A)}-\xi_{\sigma(B)}}
			\bigg(\sum_{n_1\ge \ldots\ge n_k\ge0}
			\st(\n)
			\prod_{j=1}^{k}
			\xi_{\sigma(j)}^{n_j}
			\bigg)
			\\&\hspace{170pt}\displaystyle=
			\left(\frac{1-\nu}
			{1-q}\right)^{k}
			\left(\frac{q(1-\nu)}{1-q\nu}\right)^{\frac{k(k-1)}{2}}
			\prod_{j=1}^{k}
			\frac{1}{1-\xi_j},
		\end{array}
	\end{align}
	provided that $|\xi_j|<1$ for $1\le j\le k$.
\end{corollary}
\begin{proof}
	Observe that our restrictions on $c$ and the $\xi_j$'s
	ensures that both sides of \eqref{q_nu_symmetrization_identity} are well defined.
	If $c$ is real and $c>q^{-(k-1)}\nu^{-1}$, the statement follows from
	Proposition \ref{prop:TW_qnu} plus approximation
	considerations similar to \S \ref{sub:nested_contour_integral_formulas_for_the_q_hahn_tasep}
	which allow to work with the function
	$G(\z)=\prod_{j=1}^{k}\frac{1-\nu z_j}{1-c\nu z_j}$
	not belonging to the space $\Cc^{k}_{z}$.
	(Note that the restriction $c>q^{-(k-1)}\nu^{-1}$ corresponds to
	the restriction $\rho<q^{k}$ in Proposition \ref{prop:analytic_argument}.)
	For other values of $c$, identity
	\eqref{q_nu_symmetrization_identity}
	holds because its both sides are rational functions in $c$
	(indeed, the summation over $\n$ in the left-hand side also results in a
	rational function, cf. Remark \ref{rmk:rational_f_summation} below).
	Finally, the identity \eqref{q_nu_symmetrization_identity_step} is
	a limit of \eqref{q_nu_symmetrization_identity} as $c\to\infty$.
\end{proof}
\begin{remark}\label{rmk:rational_f_summation}
	For general parameters $q$ and $\nu$
	it is not clear how to evaluate the sum over $\n$
	in the left-hand side of
	\eqref{q_nu_symmetrization_identity_step} (or
	\eqref{q_nu_symmetrization_identity}) in a closed form.
	The reason is that the quantities $\st(\n)$ \eqref{stationary_measure}
	depend on the structure of clusters of $\n$.
	The part of the sum corresponding to each fixed cluster structure
	of $\n$ (they are indexed by partitions $\la$ of $k$)
	reduces to a sum of several geometric sequences.
	For example, for $k=4$
	the terms corresponding to the cluster structure
	$\la=(3,1)$ are
	\begin{align*}
		&
		\sum_{m_1>m_2\ge0}
		\Big(
		(\xi_1\xi_2\xi_3)^{m_1}\xi_4^{m_2}
		+\xi_1^{m_1}(\xi_2\xi_3\xi_4)^{m_2}
		\Big)
		=
		\frac{\xi_1 (1 + \xi_2 \xi_3 - 2 \xi_1 \xi_2 \xi_3)}
		{(1-\xi_1)(1-\xi_1\xi_2\xi_3)(1-\xi_1\xi_2\xi_3\xi_4)}
	\end{align*}
	(we assumed $\sigma=id$ to simplify the notation).
	The whole sum $\sum_{n_1\ge \ldots\ge n_k\ge0}\st(\n)\xi_1^{n_1}\ldots\xi_k^{n_k}$
	is a rational function in the $\xi_j$'s with the common denominator
	$(1-\xi_1)(1-\xi_1\xi_2)\ldots(1-\xi_1 \ldots\xi_k)$, and
	it is not clear whether it is possible to express the
	numerator of this rational function
	in a closed form.

	However, for certain special values of $q$ and $\nu$,
	the coefficients $\st(\n)$
	simplify,
	which allows to sum the series in $\n$
	in a closed form. Then identities
	\eqref{q_nu_symmetrization_identity}--\eqref{q_nu_symmetrization_identity_step}
	reduce to
	Tracy-Widom symmetrization identities,
	see \S \ref{sub:tracy_widom_symmetrization_identities} below.
\end{remark}



\section{Conjugated $q$-Hahn operator} 
\label{sec:spectral_theory_for_the_conjugated_hahn_boson_operator}

In this section we briefly discuss
analogues of our main results
for the $q$-Hahn transition operator
conjugated in a certain manner.
The conjugated operator
is no longer
stochastic, but the main results can be
readily extended to its eigenfunctions.
Under a certain degeneration
(which we discuss in \S \ref{sec:application_to_six_vertex_model} below), these eigenfunctions
become related to the
six-vertex model and the Heisenberg XXZ quantum spin chain.
Relations between eigenfunctions of various models are indicated on Fig.~\ref{fig:big_scheme}.

\subsection{Analogues of main results} 
\label{sub:analogues_of_main_results}

Let $\Dil{\dilp}$, where $\dilp\in\C\setminus\{0\}$, be
the dilation operator acting on $\Wc^{k}$
by
\begin{align}\label{Dilation}
	(\Dil{\dilp}f)(\n):=\dilp^{n_1+\ldots+n_k}f(\n),
\end{align}
and consider the \emph{conjugated $q$-Hahn operator}
$\Dil{\dilp}^{-1}\Hbwd \Dil{\dilp}$. Note that for $\dilp\ne 1$
this operator is no longer stochastic.
As in the stochastic $\dilp=1$ case (cf. \S \ref{sub:formal_stationary_vector_pt_symmetry_and_the_right_eigenfunctions}), the conjugated operator
$\Dil{\dilp}^{-1}\Hbwd \Dil{\dilp}$
is not Hermitian symmetric.

Let us explain how
our main results (Plancherel formulas, spectral biorthogonality, symmetrization identities)
can be extended to left and right eigenfunctions
($\Dil{\dilp}^{-1}\Psil_{\z}$
and $\Dil{\dilp}\Psir_{\z}$,
respectively)
of the conjugated
$q$-Hahn operator.
All modified results below in this section
are equivalent to the corresponding results in \S \ref{sec:main_results},
\S \ref{sec:spectral_biorthogonality_of_eigenfunctions},
and \S \ref{sec:the_q_mu_nu_boson_process_and_coordinate_bethe_ansatz}.

A slightly different notation
for the modified functions
turns out to be convenient
(cf. \eqref{Psil} and \eqref{Psir}):
\begin{align}
	\Psild_{\z}(\n)&:=
	\sum_{\sigma\in S(k)}\prod_{B<A}
	\frac{z_{\sigma(A)}-qz_{\sigma(B)}}
	{z_{\sigma(A)}-z_{\sigma(B)}}\prod_{j=1}^{k}
	\left(\frac{\dilp-z_{\sigma(j)}}{1-\nu z_{\sigma(j)}}\right)^{-n_j},
	\label{Psild}\\
	\label{Psird}
	\Psird_{\z}(\n)&:=
	(-1)^k(1-q)^{k}q^{\frac{k(k-1)}{2}}\std(\n)
	\sum_{\sigma\in S(k)}\prod_{B<A}
	\frac{z_{\sigma(A)}-q^{-1}z_{\sigma(B)}}
	{z_{\sigma(A)}-z_{\sigma(B)}}\prod_{j=1}^{k}
	\left(\frac{\dilp-z_{\sigma(j)}}{1-\nu z_{\sigma(j)}}\right)^{n_j},
\end{align}
where $\std(\n):=\prod_{j=1}^{M(\n)}\frac{(\dilp \nu;q)_{c_j}}{(q;q)_{c_j}}$,
cf. \eqref{stationary_measure}.
The restrictions on parameters are $0<q<1$ and $0\le \mu\le \dilp\nu<1$.
When $\dilp=1$, the functions \eqref{Psild}--\eqref{Psird}
become the eigenfunctions $\Psil_{\z}$ and $\Psir_{\z}$
studied in the previous sections.

\begin{proposition}
	The functions
	$\Psild_{\z}$ and $\Psird_{\z}$
	are respectively the left and right eigenfunctions of
	the conjugated $q$-Hahn
	operator
	$\Dil{\dilp}^{-1}
	\mathcal{H}^{\mathrm{bwd}}_{q,\mu,\dilp\nu}
	\Dil{\dilp}$.
	The eigenvalue
	of either $\Psild_{\z}$ or $\Psird_{\z}$
	is equal to
	$\prod_{j=1}^{k}\frac{1-\mu z_j/\dilp}{1-\nu z_j}$.
\end{proposition}
\begin{proof}
	One can readily check that
	\begin{align*}
		\Psild_{\z}(\n)=\dilp^{-n_1- \ldots-n_k}\Psil_{\z/\dilp}(\n)\vert_{\nu\to\dilp\nu},
		\qquad
		\Psird_{\z}(\n)=\dilp^{n_1+ \ldots+n_k}\Psir_{\z/\dilp}(\n)\vert_{\nu\to\dilp\nu},
	\end{align*}
	so the claim
	follows by
	considering the functions
	$\Dil{\dilp}^{-1}\Psil_{\z}$
	and $\Dil{\dilp}\Psir_{\z}$
	and
	rescaling the spectral variables $z_j\to z_j/\dilp$
	and the parameter $\nu\to \dilp\nu$.
\end{proof}
Note that
the eigenfunctions $\Psild_{\z}$ and $\Psird_{\z}$
\eqref{Psild}--\eqref{Psird}
can be constructed by applying the coordinate Bethe ansatz
to the conjugated operator $\Dil{\dilp}^{-1}
\mathcal{H}^{\mathrm{bwd}}_{q,\mu,\dilp\nu}
\Dil{\dilp}$ similarly
to what was done in \S \ref{sub:coordinate_bethe_ansatz_integrability_and_the_left_eigenfunctions}
and \S \ref{sub:formal_stationary_vector_pt_symmetry_and_the_right_eigenfunctions}.

\begin{remark}\label{rmk:generic_linear_fractional_not_more_general}
	Note that replacing
	$\frac{\dilp-z}{1-\nu z}$
	in \eqref{Psild}--\eqref{Psird}
	by a generic linear fractional expression
	$\frac{az+b}{cz+d}$
	does not introduce further generality.
	Indeed, the third independent parameter
	coming from a generic expression
	could be absorbed by rescaling the
	spectral variables $\z$.
	Thus, the conjugated $q$-Hahn eigenfunctions \eqref{Psild}--\eqref{Psird}
	are essentially the most general rational symmetric functions
	which can be obtained from the Hall-Littlewood polynomials
	\cite[Ch. III]{Macdonald1995}
	\begin{align*}
		P_\la(z_1,\ldots,z_k;t)=\textit{const}_{\la,k}\cdot\sum_{\sigma\in S(k)}
		\prod_{B<A}\frac{z_{\sigma(B)}-tz_{\sigma(A)}}
		{z_{\sigma(B)}-z_{\sigma(A)}}
		z_{\sigma(1)}^{\la_1}\ldots
		z_{\sigma(k)}^{\la_k},
		\qquad\la_1\ge \ldots\ge\la_k\ge0,\qquad\la_i\in\Z
	\end{align*}
	by replacing the power terms
	$z_{\sigma(j)}^{\la_j}$
	by arbitrary linear fractional expressions (raised to powers $\la_j$).
	See also \S \ref{ssub:_q_boson_to_van_diejen_s_delta_bose_gas}
	and \cite{vanDiejen2004HL}
	for a discussion of an interacting particle system diagonalized by the
	Hall-Littlewood polynomials.
\end{remark}

Define the modified map
$\xinu_{\dilp}(z):=\frac{\dilp-z}{1-\nu z}$ (cf. \eqref{xinu}). One readily checks that it is an involution, too:
$\xinu_{\dilp}(\xinu_{\dilp}(z))=z$.
Let
$\Cc^{k}_{z;\dilp}$
be the space of
symmetric Laurent polynomials in the
variables $\xinu_{\dilp}(z_1),\ldots,\xinu_{\dilp}(z_k)$.
The map $\xinu_{\dilp}$ is an isomorphism between
$\Cc^{k}_{z;\dilp}$ and the space $\Cc^{k}_{\xi}$
of symmetric Laurent polynomials
in the variables $\xi_1,\ldots,\xi_k$.
In this way the functions
\eqref{Psild}--\eqref{Psird}
induce the following functions
in the variables $\vxi$ (cf. \eqref{Phil_xi}--\eqref{Phir_xi}):
\begin{align}
	\Phild_{\vxi}(\n)&:=
	\sum_{\sigma\in S(k)}
	\sgn(\sigma)\prod_{B<A}\Smd(\xi_{\sigma(A)},\xi_{\sigma(B)})
	\prod_{j=1}^{k}\xi_{\sigma(j)}^{-n_j};
	\label{Phild}\\\label{Phird}
	\Phird_{\vxi}(\n)&:=
	\std(\n)\left(\frac{1-q}{1-\dilp\nu}\right)^{k}
	\sum_{\sigma\in S(k)}
	\sgn(\sigma)\prod_{B<A}\Smd(\xi_{\sigma(B)},\xi_{\sigma(A)})
	\prod_{j=1}^{k}\xi_{\sigma(j)}^{n_j},
\end{align}
where
\begin{align}\label{Sm_theta}
	\Smd(\xi_1,\xi_2):=
	\frac{\dilp(1-q)}{1-q\dilp\nu}+
	\frac{q-\dilp\nu}{1-q\dilp\nu}\xi_2+
	\frac{\nu(1-q)}{1-q\dilp\nu}\xi_1\xi_2
	-\xi_1.
\end{align}
One can readily check that
\begin{align}\label{Psi_Phi_relations}
	\begin{array}{>{\displaystyle}rc>{\displaystyle}l}
	 	\Psild_{\xinu_{\dilp}(\vxi)}(\n)&=&\frac{1}{\Vand(\vxi)}
		\left(\frac{1-q\dilp\nu}{1-\dilp\nu}\right)^{\frac{k(k-1)}{2}}
		\Phild_{\vxi}(\n),\\
		\Psird_{\xinu_{\dilp}(\vxi)}(\n)&=&\frac{(-1)^{\frac{k(k+1)}{2}}}{\Vand(\vxi)}
		(1-\dilp\nu)^{k}
		\left(\frac{1-q\dilp\nu}{1-\dilp\nu}\right)^{\frac{k(k-1)}{2}}
		\Phird_{\vxi}(\n).
	\end{array}
\end{align}
(this should be compared to formulas \eqref{Psil_xi}--\eqref{Psir_xi}).

\begin{remark}
	The quadratic cross-term $\Smd(\xi_1,\xi_2)$
	depends on three parameters, and so
	one may also parametrize the eigenfunctions
	by $\al=\frac{\nu(1-q)}{1-q\dilp\nu}$,
	$\be=\frac{q-\dilp\nu}{1-q\dilp\nu}$, and
	$\gamma=\frac{\dilp(1-q)}{1-q\dilp\nu}$,
	where now $\al+\be+\gamma$ is not necessarily equal to $1$
	(cf. Remark \ref{rmk:Povolotsky_parameters}).
	We see that the cross-term $\Smd$
	now takes a ``generic'' form,
	which is
	another indication
	towards the claim discussed in
	Remark~\ref{rmk:generic_linear_fractional_not_more_general} above.
	
	Note also that formulas expressing $(q,\nu,\dilp)$
	through $(\al,\be,\gamma)$ involve solving quadratic equations
	which greatly simplify when $\al+\be+\gamma=1$.
\end{remark}

\begin{definition}
	The modified transforms are defined as
	\begin{align}
		\label{Pldd}
		(\Plde^{q,\nu,\dilp} f)(\z)&:=
		\sum_{\n\in\Weyl{k}}f(\n)\Psird_{\z}(\n),\\
		\label{Plid}
		(\Plie^{q,\nu,\dilp} G)(\n)&:=
		\oint_{\ga_1^{\dilp}}\frac{dz_1}{2\pi\i}
		\ldots
		\oint_{\ga_k^{\dilp}}\frac{dz_k}{2\pi\i}
		\prod_{A<B}\frac{z_A-z_B}{z_A-qz_B}
		\prod_{j=1}^{k}
		\frac{1}{(\dilp-z_j)(1-\nu z_j)}
		\left(\frac{\dilp-z_j}{1-\nu z_j}\right)^{-n_j}
		G(\z).
	\end{align}
	The integration contours are
	as follows:
	$\ga_k^{\dilp}$
	is a small positively oriented contour around
	$\dilp$ not containing $q\dilp$,
	$\ga_A^{\dilp}$ contains $q\ga_B^{\dilp}$
	for all $1\le A<B\le k$, and,
	moreover, $\nu^{-1}$ is outside all contours
	(cf. Definition \ref{def:contours} of contours corresponding to the case $\dilp=1$).

	The operator $\Plde^{q,\nu,\dilp}$
	takes functions from $\Wc^{k}$ to
	$\Cc^{k}_{z;\dilp}$,
	and the operator $\Plie^{q,\nu,\dilp}$ acts in the opposite direction.
\end{definition}

\begin{theorem}[Plancherel isomorphisms; cf. Theorems \ref{thm:spatial_Plancherel} and \ref{thm:spectral_Plancherel}]\label{thm:Plancherel_isom_theta}
	The maps $\Plde^{q,\nu,\dilp}$ and $\Plie^{q,\nu,\dilp}$ are mutual inverses in the sense that
	$\Plspatiale^{q,\nu,\dilp}=\Plie^{q,\nu,\dilp}\Plde^{q,\nu,\dilp}$ acts as the identity
	operator on $\Wc^{k}$, and
	$\Plspectrale^{q,\nu\dilp}=\Plde^{q,\nu,\dilp}\Plie^{q,\nu,\dilp}$
	coincides with the identity operator on
	$\Cc^{k}_{z;\dilp}$.
\end{theorem}
We will need a modified version of identity \eqref{completeness_big_1}
with $f(\n)=\mathbf{1}_{\n=\x}$:
\begin{align}\label{theta_small_contour}
	\sum_{\la\vdash k}
	\oint_{\ga_k^{\dilp}}\ldots\oint_{\ga_k^{\dilp}}
	\dilp^{-k}d\Plm^{(q)}_\la(\w)
	\prod_{j=1}^{\ell(\la)}\frac{1}{(w_j/\dilp;q)_{\la_j}(\nu w_j;q)_{\la_j}}
	\Psild_{\w\circ\la}(\x)
	\Psird_{\w\circ\la}(\y)=\mathbf{1}_{\x=\y}
\end{align}
for all $\x,\y\in\Weyl{k}$. Identity \eqref{theta_small_contour} follows from
\eqref{completeness_big_1} via a change of variables $w_j\to w_j/\dilp$
plus the change of parameter $\nu\to\dilp\nu$
(note that $d\Plm^{(q)}_\la(\w/\dilp)=\dilp^{-k}d\Plm^{(q)}_\la(\w)$).
Note that \eqref{theta_small_contour} is
equivalent to the statement that $\Plspatiale^{q,\nu,\dilp}$
is the identity operator
(here one should use an analogue of Proposition \ref{prop:nesting_unnesting}).

\begin{remark}\label{rmk:theta_bilinear_pairings_not_needed}
	One can define a bilinear pairing
	on the space $\Cc^{k}_{z;\dilp}$
	similarly to \eqref{Cz_pairing}
	which would correspond to
	the standard bilinear pairing \eqref{W_pairing}
	on $\Wc^{k}$ in the sense of
	\eqref{Plspatial_isomorphism} and \eqref{Plspectral_isomorphism}.
	We will not write down the corresponding formulas.
\end{remark}

We record the spectral biorthogonality statements
using the formal notation as in Remark~\ref{rmk:spectral_biorthogonality_with_deltas}:

\begin{theorem}[Spectral biorthogonality; cf. Theorems \ref{thm:spectral_biorthogonality} and \ref{thm:spectral_biorthogonality_xi}]
	In the spectral coordinates $\z$ and $\vxi$, one has, respectively,
	\begin{align}
		\label{spec_bio_theta}
		\sum_{\n\in\Weyl{k}}
		\Psird_{\z}(\n)\Psild_{\w}(\n)
		\Vand(\z)\Vand(\w)&=
		(-1)^{\frac{k(k-1)}{2}}
		\prod_{j=1}^{k}(\dilp-z_j)(1-\nu z_j)
		\prod_{A\ne B}(z_A-qz_B)
		\det[\delta({z_i-w_j})]_{i,j=1}^{k},
		\\
		\label{spec_bio_theta_xi}
		\sum_{\n\in\Weyl{k}}
		\Phird_{\vxi}(\n)\Phild_{\vom}(\n)
		&=
		\prod_{j=1}^{k}\xi_j
		\prod_{A\ne B}\Smd(\xi_A,\xi_B)
		\det[\delta({\xi_i-\om_j})]_{i,j=1}^{k}.
	\end{align}
	Identity \eqref{spec_bio_theta} holds with test functions
	which are Laurent polynomials
	(not necessarily symmetric) in $\frac{\dilp-z_j}{1-\nu z_j}$
	and $\frac{\dilp-w_j}{1-\nu w_j}$, respectively, with the
	integration contours being small positively oriented circles around
	$\dilp$ (or negatively oriented circles around $\nu^{-1}$).

	Similarly, identity \eqref{spec_bio_theta_xi} holds with test functions
	which are Laurent polynomials in $\xi_j$ and $\om_j$, respectively,
	and with integration contours being any positively oriented circles around $0$.
\end{theorem}

Note that the above theorem
also works with more general test functions
under the right choice of integration contours.
The appropriate statements can be readily formulated similarly to \S \ref{sub:formulations},
but we will not write them down.


\subsection{Symmetrization identities} 
\label{sub:symmetrization_identities}

We will now give $\dilp$-modified analogues of the symmetrization identities
of Corollary \ref{cor:summation_qnu}. Recall that these identities
were based on an explicit computation
of the inverse transform of a certain function $G(\z)$.
The corresponding $\dilp$-modified integral looks as follows:

\begin{proposition}[cf. Proposition \ref{prop:TW_qnu}]
	For any $c\in\C\setminus\{\dilp^{-1}\nu^{-1}, q^{-1}\dilp^{-1}\nu^{-1},
	\ldots, q^{-(k-1)}\dilp^{-1}\nu^{-1}\}$
	and any $\n\in\Weyl{k}$, we have
	\begin{align*}
		&
		(\Plie^{q,\nu,\dilp}G)(\n)=\oint_{\ga_1^{\dilp}}\frac{dz_1}{2\pi\i}
		\ldots
		\oint_{\ga_k^{\dilp}}\frac{dz_k}{2\pi\i}
		\prod_{1\le A<B\le k}\frac{z_A-z_B}{z_A-qz_B}
		\prod_{j=1}^{k}\left(\frac{\dilp-z_j}{1-\nu z_j}\right)^{-n_j}
		\frac{1}{(\dilp-z_j)(1-c\nu  z_j)}
		\\&\hspace{220pt}=
		\frac{(-1)^{k}\nu^{n_1+n_2+\ldots+n_k}}{(c\dilp\nu;q)_{k}}\prod_{j=1}^{k}
		\left(\frac{1-cq^{j-1}}{1-c\dilp\nu q^{j-1}}\right)^{n_j}
		\mathbf{1}_{n_k\ge0},
	\end{align*}
	where
	$G(\z):=\prod_{j=1}^{k}\frac{1-\nu z_j}{1-c\nu z_j}$.
	Here the integration contours
	$\ga_1^{\dilp},\ldots,\ga_k^{\dilp}$ are as in
	\eqref{Plid}
	with an additional condition that they do
	not contain $c^{-1}\nu^{-1}$ (this is possible because of our restrictions on
	$c$, but this could also mean that each contour is a union of several disjoint
	simple contours).
\end{proposition}
The symmetrization identities (analogues of Corollary \ref{cor:summation_qnu})
we obtain in the
$\dilp$-modified case turn out to be equivalent
to the ones for $\dilp=1$
(up to rescaling of $\nu$ and of the $\xi_j$'s).
Therefore, we will not write them down.

One could also readily formulate an analogue of Proposition \ref{prop:general_summation}
for the general function $G(\z)$, but we will not do that.

\section{Application to ASEP} 
\label{sec:application_to_asep}

In this section we discuss a degeneration of the
$q$-Hahn eigenfunctions to those of
the ASEP. Note that
in contrast with the previous section,
we work at the level of
stochastic particle systems (cf. Fig.~\ref{fig:big_scheme}).
Our main results of \S \ref{sec:main_results},
\S \ref{sec:spectral_biorthogonality_of_eigenfunctions},
and \S \ref{sec:the_q_mu_nu_boson_process_and_coordinate_bethe_ansatz}
carry over to the ASEP case (with certain modifications which are not always straightforward).
This in particular leads to different proofs of some of the results
for the ASEP first obtained by Tracy and Widom \cite{TW_ASEP1}, \cite{TW_ASEP4}.

\subsection{ASEP and its Bethe ansatz eigenfunctions} 
\label{sub:asep_and_its_bethe_ansatz_eigenfunctions}

Let $k\ge1$ be an integer.
The $k$-particle ASEP is a continuous-time stochastic particle system
on $\Z$. Its state space
\begin{align}\label{tWeyl}
	\tWeyl{k}:=\{\x=(x_1<x_2<\ldots<x_k)\colon x_i\in\Z\}
\end{align}
consists of ordered $k$-tuples of distinct integers.\footnote{Our ordering of
distinct coordinates in $\x\in\tWeyl{k}$
differs from the one
for the space $\Weyl{k}$
(\S \ref{sub:spatial_and_spectral_variables}). This is done to better
reflect the notation for the ASEP
used in \cite{TW_ASEP1}, \cite{BorodinCorwinSasamoto2012}.
We continue to use this ordering in \S \ref{sec:application_to_six_vertex_model}~as~well.}

Each particle in the ASEP has an independent exponential clock with mean $1$.
When the clock of a particle rings, it immediately tries to
jump to the right with probability $\pasep>0$, or to the left with probability $\qasep>0$,
where $\pasep+\qasep=1$
(note that this $\qasep$ differs from the parameter $q$
in the $q$-Hahn system).
If the destination of the jump is already occupied, then the jump is blocked
(see Fig.~\ref{fig:asep}).
We assume that $\pasep<\qasep$ and denote $\tau:=\pasep/\qasep$, so $\tau\in(0,1)$
and $\pasep=\frac{\tau}{1+\tau}$, $\qasep=\frac{1}{1+\tau}$.
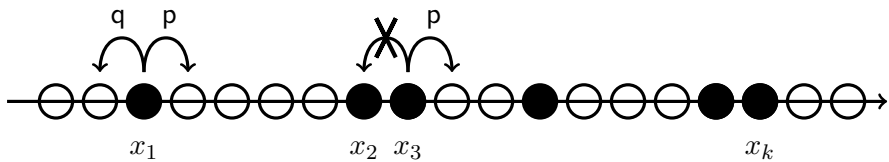
\begin{figure}[htbp]
	\begin{center}
	\begin{tikzpicture}
		[scale=1.3,very thick]
			\def\pt{.17}
			\def\ee{.1}
			\def\h{.45}
			\draw[->] (-.5,0) -- (8.5,0);
			\foreach \ii in {(0,0),(\h,0),(3*\h,0),(4*\h,0),(5*\h,0),(6*\h,0) ,(8*\h,0),(10*\h,0),(9*\h,0),(12*\h,0),(13*\h,0),(14*\h,0),(15*\h,0),(16*\h,0),(17*\h,0),(18*\h,0)}
			{
				\draw \ii circle(\pt);
			}
			\foreach \ii in {(2*\h,0),(7*\h,0),(11*\h,0),(15*\h,0),(8*\h,0),(16*\h,0)}
			{
				\draw[fill] \ii circle(\pt);
			}
			\node at (2*\h,-5*\ee) {$x_1$};
			\node at (7*\h,-5*\ee) {$x_{2}$};
			\node at (8*\h,-5*\ee) {$x_{3}$};
			\node at (16*\h,-5*\ee) {$x_{k}$};
		    \draw[->, very thick] (2*\h,.3) to [in=180,out=90] (2.5*\h,.65) to [in=90, out=0] (3*\h,.3) node [xshift=-7,yshift=20] {$\pasep$};
		    \draw[->, very thick] (2*\h,.3) to [in=0,out=90] (1.5*\h,.65) to [in=90, out=180] (1*\h,.3) node [xshift=7,yshift=20] {$\qasep$};
		    \draw[->, very thick] (8*\h,.3) to [in=0,out=90] (7.5*\h,.65) to [in=90, out=180] (7*\h,.3);
		    \draw[ultra thick] (7.5*\h,.65)--++(.1,.2)--++(-.2,-.4)--++(.1,.2)--++(-.1,.2)
		    --++(.2,-.4);
		    \draw[->, very thick] (8*\h,.3) to [in=180,out=90] (8.5*\h,.65) to [in=90, out=0] (9*\h,.3) node [xshift=-7,yshift=20] {$\pasep$};
	\end{tikzpicture}
	\end{center}
  	\caption{The ASEP particle system.}
  	\label{fig:asep}
\end{figure}

By $\tWc^{k}$ denote the space of all compactly supported functions
on $\tWeyl{k}$.
The Markov generator of the ASEP acts on
$\tWc^{k}$ as follows (here we use notation similar to $\n_i^{\pm}$ in
\S \ref{sub:_q_hahn_boson_process}):
\begin{align}\label{H_ASEP}
	(\Hasep f)(\x)=\sum_{i}
	\pasep(f(\x^{+}_{i})-f(\x))
	+\sum_{j}
	\qasep(f(\x^{-}_{j})-f(\x)),
\end{align}
where the sums are taken over all $i$ and $j$
for which the resulting configurations
$\x^{+}_{i}$ and $\x^{-}_{j}$, respectively,
belong to $\tWeyl{k}$.

Let $\Refl$ be the reflection operator
whose action on
$\tWc^{k}$
is given by \eqref{reflection_operator}.
One can readily see that
\begin{align}\label{ASEP_generator_PT}
	\Hasep=\Refl (\Hasep)^{\text{transpose}}\Refl^{-1},
\end{align}
because both the matrix transposition and the
conjugation by the
space reflection
$\Refl$
switch the roles of $\pasep$ and $\qasep$.
This means that the ASEP is PT-symmetric (as in Proposition \ref{prop:PT}), but with
the corresponding
$\mathfrak{m}$ being the identity operator.\footnote{The property $\mathfrak{m}\equiv 1$
is related to the fact that
Bernoulli measures are stationary for the ASEP \cite{Liggett1985}.
Also note the relation between
$\mathfrak{m}$ \eqref{stationary_measure}
and the stationary gaps distribution \eqref{full_q_geometric_distribution}
in the $q$-Hahn case.}

The application of the coordinate Bethe ansatz to the ASEP
dates back at least to the work of Gwa and Spohn
\cite{GwaSpohn1992}, and was
continued by Sch\"utz
\cite{Schutz1997exact}
and by Tracy and Widom
\cite{TW_ASEP1}.
One of the formulations of this integrability is that
the true
generator $\Hasep$ is equivalent (in the exact same sense
as in Proposition \ref{prop:CBA_qHahn}) to the following
free generator
\begin{align}\label{L_ASEP}
	(\Lasep u)(\x):=\sum_{i=1}^{k}[\nasep]_i u(\x),\qquad
	(\nasep u)(y):= \pasep u(y+1)+\qasep u(y-1)-u(y),\quad y\in\Z,
\end{align}
plus $k-1$ two-body boundary conditions
\begin{align}\label{BC_ASEP}
	\left.\Big(
		\pasep u(\x^{+}_{i})+\qasep u(\x^{-}_{i+1})-u(\x)
	\Big)\right\vert_{\x\in\Z^{k}\colon x_{i+1}=x_{i}+1}=0
\end{align}
for all $1\le i\le k-1$.

This reduction of the operator $\Hasep$
allows to construct its eigenfunctions
in the same way as in \S \ref{sub:coordinate_bethe_ansatz_integrability_and_the_left_eigenfunctions}
and \S \ref{sub:formal_stationary_vector_pt_symmetry_and_the_right_eigenfunctions}.
This leads to the following left (=~backward) eigenfunctions
\begin{align}\label{PsiASEP_bwd}
	\Psiasep_{\z}(\x):=\sum_{\sigma\in S(k)}
	\prod_{1\le B<A\le k}\frac{z_{\sigma(B)}-\tau z_{\sigma(A)}}{z_{\sigma(B)}-z_{\sigma(A)}}\prod_{j=1}^{k}\left(\frac{1+z_{\sigma(j)}}{1+z_{\sigma(j)}/\tau}\right)^{-x_j},\qquad \x\in\tWeyl{k}.
\end{align}
The right (=~forward) eigenfunctions are related to the above ones by a reflection,
and they are
\begin{align}\label{PsiASEP_fwd}
	(\Refl\Psiasep_{\z})(\x)=
	\sum_{\sigma\in S(k)}
	\prod_{1\le B<A\le k}\frac{z_{\sigma(A)}-\tau z_{\sigma(B)}}{z_{\sigma(A)}-z_{\sigma(B)}}
	\prod_{j=1}^{k}\left(\frac{1+z_{\sigma(j)}}{1+z_{\sigma(j)}/\tau}\right)^{x_j}
	,\qquad \x\in\tWeyl{k}.
\end{align}
Note also that (up to a constant factor plus a change of the spectral variables) the eigenfunctions
$\Refl\Psiasep_{\z}$ are obtained from $\Psiasep_{\z}$
by replacing $\tau\to\tau^{-1}$.
Let us summarize the properties of eigenfunctions:
\begin{proposition}\label{prop:ASEP_eigenfunctions}
	For all $z_1,\ldots,z_k\in\C\setminus\{-1,-\tau\}$,
	the function $\Psiasep_{\z}$ is an eigenfunction of the
	free ASEP generator $\Lasep$
	with the eigenvalue
	\begin{align}\label{ASEP_ev}
		-\frac{(1-\tau)^2}{1+\tau}\sum_{j=1}^{k}
		\frac{1}{ (1+z_j) (1+\tau/z_j)}
	\end{align}
	Moreover, $\Psiasep_{\z}$ satisfies the $(k-1)$ two-body boundary conditions \eqref{BC_ASEP}.
	Consequently,
	$\Psiasep_{\z}(\x)$ restricted to $\x\in\tWeyl{k}$
	is an eigenfunction of the
	true ASEP generator $\Hasep$ with the same eigenvalue \eqref{ASEP_ev}.

	Similarly, $\Refl\Psiasep_{\z}$
	is an eigenfunction of $(\Hasep)^{\text{transpose}}$
	with eigenvalue \eqref{ASEP_ev}.
\end{proposition}

Let us also write down the eigenfunctions \eqref{PsiASEP_bwd}--\eqref{PsiASEP_fwd}
in terms of the other spectral variables $\vxi$.
Consider the linear fractional map
\begin{align}\label{xinu_ASEP}
	\xiasep(z):=\frac{1+z}{1+z/\tau},\qquad
	\xiasep^{-1}(\xi)=
	-\frac{1-\xi}{1-\xi/\tau}.
\end{align}
For $z_j=\xiasep^{-1}(\xi_j)$,
one has
\begin{align*}
	\frac{z_B-\tau z_A}{z_B-z_A}=
	\frac{\Smasep(\xi_A,\xi_B)}{\xi_B-\xi_A},
\end{align*}
where
\begin{align}\label{Sm_ASEP}
	\Smasep(\xi_1,\xi_2):=\tau-(1+\tau)\xi_1+\xi_1\xi_2
	=(1+\tau)(\pasep-\xi_1+\qasep \xi_1\xi_2)
\end{align}
is the corresponding cross-term
(up to a constant factor it is the same as in \cite{TW_ASEP1}).

Define
\begin{align}\label{PhiASEP_bwd}
	\Phiasep_{\vxi}(\x):=
	\sum_{\sigma\in S(k)}\sgn(\sigma)
	\prod_{1\le B<A\le k}
	\Smasep(\xi_{\sigma(A)},\xi_{\sigma(B)})\prod_{j=1}^{k}\xi_{\sigma(j)}^{-x_j},
\end{align}
and note that
\begin{align}\label{PhiASEP_fwd}
	(\Refl\Phiasep_{\vxi})(\x)=
	(-1)^{\frac{k(k-1)}{2}}
	\sum_{\sigma\in S(k)}\sgn(\sigma)
	\prod_{1\le B<A\le k}
	\Smasep(\xi_{\sigma(B)},\xi_{\sigma(A)})\prod_{j=1}^{k}\xi_{\sigma(j)}^{x_j}.
\end{align}
Clearly, the eigenfunctions $\Psiasep_{\z}$ and
$\Phiasep_{\vxi}$ are related in the following way:
\begin{align*}
	\Psiasep_{\xiasep^{-1}(\vxi)}(\x)=(\Vand(\vxi))^{-1}\Phiasep_{\vxi}(\x),
	\qquad
	(\Refl\Psiasep_{\xiasep^{-1}(\vxi)})(\x)=(\Vand(\vxi))^{-1}(\Refl\Phiasep_{\vxi})(\x),
\end{align*}
where $\Vand(\vxi)$ is the Vandermonde determinant.

The eigenvalues \eqref{ASEP_ev} in the coordinates $\vxi$
take a simpler form
\begin{align}\label{ASEP_ev_xi}
	\evasep(\vxi):=\sum_{j=1}^{k}\big(\pasep\xi_j^{-1}+\qasep\xi_j-1\big).
\end{align}


\subsection{Relation to $q$-Hahn eigenfunctions} 
\label{sub:relation_to_q_hahn_boson_eigenfunctions}

There are two different ways
to specialize the parameters
$(q,\nu)$
in the
$q$-Hahn eigenfunctions
which lead to the
ASEP eigenfunctions.
Indeed, the ASEP cross-term $\Smasep$ \eqref{Sm_ASEP}
has only one nontrivial linear term, while the
corresponding $q$-Hahn system's cross-term
$\Sm$ \eqref{quadratic_S_matrix}
has two linear terms.
Thus, there are two ways to specialize
$\Sm$ into $\Smasep$
by forcing one of the linear terms in $\Sm$ to
vanish:\footnote{Clearly,
multiplying cross-terms by constants
is allowed as it reduces to
constant multiplicative factors
in front of the
eigenfunctions.}

\medskip

\par\noindent{\textbf{First degeneration}.}
If $q=1/\nu=\tau$ (where $\tau$ is the ASEP parameter),
	then
	\begin{align}\label{ASEP_first_degen}
		\big((1-q \nu)\Sm(\xi_1,\xi_2)\big)\vert_{q=\tau,\;\nu=\tau^{-1}}
		=\frac{1- \tau}{\tau}
		\big(
		\tau-(1+\tau)\xi_2+\xi_1\xi_2\big)=\frac{1- \tau}{\tau}\Smasep(\xi_2,\xi_1).
	\end{align}
	In this case note that for all $\n\in\Weyl{k}$,
	\begin{align}\label{st_ASEP}
		\st(\n)\vert_{q=\tau,\;\nu=\tau^{-1}}=
		\begin{cases}
			(-\tau)^{-k},&\text{if $n_1> \ldots>n_k$};\\
			0,&\text{otherwise}.
		\end{cases}
	\end{align}
	One can readily check that for any integers $x_1\le x_2\le \ldots\le x_k$
	we have
	(formally extending the definition of the
	ASEP eigenfunctions to weakly increasing sequences of indices)
	\begin{align}\label{Psi_qHahn_to_ASEP_main}
		\begin{array}{>{\displaystyle}rc>{\displaystyle}l}
			\Psiasep_{\z}(x_1,\ldots,x_k)
			&=&\Psil_{-\z}(x_k,\ldots,x_1)\vert_{q=\nu^{-1}=\tau},
			\\\rule{0pt}{15pt}
			(\Refl\Psiasep_{\z})(x_1,\ldots,x_k)
			\cdot\mathbf{1}_{x_1<\ldots<x_k}&=&(\tau^{-1}-1)^{-k}
			\Psir_{-\z}(x_k,\ldots,x_1)\vert_{q=\nu^{-1}=\tau}
			\\\rule{0pt}{15pt}
			\Phiasep_{\vxi}(x_1,\ldots,x_k)
			&=&(1-\tau^{-1})^{-\frac{k(k-1)}{2}}
			\\&&\hspace{10pt}\times
			\big((1-q\nu)^{\frac{k(k-1)}{2}}
			\Phil_{\vxi}(x_k,\ldots,x_1)\big)\vert_{q=\nu^{-1}=\tau}
			\\\rule{0pt}{15pt}
			(\Refl\Phiasep_{\vxi})(x_1,\ldots,x_k)
			\cdot\mathbf{1}_{x_1<\ldots<x_k}&=&
			(\tau^{-1}-1)^{-\frac{k(k-1)}{2}}
			\\&&\hspace{10pt}\times
			\big((1-q\nu)^{\frac{k(k-1)}{2}}
			\Phir_{\vxi}(x_k,\ldots,x_1)\big)\vert_{q=\nu^{-1}=\tau}.
		\end{array}
	\end{align}
	We will mainly work in the $\z$ spectral variables,
	the third and forth
	formulas for $\Phiasep_{\vxi}$ above are
	given for illustration.

\medskip

\par\noindent{\textbf{Second degeneration}.}
If $q=\nu=1/\tau$,
	then
	\begin{align}\label{ASEP_second_degen}
		\Sm(\xi_1,\xi_2)\vert_{q=\nu=\tau^{-1}}
		=\frac{1}{1+\tau}
		\big(
		\tau-(1+\tau)\xi_1+\xi_1\xi_2\big)=\frac{1}{1+\tau}\Smasep(\xi_1,\xi_2).
	\end{align}
	In this case $\st(\n)=1$ for all $\n\in\Weyl{k}$,
	so this second way of degeneration
	does not lead to any formulas
	with strictly ordered spatial variables.
	Therefore, our main results (Plancherel formulas, spectral biorthogonality)
	for $q=\nu=\tau^{-1}$
	will not directly correspond to formulas for the ASEP particle system.
	
	On the other hand, for $q=\nu=\tau^{-1}$
	one also could write down formulas relating the $q$-Hahn and the ASEP
	eigenfunctions
	similarly to \eqref{Psi_qHahn_to_ASEP_main}.
	Thus, our main results would imply certain other identities
	for the ASEP eigenfunctions, but with \emph{weakly ordered}
	spatial coordinates. One could readily
	write down such statements
	because the degeneration $q=\nu$ does not involve
	problems with integration contours as in
	the case $q=\nu^{-1}$
	(for the latter cf. \S \ref{sub:plancherel_formulas_for_asep} below).
	The role of some of these identities
	remains unclear, and
	we will not pursue this direction
	except for
	symmetrization identities (see
	\S \ref{sub:tracy_widom_symmetrization_identities}
	below).

\medskip

See Fig.~\ref{fig:big_scheme} on how the ASEP eigenfunctions
fit into the general picture of
Hall-Littlewood type eigenfunctions
of Bethe ansatz solvable particle systems.


\subsection{Spectral biorthogonality of the ASEP eigenfunctions} 
\label{sub:ASEP_spectral_biorthogonality}

Let us obtain an ASEP analogue of the
spectral biorthogonality statement of
Theorem \ref{thm:spectral_biorthogonality}.
Let $\widetilde\ga_{-1}$
denote a small positively oriented
closed circle around $(-1)$ which does not encircle $(-\tau)$.
Let also $\widetilde\ga_{-1}'$ be circle around
$(-1)$
containing $\widetilde\ga_{-1}$
such that for all $z\in\widetilde\ga_{-1}$,
$w\in\widetilde\ga_{-1}'$,
one has
\begin{align*}
	\left|\frac{1+z}{1+z/\tau}\right|<\left|\frac{1+w}{1+w/\tau}\right|.
\end{align*}
The existence of $\widetilde\ga_{-1}'$ can be established similarly to
\S \ref{sub:contour_}.

\begin{theorem}\label{thm:ASEP_biorthogonality}
	Let $F(\z)$ be
	a function
	such that
	for $M$ large enough,
	\begin{align*}
		\Vand(\z)F(\z)\prod_{j=1}^{k}\left(\frac{1+z_j}{1+z_j/\tau}\right)^{-M}
	\end{align*}
	is holomorphic in the closed exterior of the contour $\widetilde\ga_{-1}$
	(including $\infty$).
	Let $G(\w)$ be such that $\Vand(\w)G(\w)$
	is holomorphic in the closed
	region between
	$\widetilde\ga_{-1}$
	and $\widetilde\ga_{-1}'$.
	Then
	\begin{align*}
		&\sum_{\x\in\tWeyl{k}}
		\left(
		\oint_{\widetilde\ga_{-1}}\ldots\oint_{\widetilde\ga_{-1}}
		\frac{d\z}{(2\pi\i)^{k}}
		(\Refl\Psiasep_{\z})(\x)\Vand(\z)F(\z)\right)
		\left(
		\oint_{\widetilde\ga_{-1}}\ldots\oint_{\widetilde\ga_{-1}}
		\frac{d\w}{(2\pi\i)^{k}}
		\Psiasep_{\w}(\x)\Vand(\w)G(\w)
		\right)
		\\&\hspace{11pt}=
		\oint_{\widetilde\ga_{-1}}\ldots\oint_{\widetilde\ga_{-1}}
		\frac{d\z}{(2\pi\i)^{k}}
		(-1)^{\frac{k(k-1)}{2}}
		\prod_{j=1}^{k}\frac{(1+z_j)(1+z_j/\tau)}{1-1/\tau}
		\prod_{A\ne B}(z_A-\tau z_B)\sum_{\sigma\in S(k)}
		\sgn(\sigma)
		F(\z)G(\sigma \z).
	\end{align*}
\end{theorem}
\begin{proof}
	If the test functions $F(\z)$ and $G(\w)$ are Laurent polynomials
	in
	$\frac{1+z_j}{1+z_j/\tau}$
	and
	$\frac{1+w_j}{1+w_j/\tau}$,
	respectively, then the statement readily follows from
	the spectral biorthogonality for the $q$-Hahn eigenfunctions.
	Indeed, for Laurent polynomials one can let the integration contours
	in Theorem \ref{thm:spectral_biorthogonality}
	to be
	small positively oriented closed circles $\ga_k$ around $1$ (see Remark \ref{rmk:spectral_biorthogonality_over_1_nuin}).
	Then,
	using \eqref{Psi_qHahn_to_ASEP_main}, we can
	specialize
	$q=\tau$, $\nu=1/\tau$, and this will turn
	contours $\ga_k$
	into our contours $\widetilde\ga_{-1}$.
	Note also that the
	negation of all variables $\z$
	and $\w$ introduces an extra factor of $(-1)^{k}$ in the right-hand side.

	For more general test functions satisfying the above properties, one sees that the
	series in $\x$ in the spectral biorthogonality identity
	converges similarly to Proposition \ref{prop:convergence_of_biorthogonality}.
	Then the desired claim follows by
	approximating test functions by Laurent polynomials,
	by virtue of Runge's theorem.
\end{proof}

One can also formulate an analogue
of Theorem \ref{thm:spectral_biorthogonality_xi}
concerning the eigenfunctions $\Phiasep_{\vxi}$,
but we will not write it down.


\subsection{Plancherel formulas for the ASEP} 
\label{sub:plancherel_formulas_for_asep}

By $\Cc^{k}_{\z;\mathrm{ASEP}}$ let us denote the space of symmetric Laurent
polynomials in $\frac{1+z_j}{1+z_j/\tau}$, $1\le j\le k$.
Define the following direct and (candidate)
inverse transforms:
\begin{align}
	\label{Pld_asep}
	(\Plde^{\mathrm{ASEP}} f)(\z)&:=
	\sum_{\x\in\tWeyl{k}}f(\x)(\Refl\Psiasep_{\z})(\x),\\
	\label{Pli_asep}
	(\Plie^{\mathrm{ASEP}} G)(\x)&:=
	\oint_{\widetilde\ga_{-1}}\frac{dz_1}{2\pi\i}
	\ldots
	\oint_{\widetilde\ga_{-1}}\frac{dz_k}{2\pi\i}
	\prod_{B<A}\frac{z_A-z_B}{z_A-\tau z_B}
	\prod_{j=1}^{k}
	\frac{1-1/\tau}{(1+z_j)(1+z_j/\tau)}
	\left(\frac{1+z_j}{1+z_j/\tau}\right)^{-x_j}
	G(\z).
\end{align}
The contours $\widetilde\ga_{-1}$ are small positive
circles around $(-1)$.
The operator $\Plde^{\mathrm{ASEP}}$
takes functions from $\tWc^{k}$ to
$\Cc^{k}_{z;\mathrm{ASEP}}$,
and the operator $\Plie^{\mathrm{ASEP}}$ acts in the opposite direction.

Since all integration contours in \eqref{Pli_asep}
are the same and $G$ is symmetric, we can symmetrize
that formula (cf. the proof of Proposition \ref{prop:nesting_unnesting}), and rewrite the ASEP
inverse transform as
\begin{align}\label{Pli_asep_symmetrized}
	(\Plie^{\mathrm{ASEP}} G)(\x)&=
	\oint_{\widetilde\ga_{-1}}
	\ldots
	\oint_{\widetilde\ga_{-1}}
	d\Plm^{(\tau)}_{(1^{k})}(\z)
	\prod_{j=1}^{k}
	\frac{1-1/\tau}{(1+z_j)(1+z_j/\tau)}
	\Psiasep_{\z}(\x)
	G(\z),
\end{align}
where $\displaystyle d\Plm_{(1^{k})}^{(\tau)}(\z)$
is the Plancherel measure
(Definition \ref{def:Plancherel_measure}, see also \eqref{dmu_large})
with $q$ replaced by $\tau$ as it should be under our degeneration.

The results of \cite{TW_ASEP1} imply the 
spatial Plancherel formula for the ASEP:

\begin{theorem}[The spatial Plancherel formula]\label{thm:spatial_Plancherel_ASEP}
	The map
	$\Plspatiale^{\mathrm{ASEP}}=\Plie^{\mathrm{ASEP}}\Plde^{\mathrm{ASEP}}$ acts as the identity
	operator on $\tWc^{k}$. Equivalently,
	for any $\x,\y\in\tWeyl{k}$,
	\begin{align}
		\label{spatial_plancherel_formula_ASEP_identity}
		\oint_{\widetilde\ga_{-1}}\ldots\oint_{\widetilde\ga_{-1}}
		d\Plm^{(\tau)}_{(1^{k})}(\z)\prod_{j=1}^{k}\frac{1-1/\tau}{(1+z_j)(1+z_j/\tau)}
		\Psiasep_{\z}(\x)(\Refl\Psiasep_{\z})(\y)=\mathbf{1}_{\x=\y}.
	\end{align}
\end{theorem}
\begin{proof}
	We derive the statement of the theorem from
	\cite[Theorem 2.1]{TW_ASEP1}.
	The desired identity 
	can be rewritten
	using \eqref{Pli_asep}
	in the following form:
	\begin{align}\label{ASEP_Plancherel_desymm}
		\oint_{\widetilde\ga_{-1}}\frac{dz_1}{2\pi\i}
		\ldots
		\oint_{\widetilde\ga_{-1}}\frac{dz_k}{2\pi\i}
		\prod_{B<A}\frac{z_A-z_B}{z_A-\tau z_B}
		\prod_{j=1}^{k}
		\frac{1-1/\tau}{(1+z_j)(1+z_j/\tau)}
		\left(\frac{1+z_j}{1+z_j/\tau}\right)^{-x_j}
		(\Refl\Psiasep_{\z})(\y)=\mathbf{1}_{\x=\y}.
	\end{align}
	For $\sigma\in S(k)$, denote by
	\begin{align*}
		\Across_\si(\z):=\prod_{B<A}\frac{z_A-z_B}{z_A-\tau z_B}
		\prod_{B<A}\frac{z_{\sigma(A)}-\tau z_{\sigma(B)}}{z_{\sigma(A)}-z_{\sigma(B)}}=
		\prod_{B<A\colon \sigma(B)>\sigma(A)}\left(-\frac{z_{\sigma(A)}-\tau z_{\sigma(B)}}{z_{\sigma(B)}-\tau z_{\sigma(A)}}\right)
	\end{align*}
	the term which arises from the cancellation of
	two double products, one inside the integral, and another
	coming from $(\Refl\Psiasep_{\z})$
	(the minus sign is the ratio of two expressions without $\tau$).
	Then \eqref{ASEP_Plancherel_desymm} takes the form
	\begin{align}\label{ASEP_Plancherel_desymm2}
		\sum_{\sigma\in S(k)}
		\oint_{\widetilde\ga_{-1}}\ldots\oint_{\widetilde\ga_{-1}}
		\Across_{\sigma}(\z)
		\prod_{j=1}^{k}\frac{1-1/\tau}{(1+z_j)(1+z_j/\tau)}
		\left(\frac{1+z_j}{1+z_j/\tau}\right)^{-x_j+y_{\sigma^{-1}(j)}}\frac{dz_j}{2\pi\i}
		=\mathbf{1}_{\x=\y}.
	\end{align}
	Under the change of variables $\xiasep$ \eqref{xinu_ASEP},
	we have
	\begin{align*}
		\Across_{\sigma}(\xiasep^{-1}(\vxi))=\prod_{A<B\colon \sigma(A)>\sigma(B)}S_{\sigma(A)\sigma(B)},
		\qquad
		S_{\al\be}=-\frac{\pasep+\qasep \xi_\al\xi_\be-\xi_\al}{\pasep+\qasep \xi_\al\xi_\be-\xi_\be}
		=-\frac{\Smasep(\xi_\al,\xi_\be)}{\Smasep(\xi_\be,\xi_\al)},
	\end{align*}
	so $\Across_{\sigma}$ is exactly the same product over inversions in the permutation $\sigma$
	as in \cite[\S2]{TW_ASEP1}.
	Thus, we may rewrite \eqref{ASEP_Plancherel_desymm2} in the coordinates
	$\vxi$ as follows:
	\begin{align}\label{ASEP_Plancherel_TW}
		\sum_{\sigma\in S(k)}
		\oint_{\widetilde\ga_{0}}\ldots\oint_{\widetilde\ga_{0}}
		\Across_{\sigma}(\xiasep^{-1}(\vxi))
		\prod_{j=1}^{k}
		\xi_{\sigma(j)}^{-x_{\sigma(j)}+y_{j}-1}\frac{d\xi_j}{2\pi\i}
		=\mathbf{1}_{\x=\y},
	\end{align}
	where the integration contours $\widetilde\ga_{0}$ are now small positively oriented
	closed circles around $0$.
	This is exactly the Tracy-Widom's formula
	\cite[(2.3)]{TW_ASEP1} for time $t=0$
	up to swapping $\x\leftrightarrow\y$.
\end{proof}

\begin{remark}[Transition probabilities for the ASEP]
	\label{rmk:ASEP_transition_probabilities}
	Identity \eqref{ASEP_Plancherel_TW} to which we have reduced the desired 
	identity of Theorem \ref{thm:spatial_Plancherel_ASEP}
	is a $t=0$ version of the general formula for the
	transition probabilities in the ASEP 
	\cite[(2.3)]{TW_ASEP1}:
	\begin{align*}
		\mathbb{P}(\text{from $\x$ to $\y$ in time $t$})=
		\sum_{\sigma\in S(k)}
		\oint_{\widetilde\ga_{0}}\ldots\oint_{\widetilde\ga_{0}}
		\Across_{\sigma}(\xiasep^{-1}(\vxi))
		e^{t\cdot \evasep(\vxi)}\prod_{j=1}^{k}
		\xi_{\sigma(j)}^{-x_{\sigma(j)}+y_{j}-1}\frac{d\xi_j}{2\pi\i}.
	\end{align*}
	These transition probabilities
	solve the forward Kolmogorov equation
	for the ASEP (cf. Proposition~\ref{prop:fwd_eqn}
	for its $q$-Hahn generalization). 
	This explains the appearance of the
	eigenvalues \eqref{ASEP_ev_xi} under the integral.
\end{remark}
\begin{remark}
	\label{rmk:update_ASEP_statement}
	The ``proof'' of Theorem~\ref{thm:spatial_Plancherel_ASEP}
	given in the previous version of this paper claimed to
	deduce \eqref{spatial_plancherel_formula_ASEP_identity}
	from the spatial biorthogonality of the $q$-Hahn eigenfunctions 
	(Corollary \ref{cor:C_biorthogonality})
	by plugging $\nu=1/q$ into the small contour
	formula (with all integration contours equal to $\ga_k$ which are small circles
	around~$1$). 
	Indeed, identity \eqref{spatial_plancherel_formula_ASEP_identity} looks
	as if one takes the $q$-Hahn small contour formula, 
	removes all terms corresponding to 
	partitions $\lambda\ne (1^k)$,
	and then plugs in $\nu=1/q$, $q=\tau$.
	Theorem \ref{thm:spatial_Plancherel_ASEP} 
	\emph{a posteriori} implies that under this specialization, the contribution
	of all additional terms with $\lambda\ne (1^k)$
	vanishes. 
	
	First, observe that the substitution $\nu=1/q$ before the integration
	might change the value of the integral because 
	of the factors of the form $\frac{1}{1-q \nu w_i}$ in the integrand for $\lambda\ne (1^k)$.
	Before the substitution $\nu=1/q$ the residue at $w_i=(q \nu)^{-1}$ was not picked
	while after the substitution we have $1-q \nu w_i=1-w_i$, so this
	factor adds an extra pole inside the integration contour.
	
	With the agreement that the substitution $\nu=1/q$ 
	occures after the integration, 
	the ``proof'' of 
	Theorem \ref{thm:spatial_Plancherel_ASEP} 
	presented in an earlier version of this paper
	asserted a stronger statement:
	For each individual
	$\lambda\ne(1^k)$ and any two permutations
	$\sigma, \omega\in S(k)$ 
	(coming from $\Psi_{\vec z}^\ell$ and $\Psi_{\vec z}^{r}$, respectively)
	the corresponding term vanishes after setting $\nu=1/q$.
	This assertion is wrong.
	(We are grateful to Yier Lin for pointing this out to us.)

	For example, take
	$\vec x=(10,9,8,7,6,5)$ and $\vec y=(5,4,3,2,1,0)$.
	The summand in the integrand in the small contour formula (before setting $q=1/\nu=\tau$)
	corresponding to $\lambda=(3,2,1)$, and permutations
	$\sigma=321546$ and $\omega=645123$ has the form
	\begin{align*}
		&
		\mathrm{const}\cdot
		\frac
		{(1-\nu  q w_1)^7 (1-\nu  q w_2)^3}
		{(1-w_1)^7 (1-w_2)^3 (1-w_3)}
		\\&\hspace{20pt}\times
		\frac
		{
			(q w_1-w_2) 
			\left(q^2 w_1-w_2\right)^2 
			\left(q^3 w_1-w_2\right) 
			\left(q^2 w_1-w_3\right) 
			\left(q^3 w_1-w_3\right) 
			(q w_2-w_3) 
			\left(q^2 w_2-w_3\right) 
		}
		{
			(w_1-w_2) 
			(w_1-w_3) 
			(w_2-w_3) 
			(q w_2-w_1)^2 
			\left(q^2 w_2-w_1\right) 
			(q w_3-w_1) 
			(q w_3-w_2) 
		}
		\\&\hspace{20pt}\times
		f_1(w_1)f_2(w_2)f_3(w_3).
	\end{align*}
	Here $f_1(w_1)$ is independent of $w_2,w_3$ and has no 
	zeroes or poles at $w_1=1$ and $w_1=1/(q\nu)$, 
	and similarly for $f_2(w_2)$ and $f_3(w_3)$.
	One can check that the residue of this term at 
	$w_3=1$, $w_2=1$, and $w_1=1$\footnote{In this order. Note that the
	result of the integration depends on the order of taking the residues
	for individual summands due to the presence of the factors of the form $w_i-w_j$
	in the denominators. These factors cancel out after summing over all permutations $\sigma,\omega$, and 
	each summand indexed by $\lambda$ is independent of the order of integration.}
	does not vanish when setting $q=1/\nu$.

	We thus have replaced the incorrect ``proof'' of 
	Theorem~\ref{thm:spatial_Plancherel_ASEP}
	by its reduction to the earlier result of Tracy and Widom \cite{TW_ASEP1}.
	We have also removed a spatial biorthogonality
	statement for the XXZ eigenfunctions 
	present in the previous version of this paper in
	\S\ref{sub:formulas_for_svd_1_and_complex_q_dilp_}
	because it was based on a similar direct degeneration which is incorrect.
\end{remark}
\begin{remark}
	Let us discuss another (possibly related)
	subtlety in the spatial biorthogonality of the ASEP eigenfunctions
	concerning
	contributions of individual permutations $\sigma$
	to formulas \eqref{ASEP_Plancherel_desymm2} and \eqref{ASEP_Plancherel_TW}.
	Recall that
	in the $q$-Hahn setting we have proved a similar identity
	\eqref{Plancherel_identity_representation}
	by showing that the contribution of each individual
	$\sigma\ne id$ vanishes (Lemma \ref{lemma:Plspatial_permutation_vanishing}).

	This is not the case for the ASEP:
	For example, if $k=3$, $\sigma=(321)$ is the transposition $1\leftrightarrow 3$,
	$\x=(0,1,2)$, $\y=(-2,-1,0)$, then
	\begin{align*}
		\oint_{\widetilde\ga_{-1}}\ldots\oint_{\widetilde\ga_{-1}}
		\Across_{\sigma}(\z)
		\prod_{j=1}^{3}\frac{1-1/\tau}{(1+z_j)(1+z_j/\tau)}
		\left(\frac{1+z_j}{1+z_j/\tau}\right)^{-x_j+y_{\sigma^{-1}(j)}}\frac{dz_j}{2\pi\i}
		=-\frac{(1+\tau)^{2}}{\tau^4},
	\end{align*}
	which is nonzero.

	In fact, the above
	contribution is compensated by
	the summand corresponding to
	$\sigma=(312)$, which has the opposite sign.
	The proof of
	\eqref{ASEP_Plancherel_TW}
	in \cite{TW_ASEP1}
	employs nontrivial combinatorics
	to determine such cancellations in general.
\end{remark}

\begin{theorem}[The spectral Plancherel formula]\label{thm:spectral_Plancherel_ASEP}
	The map $\Plspectrale^{\mathrm{ASEP}}=\Plde^{\mathrm{ASEP}}\Plie^{\mathrm{ASEP}}$
	acts as the identity
	operator on $\Cc^{k}_{\z;\mathrm{ASEP}}$.
\end{theorem}
\begin{proof}
	This follows from the spectral biorthogonality
	(Theorem \ref{thm:ASEP_biorthogonality})
	in a way similar to Theorem \ref{thm:spectral_biorthogonality}.
\end{proof}

\begin{remark}
	One could equip spaces $\tWc^{k}$
	and $\Cc^{k}_{z;\mathrm{ASEP}}$
	with suitable bilinear pairings
	which correspond to each other under the
	spatial and the spectral Plancherel
	isomorphism theorems (\ref{thm:spatial_Plancherel_ASEP}
	and \ref{thm:spectral_Plancherel_ASEP}), but we will not write
	these pairings down.
\end{remark}

\subsection{Tracy-Widom symmetrization identities} 
\label{sub:tracy_widom_symmetrization_identities}

The Tracy-Widom symmetrization identity for the ASEP
\cite[(1.6)]{TW_ASEP1}
can be obtained as a corollary of our main results.
In an exactly similar manner one could
get a generalization
corresponding to the step Bernoulli
(i.e., half-stationary)
initial condition, first obtained as
\cite[(9)]{TW_ASEP4}.
Both combinatorial
identities served as crucial steps
towards the asymptotic analysis of the ASEP,
see
\cite{TW_ASEP4},
\cite{BorodinCorwinSasamoto2012},
\cite{CorwinQuastel2013},
\cite{BCG6V}.
To shorten the formulas, we will only focus on the step case,
i.e., on the identity \cite[(1.6)]{TW_ASEP1}
which is equivalent to
\begin{align}\label{TW_16}
	\begin{array}{l>{\displaystyle}l}
		&\sum_{\sigma\in S(k)}
		\prod_{i<j}\frac{\Smasep(\xi_{\sigma(i)},\xi_{\sigma(j)})}
		{\xi_{\sigma(j)}-\xi_{\sigma(i)}}
		\frac{\xi_{\sigma(2)}\xi_{\sigma(3)}^{2}\ldots \xi_{\sigma(k)}^{k-1}}
		{(1-\xi_{\sigma(1)}\xi_{\sigma(2)}\ldots \xi_{\sigma(k)})
		(1-\xi_{\sigma(2)}\ldots \xi_{\sigma(k)})\ldots
		(1-\xi_{\sigma(k)})}
		\\&\hspace{300pt}=
		\frac{\tau^{\frac{k(k-1)}2}}{\prod_{j=1}^{k}(1-\xi_j)}.
	\end{array}
\end{align}
This identity can be derived from the spectral
Plancherel formula for the ASEP (Theorem \ref{thm:spectral_Plancherel_ASEP})
in a way similar to \S \ref{sub:symmetrization_formula},
but a shorter way to achieve it is to simply specialize the existing $(q,\nu)$
identity \eqref{q_nu_symmetrization_identity_step}.
We will consider both ways to specialize the parameters
as in \S \ref{sub:relation_to_q_hahn_boson_eigenfunctions}.

\medskip

\par\noindent{\textbf{First degeneration}.}
For $q=1/\nu=\tau$, the sum over $\n$
in \eqref{q_nu_symmetrization_identity_step} becomes
(see \eqref{st_ASEP})
\begin{align*}
	\sum_{n_1\ge \ldots\ge n_k\ge0}
	\st(\n)
	\prod_{j=1}^{k}
	\xi_{\sigma(j)}^{n_j}
	&=
	(-\tau)^{-k}
	\sum_{n_1> \ldots> n_k\ge0}
	\prod_{j=1}^{k}
	\xi_{\sigma(j)}^{n_j}
	\\&=
	(-\tau)^{-k}
	\frac{\xi_{\sigma(1)}^{k-1}\xi_{\sigma(2)}^{k-2}\ldots\xi_{\sigma(k-1)}}
	{(1-\xi_{\sigma(1)})(1-\xi_{\sigma(1)}\xi_{\sigma(2)})
	\ldots
	(1-\xi_{\sigma(1)}\ldots\xi_{\sigma(k)})}.
\end{align*}
Using \eqref{ASEP_first_degen}, we can readily rewrite \eqref{q_nu_symmetrization_identity_step} as
\eqref{TW_16} (note that it involves multiplying the permutation $\sigma$
over which we are summing
by the permutation $j\leftrightarrow k+1-j$, $1\le j\le k$).

\medskip

\par\noindent{\textbf{Second degeneration}.}
For $q=\nu=1/\tau$, we have in \eqref{q_nu_symmetrization_identity_step}:
\begin{align*}
	\sum_{n_1\ge \ldots\ge n_k\ge0}
	\st(\n)
	\prod_{j=1}^{k}
	\xi_{\sigma(j)}^{n_j}
	&=
	\sum_{n_1\ge \ldots\ge n_k\ge0}
	\prod_{j=1}^{k}
	\xi_{\sigma(j)}^{n_j}
	\\&=
	\frac{1}
	{(1-\xi_{\sigma(1)})(1-\xi_{\sigma(1)}\xi_{\sigma(2)})
	\ldots
	(1-\xi_{\sigma(1)}\ldots\xi_{\sigma(k)})}.
\end{align*}
Therefore, with the help of \eqref{ASEP_second_degen}, we can rewrite \eqref{q_nu_symmetrization_identity_step} as
\begin{align*}
	&\displaystyle
	\sum_{\sigma\in S(k)}
	\prod_{B<A}
	\frac{\Smasep(\xi_{\sigma(B)},\xi_{\sigma(A)})}
	{\xi_{\sigma(A)}-\xi_{\sigma(B)}}
	\frac{1}
	{(1-\xi_{\sigma(1)})(1-\xi_{\sigma(1)}\xi_{\sigma(2)})
	\ldots
	(1-\xi_{\sigma(1)}\ldots\xi_{\sigma(k)})}
	=
	\prod_{j=1}^{k}
	\frac{1}{1-\xi_j},
\end{align*}
which is equivalent to identity \cite[(1.7)]{TW_ASEP1},
and can be obtained from \eqref{TW_16} by interchanging $\pasep\leftrightarrow\qasep$
and renaming $\xi_i\to\xi_{k+1-i}^{-1}$.

We see that while
all results under
the first degeneration $q=1/\nu=\tau$
are directly relevant to the ASEP (as discussed earlier in this section),
the second degeneration
$q=\nu=1/\tau$
provides a symmetrization identity
for the ASEP, too (however, the two identities thus obtained are
equivalent to each other).
A similar effect can be observed
in identities
like \cite[(9)]{TW_ASEP4}
(corresponding
to
the step Bernoulli
initial condition) which follow from the
$q$-Hahn level identity \eqref{q_nu_symmetrization_identity}.



\section{Application to six-vertex model and XXZ spin chain} 
\label{sec:application_to_six_vertex_model}

In this section we briefly explain how
the eigenfunctions of the conjugated $q$-Hahn operator
are related to eigenfunctions
of the transfer matrix of the (asymmetric) six-vertex model.
The six-vertex model is one of the most well-known solvable models
in statistical physics.
Its first solution was obtained by
Lieb
\cite{Lieb67}.
See also the book by Baxter
\cite{baxter2007exactly}
and the lecture notes by Reshetikhin
\cite{reshetikhin2010lectures}
for details and perspectives.
In \S \ref{sub:xxz} we describe the connection
to eigenfunctions of the Heisenberg XXZ quantum spin chain
(which is a certain degeneration of the
six-vertex model).
Then we discuss how our main results are
applied to eigenfunctions of the six-vertex model and the XXZ spin chain.

\subsection{Transfer matrix and its eigenfunctions} 
\label{sub:transfer_matrix_and_its_eigenfunctions}

We will work only in ``infinite volume'' (i.e., on the lattice $\Z$),
which is similar to the setup of the rest of the present paper.
In the \emph{line representation} of the six-vertex model, configurations
$\x\in \tWeyl{k}$ (see \eqref{tWeyl}) encode locations of vertical lines
in a horizontal slice of the infinite square grid. We assume that there are
$k$ such vertical lines, see Fig.~\ref{fig:6v1}
(one of the properties of the six-vertex model is that
the number of vertical lines is preserved).
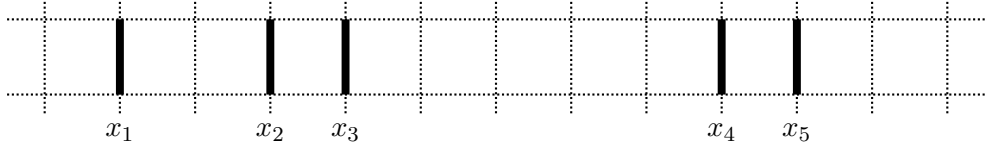
\begin{figure}[htb]
	\begin{center}
		\begin{tikzpicture}
		[scale=1, thick]
		\def\e{.5}
		\foreach \i in {-6, ..., 6}
		{
			\draw[densely dotted] (\i,-1.5*\e) -- (\i,1.5*\e);
		}
		\draw[densely dotted] (-6.5,\e) -- (6.5,\e);
		\draw[densely dotted] (-6.5,-\e) -- (6.5,-\e);
		\foreach \i in {-5,-3,-2,3,4}
		{
			\draw[line width=3] (\i,-\e) -- (\i,\e);
		}
		\node at (-5,-2*\e) {$x_1$};
		\node at (-3,-2*\e) {$x_2$};
		\node at (-2,-2*\e) {$x_3$};
		\node at (3,-2*\e) {$x_4$};
		\node at (4,-2*\e) {$x_5$};
		\end{tikzpicture}
	\end{center}
  	\caption{A configuration of vertical lines in one ``slice'' of the six-vertex model.}
  	\label{fig:6v1}
\end{figure}
Define the \emph{Boltzmann weights}
at each vertex of the square grid depending on the
configuration of lines at this vertex, see Fig.~\ref{fig:6v2}.
\begin{figure}[htb]
	\begin{center}
		\begin{tikzpicture}
			[scale=1, thick]
			\def\e{.8}
			\def\p{.15}
			\foreach \i in {0,2.5,5,7.5,10,12.5}
			{
				\draw[densely dotted] (\i,-\e) -- (\i,\e);
				\draw[densely dotted] (\i-\e,0) -- (\i+\e,0);
			}
			\draw[line width=3] (2.5,-\e)--++(0,\e-\e*\p)--++(\p*\e,\p*\e)--++(\e-\e*\p,0);
			\draw[line width=3] (2.5,\e)--++(0,-\e+\e*\p)--++(-\p*\e,-\p*\e)--++(-\e+\e*\p,0);
			\draw[line width=3] (5,-\e)--++(0,2*\e);
			\draw[line width=3] (7.5-\e,0)--++(2*\e,0);
			\draw[line width=3] (10,-\e)--++(0,\e-\e*\p)--++(\p*\e,\p*\e)--++(\e-\e*\p,0);
			\draw[line width=3] (12.5,\e)--++(0,-\e+\e*\p)--++(-\p*\e,-\p*\e)--++(-\e+\e*\p,0);
			\node at (0,-1.5*\e) {$a_1$};
			\node at (2.5,-1.5*\e) {$a_2$};
			\node at (5,-1.5*\e) {$b_1$};
			\node at (7.5,-1.5*\e) {$b_2$};
			\node at (10,-1.5*\e) {$c_1$};
			\node at (12.5,-1.5*\e) {$c_2$};
		\end{tikzpicture}
	\end{center}
  	\caption{Boltzmann weights that depend on the configuration of
  	lines at a vertex.}
  	\label{fig:6v2}
\end{figure}
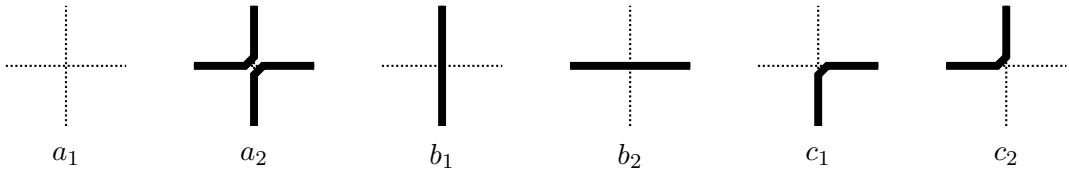
Here $a_1,a_2,b_1,b_2,c_1$, and $c_2$
are some \emph{positive} real parameters.
All other configurations of lines at a vertex are forbidden.

Broadly speaking, the \emph{six-vertex model} on a \emph{finite} subset of the
square grid
assigns the weight
$W(\mathsf{conf})=a_1^{\#[a_1]}
a_2^{\#[a_2]}
b_1^{\#[b_1]}
b_2^{\#[b_2]}
c_1^{\#[c_1]}
c_2^{\#[c_2]}$
to each allowed configuration $\mathsf{conf}$ of lines inside this subset
(with certain specified boundary conditions).
Here $\#[a_1]$
is the number of vertices of type $a_1$ inside this finite
subset, etc.
Since weights of all configurations are positive,
this can be interpreted as
a probabilistic model
(i.e., one can speak about random configuration of lines).

We will not define the six-vertex model
in our infinite setting
(see \cite{BCG6V}). Instead we focus
on the corresponding transfer matrix which is well-defined
for configurations such that there are finitely many
lines on each horizontal slice. We put
$a_1=1$, as this is the weight
of the empty line configuration,
the only one that
repeats infinitely often.

The rows and columns of the
\emph{transfer matrix}
$T_k(\x,\y)$
are indexed by line configurations
$\x,\y\in\tWeyl{k}$.
We assume that these configurations are
put on top of each other (see the picture below).
Define
\begin{align*}
	T_k(\x,\y):=\begin{cases}
		a_2^{\#[a_2]}
		b_1^{\#[b_1]}
		b_2^{\#[b_2]}
		c_1^{\#[c_1]}
		c_2^{\#[c_2]},
		&\parbox{.55\textwidth}{if there is a configuration of horizontal lines
		connecting $\x$ to $\y$ (such a configuration is unique if it exists);}\\
		0,&\text{otherwise}.
	\end{cases}
\end{align*}
For example, the configuration of lines on
Fig.~\ref{fig:6v3} represents a
particular element of the transfer
matrix equal to $T_k(\x,\y)=a_2b_1b_2^{3}c_1^{3}c_2^{3}$.
\begin{figure}[htb]
	\begin{center}
	\begin{tikzpicture}
		[scale=1, thick]
		\def\e{.5}
		\def\p{.2}
		\foreach \i in {-6, ..., 6}
		{
			\draw[densely dotted] (\i,-1.5*\e) -- (\i,3.5*\e);
		}
		\draw[densely dotted] (-6.5,\e) -- (6.5,\e);
		\draw[densely dotted] (-6.5,-\e) -- (6.5,-\e);
		\draw[densely dotted] (-6.5,3*\e) -- (6.5,3*\e);
		\node at (-5,-2*\e) {$x_1$};
		\node at (-3,-2*\e) {$x_2$};
		\node at (-2,-2*\e) {$x_3$};
		\node at (3,-2*\e) {$x_4$};
		\node at (4,-2*\e) {$x_5$};
		\node at (-4,4*\e) {$y_1$};
		\node at (-2,4*\e) {$y_2$};
		\node at (1,4*\e) {$y_3$};
		\node at (3,4*\e) {$y_4$};
		\node at (6,4*\e) {$y_5$};
		\draw[line width=3] (-5,-\e)--++(0,\e)--++(0,\e-\e*\p)--++(\p*\e,\p*\e)--++(2*\e-2*\e*\p,0)--++(\p*\e,\p*\e)--(-4,3*\e);
		\draw[line width=3] (-3,-\e)--++(0,\e)--++(0,\e-\e*\p)--++(\p*\e,\p*\e)--++(2*\e-2*\e*\p,0)--++(\p*\e,\p*\e)--(-2,3*\e);
		\draw[line width=3] (-2,-\e)--++(0,\e)--++(0,\e-\e*\p)--++(\p*\e,\p*\e)--(1-\e*\p,\e)--++(\p*\e,\p*\e)--(1,3*\e);
		\draw[line width=3] (3,-\e)--(3,3*\e);
		\draw[line width=3] (4,-\e)--++(0,\e)--++(0,\e-\e*\p)--++(\p*\e,\p*\e)--(6-\e*\p,\e)--++(\p*\e,\p*\e)--(6,3*\e);
	\end{tikzpicture}
	\end{center}
  	\caption{Two consecutive ``slices'' of the six-vertex model defining
  	an element of the transfer matrix.}
  	\label{fig:6v3}
\end{figure}
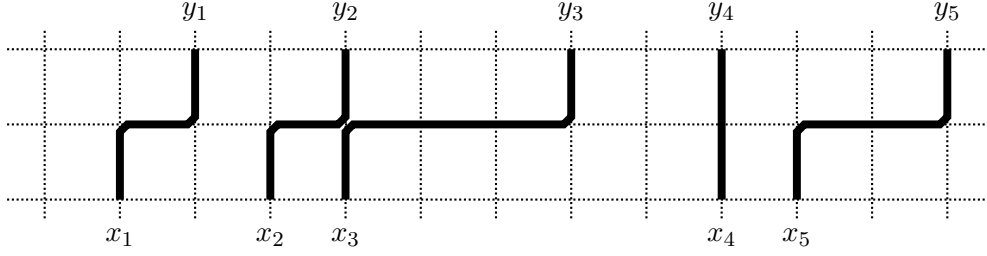
\begin{remark}
	The transfer matrix (and the corresponding six-vertex model)
	possesses a clear asymmetry because of the very different roles played by
	$a_1$ and $a_2$ vertices.
\end{remark}

Following the coordinate Bethe ansatz approach
of \cite{Lieb67}, \cite{baxter2007exactly}, the eigenfunctions
of the transfer matrix can be computed (note that no
Bethe equations are needed as we work on the infinite lattice).
Let us define
\begin{align}\label{Sm_6V}
	\Smv(\xi_1,\xi_2):=1- \frac{a_2+b_1b_2-c_1c_2}{b_1}\xi_1+\frac{a_2b_2}{b_1}\xi_1\xi_2,
\end{align}
and
\begin{align}\label{Phi_6V}
	\Phi_{\vxi}^{6V}(\x):=\sum_{\sigma\in S(k)}
	\sgn(\sigma)\prod_{1\le B<A\le k}\Smv(\xi_{\sigma(A)},\xi_{\sigma(B)})
	\prod_{j=1}^{k}\xi_{\sigma(j)}^{-x_j},\qquad \x\in\tWeyl{k}.
\end{align}

\begin{proposition}[\cite{BCG6V}]\label{prop:6v_ef}
	For all sufficiently small complex $\xi_j$'s (such
	that the series in $\x$ below converges), the
	functions \eqref{Phi_6V} are eigenfunctions of the transposed
	transfer matrix $T_k$ in the sense that
	\begin{align*}
		\sum_{\x\in\tWeyl{k}}
		\Phi_{\vxi}^{6V}(\x)T_k(\x,\y)
		=
		\left(\prod_{j=1}^{k}\frac{b_1+(c_1c_2-b_1b_2)\xi_j}{1-b_2\xi_j}\right)
		\Phi_{\vxi}^{6V}(\y),\qquad \y\in\tWeyl{k}.
	\end{align*}
\end{proposition}

The eigenfunctions of the transposed transfer matrix
correspond to left eigenfunctions
in the notation of the rest of the paper.
One could similarly write down the statement for the
right eigenfunctions:
\begin{corollary}[\cite{BCG6V}]
	For all sufficiently small complex $\xi_j$'s (such
	that the series in $\y$ below converges),
	the eigenfunctions of the transfer matrix
	$T_k$ are
	\begin{align}\label{PhiR_6V}
		(\Refl\Phi_{\vxi}^{6V})(\x)=(-1)^{\frac{k(k-1)}{2}}\sum_{\sigma\in S(k)}
		\sgn(\sigma)\prod_{1\le B<A\le k}\Smv(\xi_{\sigma(B)},\xi_{\sigma(A)})
		\prod_{j=1}^{k}\xi_{\sigma(j)}^{x_j},
	\end{align}
	in the sense that
	\begin{align*}
		\sum_{\y\in\tWeyl{k}}
		T_k(\x,\y)
		(\Refl\Phi_{\vxi}^{6V})(\y)
		=
		\left(\prod_{j=1}^{k}\frac{b_1+(c_1c_2-b_1b_2)\xi_j}{1-b_2\xi_j}\right)
		(\Refl\Phi_{\vxi}^{6V})(\x),\qquad \x\in\tWeyl{k},
	\end{align*}
	where $\Refl$ is the space reflection operator
	\eqref{reflection_operator}.
\end{corollary}
\begin{proof}
	This follows from Proposition \ref{prop:6v_ef}
	and a simple symmetry relation for the
	transfer matrix: $T_k\big((y_1,\ldots,y_k),(x_1,\ldots,x_k)\big)=
	T_k\big((-x_k,\ldots,-x_1),(-y_k,\ldots,-y_1)\big)$.
\end{proof}


\subsection{Connection to eigenfunctions of the conjugated $q$-Hahn operator} 
\label{sub:connection_to_eigenfunctions_of_the_conjugated_q_hahn_boson_operator}

\subsubsection{Matching the cross-terms} 
\label{ssub:matching_the_cross_terms}

We will now match the eigenfunctions
of the six-vertex model to
certain degenerations
of the eigenfunctions
of the conjugated
$q$-Hahn operator.
This degeneration is very similar
to the passage from the $q$-Hahn eigenfunctions
to the ASEP eigenfunctions (see \S \ref{sub:relation_to_q_hahn_boson_eigenfunctions}).

Recall that the conjugated $q$-Hahn eigenfunctions
$\Phild_{\vxi}(\n)$ \eqref{Phild} are given
by the same formula as \eqref{Phi_6V},
but with a different cross-term \eqref{Sm_theta}.
Slightly rewriting this cross-term, we arrive at
\begin{align}\label{Sm_theta_to_6V}
	\frac{1-q\dilp\nu}{\dilp(1-q)}\Smd(\xi_1,\xi_2)=
	1-\frac{1-q\dilp\nu}{\dilp(1-q)}\xi_1+
	\frac{q-\dilp\nu}{\dilp(1-q)}\xi_2+
	\frac{\nu}{\dilp}\xi_1\xi_2.
\end{align}
We can set $\nu=1/(q\dilp)$ or $\nu=q/\dilp$ to kill one of the linear terms in
\eqref{Sm_theta_to_6V}, and arrive at an expression which one can hope to match
to \eqref{Sm_6V}. We will not consider the second degeneration
because it leads to weakly ordered spatial variables
(see the discussion in \S \ref{sub:relation_to_q_hahn_boson_eigenfunctions}).
Thus, we are left with
\begin{align}\label{6V_cross_term_xi}
	\left.\left(\frac{1-q\dilp\nu}{\dilp(1-q)}\Smd(\xi_1,\xi_2)\right)\right\vert_{\nu=1/(q\dilp)}=
	1-
	\frac{1+q^{-1}}{\dilp}\xi_2+
	\frac{q^{-1}}{\dilp^{2}}\xi_1\xi_2,
\end{align}
and so for any $x_1\le \ldots\le x_k$:
\begin{align}\label{Phi_theta_to_6V}
	\begin{array}{>{\displaystyle}l>{\displaystyle}l}
		&\bigg(\left(-\frac{1-q\dilp\nu}{\dilp(1-q)}\right)^{\frac{k(k-1)}{2}}
		\Phild_{\vxi}(x_k,\ldots,x_1)
		\bigg)\bigg\vert_{\nu=1/(q\dilp)}
		\\&\hspace{80pt}=
		\sum_{\sigma\in S(k)}
		\sgn(\sigma)\prod_{B<A}
		\left(1-
		\frac{1+q^{-1}}{\dilp}\xi_{\sigma(A)}+
		\frac{q^{-1}}{\dilp^{2}}\xi_{\sigma(A)}\xi_{\sigma(B)}\right)
		\prod_{j=1}^{k}\xi_{\sigma(j)}^{-x_{j}}.
	\end{array}
\end{align}

It is now possible to match the cross-term in the right-hand side
of \eqref{Phi_theta_to_6V}
to the six-vertex cross-term \eqref{Sm_6V}.
Let us denote
(we include $a_1=1$ in formulas below for symmetry)
\begin{align}\label{delta_delta_6V}
	\svd:=\frac{a_1a_2+b_1b_2-c_1c_2}{2\sqrt{a_1a_2b_1b_2}},
	\qquad
	\svpar:=\frac{a_2 b_2}{a_1 b_1}.
\end{align}
\begin{proposition}
	If the pair of parameters $(\dilp,q)$ takes one of
	the two values $(\dilp_{\pm},q_{\pm})$, where
	\begin{align}\label{q_theta_6V}
		\begin{array}{>{\displaystyle}rc>{\displaystyle}l}
		\dilp_{\pm}&:=&\frac{1}{\sqrt{\svpar}}\left(\svd\pm\sqrt{\svd^{2}-1}\right);\\
		q_{\pm}&:=&-1+2\svd^{2}\mp2\svd\sqrt{\svd^{2}-1},
		\end{array}
	\end{align}
	then for any integers $x_1\le \ldots\le x_k$:
	\begin{align}\label{Phi_6V_degeneration}
		\begin{array}{>{\displaystyle}rc>{\displaystyle}l}
		\bigg(\left(-\frac{1-q\dilp\nu}{\dilp(1-q)}\right)^{\frac{k(k-1)}{2}}
		\Phild_{\vxi}(x_k,\ldots,x_1)
		\bigg)\bigg\vert_{\nu=1/(q\dilp)}&=&
		\Phi_{\vxi}^{6V}(x_1,\ldots,x_k),\\
		\bigg(\left(\frac{1-q\dilp\nu}{\dilp(1-q)}\right)^{\frac{k(k-1)}{2}}
		\Phird_{\vxi}(x_k,\ldots,x_1)
		\bigg)\bigg\vert_{\nu=1/(q\dilp)}&=&
		(\Refl\Phi_{\vxi}^{6V})(x_1,\ldots,x_k)\cdot\mathbf{1}_{x_1<\ldots<x_k}.
		\end{array}
	\end{align}
\end{proposition}
This degeneration of eigenfunctions
should be compared to the ASEP setting
(second two formulas in \eqref{Psi_qHahn_to_ASEP_main}).
\begin{proof}
	The values \eqref{q_theta_6V} are obtained by matching coefficients
	in cross-terms in \eqref{Phi_theta_to_6V} and \eqref{Sm_6V},
	which leads to a quadratic equation. Formulas \eqref{Phi_6V_degeneration}
	then follow from \eqref{Phild}--\eqref{Phird}.
\end{proof}

One can readily see that for the solutions \eqref{q_theta_6V}
we have $q_+q_-=1$, and, moreover,
\begin{enumerate}[$\bullet$]
	\item If $\svd>1$, then $0<q_+<1$ and $\dilp_{\pm}>0$;
	\item If $-1<\svd<1$, then $q_{\pm}$ are complex numbers lying on the unit circle.
	They have nonzero imaginary part except $\svd=0$ when $q_-=q_+=-1$.
	The numbers $\dilp_{\pm}$
	also have nonzero imaginary part;
	\item If $\svd<-1$, then $0<q_-<1$ and $\dilp_{\pm}<0$;
	\item If $\svd=\pm1$, then $q_{+}=q_{-}=1$ and $\dilp_+=\dilp_-=\pm\svpar^{-1/2}$.
\end{enumerate}


\begin{remark}[Stochastic six-vertex model]
	If $\svd>1$, then it is possible to choose $\svpar>0$ so that $\dilp_+=1$.
	This means that the eigenfunctions of the transfer matrix
	of the six-vertex model coincide with those of the ASEP
	(see Fig.~\ref{fig:big_scheme}).
	Moreover, the transfer matrix itself can be interpreted as a
	Markov transition operator associated with a certain discrete-time
	stochastic particle system --- the (asymmetric) stochastic
	six-vertex model.
	This model is studied in detail in
	\cite{BCG6V}.
	All results discussed in \S \ref{sec:application_to_asep}
	(Plancherel isomorphisms, spectral biorthogonality, symmetrization identities)
	thus apply to
	eigenfunctions
	of the stochastic six-vertex model.
\end{remark}


\subsubsection{A change of spectral variables} 
\label{ssub:a_change_of_spectral_variables}

In the rest of the paper we will assume that $\svd\ne\pm1$.
Then there exists a change of spectral variables
$\vxi\to\z$
similar to the involution $\xinu_{\dilp}$ from \S \ref{sub:analogues_of_main_results},
which leads to a linear cross-term in the eigenfunctions.
More precisely, set
\begin{align*}
	\xinu^{6V}(z):=\frac{\dilp-z}{1-z/(q\dilp)}.
\end{align*}
Note that this map is an involution (in contrast with $\xinu^{\mathrm{ASEP}}$
from \S \ref{sub:asep_and_its_bethe_ansatz_eigenfunctions}).
\begin{definition}\label{def:q_theta_parameter_series}
	Whenever we want to take parameters
	$(q,\dilp)$
	depending on the parameters $(\svd,\svpar)$
	\eqref{delta_delta_6V}
	of the six-vertex model, we will set
	\begin{align}\label{q_theta_parameter_series}
		(q,\dilp):=\begin{cases}
		(q_+,\dilp_+)&\mbox{if $\svd>1$ or $-1<\svd<1$};
		\\
		(q_-,\dilp_-)&\mbox{if $\svd<-1$},
	\end{cases}
	\end{align}
	where $q_{\pm}$ and $\dilp_{\pm}$ are defined in
	\eqref{q_theta_6V}.
\end{definition}
Define
\begin{align}\label{Psi_6V}
	\Psi_{\z}^{6V}(\x):=
	\sum_{\sigma\in S(k)}\prod_{B<A}
	\frac{z_{\sigma(B)}-qz_{\sigma(A)}}
	{z_{\sigma(B)}-z_{\sigma(A)}}\prod_{j=1}^{k}
	\left(\frac{\dilp-z_{\sigma(j)}}{1-z_{\sigma(j)}/(q\dilp)}\right)^{-x_j},
	\qquad\x\in\tWeyl{k}.
\end{align}
The reflection of the above function gives
\begin{align}\label{PsiR_6V}
	(\Refl\Psi_{\z}^{6V})(\x)=
	\sum_{\sigma\in S(k)}\prod_{B<A}
	\frac{z_{\sigma(A)}-qz_{\sigma(B)}}
	{z_{\sigma(A)}-z_{\sigma(B)}}\prod_{j=1}^{k}
	\left(\frac{\dilp-z_{\sigma(j)}}{1-z_{\sigma(j)}/(q\dilp)}\right)^{x_j}.
\end{align}
Then one can readily check that the functions
\eqref{Psi_6V}--\eqref{PsiR_6V} are related to
\eqref{Phi_6V} and \eqref{PhiR_6V} as follows:
\begin{align*}
	\Psi_{\xinu^{6V}(\vxi)}^{6V}(\x)=
	(q\dilp)^{\frac{k(k-1)}{2}}
	(\Vand(\vxi))^{-1}\Phi_{\vxi}^{6V}(\x),
	\qquad
	(\Refl\Psi_{\xinu^{6V}(\vxi)}^{6V})(\x)=
	(q\dilp)^{\frac{k(k-1)}{2}}
	(\Vand(\vxi))^{-1}(\Refl\Phi_{\vxi}^{6V})(\x),
\end{align*}
where, as usual, $\Vand(\vxi)$ is the Vandermonde determinant.
Also, the functions \eqref{Psi_6V}--\eqref{PsiR_6V}
arise from the conjugated $q$-Hahn eigenfunctions
\eqref{Psild}--\eqref{Psird} as
\begin{align}\label{Psi_Psi_qhahn_6V}
	\begin{array}{>{\displaystyle}rc>{\displaystyle}l}
	\Psi_{\z}^{6V}(x_1,\ldots,x_k)&=&
	\Psild_{\z}(x_k,\ldots,x_1)\vert_{\nu=1/(q\dilp)}
	;
	\\\rule{0pt}{15pt}
	(\Refl\Psi_{\z}^{6V})(x_1,\ldots,x_k)\cdot
	\mathbf{1}_{x_1<\ldots<x_k}&=&
	(q^{-1}-1)^{-k}
	\Psird_{\z}(x_k,\ldots,x_1)\vert_{\nu=1/(q\dilp)}
	,
	\end{array}
\end{align}
where $x_1\le \ldots\le x_k$.
The degeneration of eigenfunctions \eqref{Psi_Psi_qhahn_6V} should be compared to the
first two formulas in \eqref{Psi_qHahn_to_ASEP_main}.



\subsection{Heisenberg XXZ spin chain} 
\label{sub:xxz}

Consider the transfer matrix $T_k$ (where $k$ is the number of vertical lines in the
configuration)
defined in \S \ref{sub:transfer_matrix_and_its_eigenfunctions},
and let
\begin{align*}
	a_1=a_2=1,
	\qquad
	b_1=b_2=\epsilon b,\qquad
	c_1=c_2=1-\epsilon c,
\end{align*}
where $\epsilon>0$ is a small parameter, $b>0$, and $c\in\R$.
Then one readily sees that
\begin{align}\label{Tk_XXZ_k_spin}
	T_k(\x,\y)=\mathbf{1}_{\y=\x+1}+
	\epsilon b\cdot \tHxxz(\x,\y)+O(\epsilon^{2}),
	\qquad \x,\y\in\tWeyl{k},
\end{align}
where $\x+1$ means increasing all coordinates by 1,
and
$\tHxxz$ is the
following matrix:
\begin{align}\label{tHxxz}
	\tHxxz(\x,\y+1)=\sum_{i}
	(\mathbf{1}_{\y=\x^{-}_{i}}-\svd\cdot\mathbf{1}_{\y=\x})
	+\sum_{j}
	(\mathbf{1}_{\y=\x^{+}_{j}}-\svd\cdot\mathbf{1}_{\y=\x}),
	\qquad\svd=\frac{c}{b}
\end{align}
(the notation $\x_{i}^{\pm}$ is as in
\S \ref{sub:asep_and_its_bethe_ansatz_eigenfunctions}).
The sums above are taken over all
$i$ and $j$ such that $\x^{-}_{i}$ and $\x^{+}_{j}$, respectively,
belong to $\tWeyl{k}$.
Note that $\svd$ above is simply the $\epsilon\to0$
limit of the corresponding six-vertex parameter $\svd$ \eqref{delta_delta_6V}.
In the same limit, the second parameter $\svpar$
of the six-vertex eigenfunctions turns into $1$.

We thus arrive at
the
Hamiltonian of the (spin-$\frac12$) Heisenberg XXZ quantum spin chain
on the infinite lattice:
\begin{align}\label{Hxxz}
	(\Hxxz f)(\x)=
	\sum_{i}\big(f(\x_{i}^{-})-\svd f(\x)\big)
	+
	\sum_{j}\big(f(\x_{j}^{+})-\svd f(\x)\big)
\end{align}
(both sums are over allowed configurations as in \eqref{tHxxz}).
The operator $\Hxxz$ acts on (compactly supported) functions
in $k$ variables $\x\in\tWeyl{k}$,
according to the transfer matrix $T_k$
in \eqref{Tk_XXZ_k_spin}.
The integers $x_1,\ldots,x_k$ are traditionally understood as
encoding positions of up spins (or magnons), and all other
lattice points correspond to down spins,
see Fig.~\ref{fig:XXZ}
(note that $\Hxxz$
preserves the number of up spins).
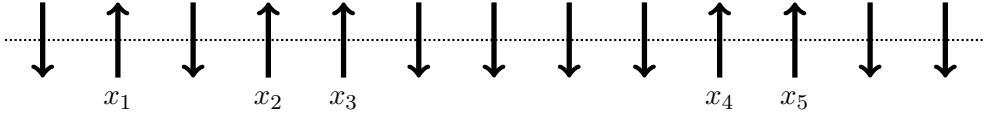
\begin{figure}[htb]
	\begin{center}
	\begin{tikzpicture}
		[scale=1, thick]
		\def\e{.5}
		\draw[densely dotted] (-6.5,0) -- (6.5,0);
		\foreach \i in {-5,-3,-2,3,4}
		{
			\draw[->,line width=2] (\i,-\e) -- (\i,\e);
		}
		\foreach \i in {-6,-4,-1,0,1,2,5,6}
		{
			\draw[->,line width=2] (\i,\e) -- (\i,-\e);
		}
		\node at (-5,-1.6*\e) {$x_1$};
		\node at (-3,-1.6*\e) {$x_2$};
		\node at (-2,-1.6*\e) {$x_3$};
		\node at (3,-1.6*\e) {$x_4$};
		\node at (4,-1.6*\e) {$x_5$};
	\end{tikzpicture}
	\end{center}
  	\caption{A configuration of spins in the XXZ spin chain.}
  	\label{fig:XXZ}
\end{figure}
\begin{remark}
	On a finite lattice with periodic boundary,
	$\Hxxz$ can be rewritten in a more traditional
	way involving nearest-neighbor quantum interactions.
	Namely, encode spin configurations
	on a finite lattice $\{1,\ldots,L\}$
	by vectors in $(\C^{2})^{\otimes L}$,
	where the basis in $\C^{2}$
	consists of vectors
	$|{\uparrow}\rangle$
	and
	$|{\downarrow}\rangle$.
	Let $\sigma^{x}$, $\sigma^{y}$, and $\sigma^{z}$ be the Pauli matrices
	acting in $\C^{2}$
	\begin{align*}
		\sigma^{x}=\begin{pmatrix}
			0&1\\1&0
		\end{pmatrix}
		\qquad
		\sigma^{y}=\begin{pmatrix}
			0&-i\\i&0
		\end{pmatrix}
		\qquad
		\sigma^{z}=\begin{pmatrix}
			1&0\\0&-1
		\end{pmatrix}.
	\end{align*}
	The XXZ spin chain Hamiltonian on the finite
	lattice is
	the following operator in $(\C^{2})^{\otimes L}$:
	\begin{align}\label{Hxxz_quantum}
		-\frac12\sum_{j=1}^{L}
		\big(\sigma_{j}^{x}\sigma_{j+1}^{x}+
		\sigma_{j}^{y}\sigma_{j+1}^{y}
		+\svd\sigma_{j}^{z}\sigma_{j+1}^{z}\big),
		\qquad \sigma_{L+1}^{a}\equiv \sigma_{1}^{a}.
	\end{align}
	Here $\sigma_{j}^{a}$ ($a\in\{x,y,z\}$)
	denotes the operator acting as $\sigma^{a}$
	in the $j$-th copy of $\C^{2}$
	(and trivially in all other copies).
	
	The operator \eqref{Hxxz_quantum} preserves the number of up
	spins.
	When restricted to the subspace of $(\C^{2})^{\otimes L}$
	corresponding to exactly $k$ up spins,
	\eqref{Hxxz_quantum}
	coincides (up to an additive constant)
	with the analogue of $\Hxxz$ \eqref{Hxxz} on the finite lattice
	with periodic boundary.
\end{remark}
The eigenfunctions of the XXZ spin chain
can be computed independently by applying the coordinate
Bethe ansatz to the operator $\Hxxz$. This
was first performed in
\cite{Bethe1931} for $\svd=1$ and
\cite{YangYang1966} for all $\svd$, see also
\cite{baxter2007exactly} for a general perspective.
On the other hand, these
are simply the six-vertex eigenfunctions
specialized at $\svpar=1$. In this case
one readily checks that $q=1/\dilp^{2}$, and
so, for example,
\eqref{Psi_6V} turns into
\begin{align}\label{XXZ}
	\Psi_{\z}^{\mathrm{XXZ}}(\x)=
	\sum_{\sigma\in S(k)}\prod_{B<A}
	\frac{z_{\sigma(B)}-\dilp^{-2}z_{\sigma(A)}}
	{z_{\sigma(B)}-z_{\sigma(A)}}\prod_{j=1}^{k}
	\left(\frac{1-\dilp^{-1}z_{\sigma(j)}}{\dilp^{-1}-z_{\sigma(j)}}\right)^{-x_j},
	\qquad\x\in\tWeyl{k},
\end{align}
where $\dilp$ is related to $\svd$ via
\eqref{q_theta_6V} which can be restated as follows
(taking into account the agreement of Definition \ref{def:q_theta_parameter_series}):
\begin{align}\label{Delta_theta_XXZ_relation}
	\svd=\frac{1}{2}\left(\dilp+\frac{1}{\dilp}\right)=
	\frac{1}{2}\left(\frac1{\sqrt q}+{\sqrt q}\right).
\end{align}
In the same way \eqref{PsiR_6V} turns into
\begin{align}\label{XXZ_R}
	(\Refl\Psi_{\z}^{\mathrm{XXZ}})(\x)=
	\sum_{\sigma\in S(k)}\prod_{B<A}
	\frac{z_{\sigma(A)}-\dilp^{-2}z_{\sigma(B)}}
	{z_{\sigma(A)}-z_{\sigma(B)}}\prod_{j=1}^{k}
	\left(\frac{1-\dilp^{-1}z_{\sigma(j)}}{\dilp^{-1}-z_{\sigma(j)}}\right)^{x_j}.
\end{align}
Note that the cross-term \eqref{6V_cross_term_xi}
in the $\vxi$
variables  becomes
$1-2\svd\xi_2+\xi_1\xi_2$ for the XXZ model.
One can also readily
specialize the eigenfunctions \eqref{Phi_6V} and \eqref{PhiR_6V},
but we will not use them.

Plancherel formulas for the XXZ eigenfunctions
were obtained in
\cite{BabbittThomas}, \cite{BabbittGutkin},
\cite{Gutkin}. Those formulas
involved real and positive
Plancherel measures governing Plancherel decompositions,
as it should be for a model with a Hermitian symmetric Hamiltonian
(cf. Remark \ref{rmk:name_for_Plancherel}).
In \S \ref{sub:formulas_for_}
and \S \ref{sub:formulas_for_svd_1_and_complex_q_dilp_}
below we discuss other Plancherel formulas for the XXZ model
which follow from
our main results of \S \S \ref{sec:main_results}--\ref{sec:the_q_mu_nu_boson_process_and_coordinate_bethe_ansatz}.

\begin{remark}\label{rmk:ASEP_XXZ_relation}
If $|\svd|>1$, then the XXZ Hamiltonian $\Hxxz$ is conjugate
to the stochastic ASEP generator $\Hasep$ \eqref{H_ASEP}:
\begin{align*}
	\frac{1}{2\svd}\Dil{\dilp}\Hxxz\Dil{\dilp}^{-1}=
	\Hasep,
	\qquad \tau=1/{\dilp^{2}},
	\qquad \begin{cases}
		\dilp>1&\text{if $\svd>1$};\\
		\dilp<-1&\text{if $\svd<-1$},
	\end{cases}
\end{align*}
where $\dilp$ is related to $\svd$ as in \eqref{Delta_theta_XXZ_relation},
and the operator $\Dil{\dilp}$ is defined in \eqref{Dilation}. This property
was noted in \cite{Schutz1997exact}.

In the case $|\svd|<1$, one could also perform a similar conjugation,
but this will lead to complex parameter $\tau$ in the ASEP generator,
so the latter will no longer be stochastic. We will discuss the relevant
results for the eigenfunctions in \S \ref{sub:formulas_for_svd_1_and_complex_q_dilp_} below.
\end{remark}


\subsection{Case $|\svd|>1$, real $q$ and $\dilp$} 
\label{sub:formulas_for_}

Observe that the six-vertex and XXZ eigenfunctions
(\eqref{Psi_6V} and \eqref{XXZ}, respectively)
can be both obtained from the ASEP eigenfunctions
\eqref{PsiASEP_bwd}
by applying the dilation operator $\Dil{a}$
with suitable $a\ne 0$ and by possibly rescaling
the $\z$ variables (cf. Fig.~\ref{fig:big_scheme}).\footnote{This relation between XXZ and ASEP also extends to their generators,
see Remark \ref{rmk:ASEP_XXZ_relation}.}
These operations do not change the cross-term
parameter of the eigenfunctions, so the
six-vertex parameter $q$ and the XXZ parameter $\dilp^{-2}$
are the same as the ASEP parameter $\tau$.
So if $|\svd|>1$, then all these parameters are real and between $0$
and $1$.
One issue which arises for $\svd<-1$
is that the parameter $\dilp$ in six-vertex or XXZ eigenfunctions
is negative, but this is readily resolved by applying the dilation $\Dil{-1}$
and negating the variables $\z$.

Therefore,
all ASEP results,
namely,
spectral biorthogonality (Theorem \ref{thm:ASEP_biorthogonality}),
Plancherel formulas (Theorems \ref{thm:spatial_Plancherel_ASEP}
and \ref{thm:spectral_Plancherel_ASEP}),
and symmetrization identities
(\S \ref{sub:tracy_widom_symmetrization_identities}),
are readily seen to be equivalent
to the corresponding results
for the six-vertex and XXZ eigenfunctions
for $|\svd|>1$.

For example, the XXZ Plancherel formula
takes the following form (cf. \eqref{ASEP_Plancherel_desymm}):
\begin{theorem}\label{thm:XXZ_Plancherel}
	For $|\svd|>1$ and any $\x,\y\in\tWeyl{k}$,
	\begin{align}\label{XXZ_Plancherel_formula}
		\begin{array}{>{\displaystyle}l>{\displaystyle}l}
		&\oint_{\widetilde\ga_{\dilp}}\frac{dz_1}{2\pi\i}
		\ldots
		\oint_{\widetilde\ga_{\dilp}}\frac{dz_k}{2\pi\i}
		\prod_{B<A}\frac{z_A-z_B}{z_A-\dilp^{-2} z_B}
		\\
		&\hspace{60pt}\times\prod_{j=1}^{k}
		\frac{1-\dilp^{-2}}{(1-\dilp^{-1}z_j)(\dilp^{-1}-z_j)}
		\left(\frac{1-\dilp^{-1}z_j}{\dilp^{-1}-z_j}\right)^{-x_j}
		(\Refl\Psi^{\mathrm{XXZ}}_{\z})(\y)=\mathbf{1}_{\x=\y}.
		\end{array}
	\end{align}
	The integration contours
	$\widetilde\ga_{\dilp}$
	are
	positively oriented small circles around $\dilp$.
	The parameter $\dilp\in\R$ is related to $\svd$ via \eqref{Delta_theta_XXZ_relation},
	and, moreover,
	$|\dilp|>1$.
\end{theorem}

In the case $|\svd|>1$
the Plancherel formulas in
\cite{BabbittThomas},
\cite{BabbittGutkin}, \cite{Gutkin}
(stated in our spectral variables $\z$)
involve string specializations
$\z=\w\circ\la$ \eqref{w_circ_la}
corresponding to partitions $\la\vdash k$
(as in, e.g., \eqref{completeness_big_1}--\eqref{completeness_big_2}).
Moreover, the integration in those XXZ Plancherel formulas
is performed over
circles centered at $0$
with radii $\dilp^{\ell}$, $\ell=0,1,\ldots$
(and the Plancherel measures understood in a suitable way are positive on such contours).
It seems plausible that our Plancherel formula
\eqref{XXZ_Plancherel_formula}
(which involves integration over
small circles around $\dilp$
and no string specializations)
can be brought to a form with large integration contours, and then matched to formulas
existing in the literature cited above.
Indeed, in the process of contour deformation (from small to large contours),
one would need to pick residues
corresponding to poles of the integrand in \eqref{XXZ_Plancherel_formula}
at $z_A=\dilp^{-2}z_B$, $B<A$. This should lead to a residue expansion
employing string specializations (as in Proposition \ref{prop:nesting_unnesting},
see also \cite[\S 7.2]{BorodinCorwinPetrovSasamoto2013}).
We do not pursue this computation here.

\medskip

We will not write down other XXZ or six-vertex formulas with real parameters $(q,\dilp)$.


\subsection{Case $|\svd|<1$, complex $q$ and $\dilp$} 
\label{sub:formulas_for_svd_1_and_complex_q_dilp_}

When $|\svd|<1$, the cross-term parameters in our models
($\dilp^{-2}$ in XXZ or $q$ in the six-vertex model)
become complex numbers of modulus $1$.
For definiteness, we will consider only the XXZ
case.\footnote{The six-vertex formulas
will be equivalent to the ones for the XXZ spin chain (cf. the beginning of \S \ref{sub:formulas_for_}).
Similarly one could also write formulas for the ASEP
eigenfunctions with complex $\tau$
of modulus $1$.}
We start with the
conjugated $q$-Hahn eigenfunctions
with complex $q$, $|q|=1$
(they can be treated
similarly to the
$q$-Boson ones with complex $q$
discussed in \cite[\S5]{BorodinCorwinPetrovSasamoto2013}), and then specialize them to the
XXZ eigenfunctions by putting
$q=\dilp^{-2}$, $\nu=\dilp$
(cf. Fig.~\ref{fig:big_scheme}).

Fix an integer $k\ge1$ (this is the number of particles in our particle system),
and assume that $q\in\C$, $|q|=1$, and $q^{j}\ne 1$ for all $j=1,2,\ldots,k-1$.\footnote{In
particular, this condition excludes the case $\svd=0$.}
The conjugated $q$-Hahn
eigenfunctions $\Psild_{\z}$ and $\Psird_{\z}$
\eqref{Psild}--\eqref{Psird},
as well as the corresponding direct transform
$\Plde^{q,\nu,\dilp}$
\eqref{Pldd},
are defined in the same way as for real $q$.
Recall that the integration contours $\ga_1^{\dilp},\ldots,\ga_k^{\dilp}$
used in the definition of
the inverse transform
$\Plie^{q,\nu,\dilp}$ \eqref{Plid}
are such that
$\ga_k^{\dilp}$
is a small positively oriented contour around
$\dilp$ not containing $\dilp q$,
$\ga_A^{\dilp}$ contains $q\ga_B^{\dilp}$
for all $1\le A<B\le k$, and,
moreover, $\nu^{-1}$ is outside all contours.
For complex $q$ satisfying our assumptions, and possibly complex $\dilp$,
these contours exist but
differ from the ones
in the real case, see Fig.~\ref{fig:complex_contours_1}
(and Fig.~\ref{fig:contours} for the real case).

\begin{figure}[htb]
	\begin{center}
	\scalebox{.9}{\begin{tikzpicture}
		[scale=3.3, x=1cm,y=1cm]
		\def\pt{0.025}
		\def\dd{(0.411, 0.911)}
		\def\ddq{(0.411, -0.911)}
		\def\ddqq{(-0.955, 0.295)}
		\def\ddqqq{(0.853, 0.520)}
		\draw[->, thick] (-1.2,0) -- (1.5,0);
	  	\draw[->, thick] (0,-1.2) -- (0,1.2);
	  	\draw[fill] (1,0) circle (\pt) node [below right] {1};
	  	\draw[fill] (0,0) circle (\pt) node [below left] {$0$};
	  	\draw[fill] (1.3,0) circle (\pt) node [above, yshift=0,xshift=3] {$\nu^{-1}$};
	  	\draw[very thick] (0,0) circle (1);
	  	\draw[fill] \dd circle (\pt) node [above] {$\dilp$};
	  	\draw[fill] \ddq circle (\pt) node [below] {$\dilp^{-1}$};
	  	\draw[fill] \ddqq circle (\pt) node [left] {$\dilp^{-3}$};
	  	\draw[fill] \ddqqq circle (\pt) node [below left, xshift=4] {$\dilp^{-5}$};
	  	\draw[ultra thick] \dd circle (.3) node (ddcont) {};
	  	\draw[thick, dotted] \ddq circle (.3);
	  	\draw[ultra thick] \ddq circle (.33) node (ddqcont) {};
	  	\draw[thick, dotted] \ddqq circle (.33);
	  	\draw[ultra thick] \ddqq circle (.36) node (ddqqcont) {};
	  	\draw[thick, dotted] \ddqqq circle (.36);
		\def\dds{(0, .95)}
	  	\draw[ultra thick, rotate=-45] \dds ellipse (.7 and .5) node (ddqqqcont) {};
		\draw[->, very thick, dashed] (1.5,1) -- ([shift={(.25,0.08)}]ddcont.east) node [pos=-.15] {$\ga_4^{\dilp}$};
		\draw[->, very thick, dashed] (.6,-.2) -- ([shift={(0,-.25)}]ddcont.south) node [pos=-.1] {$\ga_3^{\dilp}$};
		\draw[->, very thick, dashed] (.6,-.38) -- ([shift={(0.1,.28)}]ddqcont.north);
		\draw[->, very thick, dashed] (-.6,-.2) -- ([shift={(.15,-.31)}]ddqqcont.east) node [pos=-.3] {$\ga_2^{\dilp}$};
		\draw[->, very thick, dashed] (-.6,-.2) -- ([shift={(-.20,-.2)}]ddcont.west);
		\draw[->, very thick, dashed] (-.55,-.36) -- ([shift={(-.22,.22)}]ddqcont.west);
		\draw[->, very thick, dashed] (-.5,1) -- ([shift={(0,.38)}]ddqqcont.east) node [pos=-.2] {$\ga_1^{\dilp}$};
		\draw[->, very thick, dashed] (-.45,.97) -- ([shift={(.05,.35)}]ddqcont.west);
		\draw[->, very thick, dashed] (-.33,1.05) -- ([shift={(-.31,.-.05)}]ddcont.west);
	\end{tikzpicture}}
	\end{center}
  	\caption{A possible choice of integration contours $\ga_1^{\dilp},\ga_2^{\dilp},\ga_3^{\dilp},\ga_4^{\dilp}$ for $k=4$ and $q=\dilp^{-2}$, where $\dilp$ is on the unit circle.
  	In the picture, $\dilp$ is relatively close to $e^{\i\pi/3}$, so $\dilp q^{3}=\dilp^{-5}$
  	is close to $\dilp$. The contour $\ga_4^{\dilp}$
  	is a single small circle around $\dilp$. The contour
  	$\ga_3^{\dilp}$ is a union of $\ga_4^{\dilp}$
  	and a slightly larger contour around $\dilp q=\dilp^{-1}$ (so
  	$\ga_3^{\dilp}$ contains $q\ga_4^{\dilp}$).
  	The contour
  	$\ga_2^{\dilp}$ is a union of $\ga_3^{\dilp}$
  	and a yet slightly larger contour around $\dilp q^{2}=\dilp^{-3}$
  	(so
  	$\ga_2^{\dilp}$ contains both $q\ga_3^{\dilp}$ and $q\ga_4^{\dilp}$).
  	The contour $q\ga_{2}^{\dilp}$
  	intersects, $\ga_4^{\dilp}$,
  	so
	the contour
  	$\ga_1^{\dilp}$ must be a union of
  	circles around $\dilp^{-1}$, $\dilp^{-3}$,
  	and of an ellipse around $\dilp$ and $\dilp q^{3}=\dilp^{-5}$
  	which contains $\ga_4^{\dilp}$
  	(in this way, $\ga_1^{\dilp}$ contains $q\ga_2^{\dilp}$, $q\ga_3^{\dilp}$,
  	and $q\ga_4^{\dilp}$).
  	The images of circles under multiplication by $q=\dilp^{-2}$ are shown dotted.
  	}
  	\label{fig:complex_contours_1}
\end{figure}
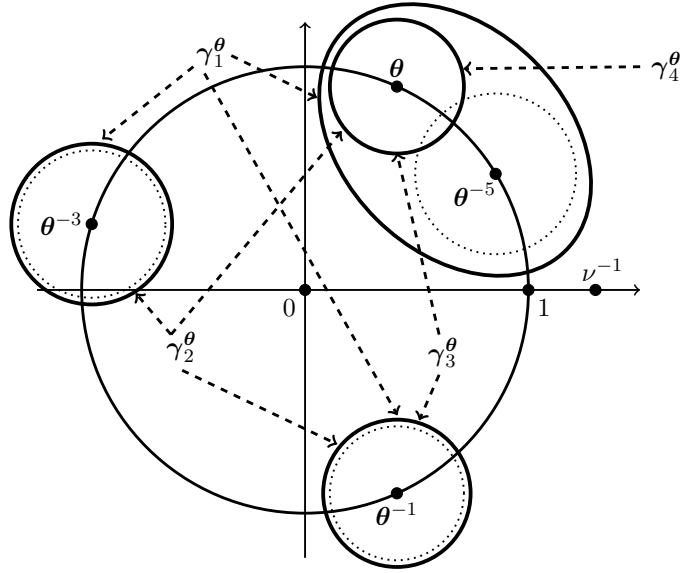

With these modifications the spatial Plancherel formula
(i.e., the first half of Theorem \ref{thm:Plancherel_isom_theta}
stating that
$\Plspatiale^{q,\nu,\dilp}=\Plie^{q,\nu,\dilp}\Plde^{q,\nu,\dilp}$
acts as the identity operator on $\Wc^{k}$)
continues to hold.
In this statement we can set $q=\dilp^{-2}$
(where $\dilp\in\C$, $|\dilp|=1$, and $\dilp^{2j}\ne 1$
for all $j=1,\ldots,k-1$)
right away,
and the spatial Plancherel theorem will continue to hold.

Now we must further specialize
$\nu=\dilp$. This
requires deforming all contours to $\ga_k^{\dilp}$, because
the ``nested'' integration contours $\ga_1^{\dilp},\ldots,\ga_k^{\dilp}$
must not contain $\nu^{-1}=\dilp^{-1}=q\dilp$ (for real $q$
a similar
obstacle
was encountered in the ASEP case,
see Theorem \ref{thm:spatial_Plancherel_ASEP}).
The required contour deformation
can be performed using an analogue of
the second part of Proposition \ref{prop:nesting_unnesting}
which requires one modification.
Let us fix a small positively oriented circle $\ga_k^{\dilp}$
around $\dilp$ which does not contain $q\dilp=\dilp^{-1}$.
The image of this contour
$\ga_k^{\dilp}$ under the multiplication
by some power of $q$ can intersect with
$\ga_k^{\dilp}$ (as this happens in Fig.~\ref{fig:complex_contours_1}),
and this affects which residues are contributing
when contours are deformed to $\ga_k^{\dilp}$.
Define the following subsets of the contour $\ga_k^{\dilp}$
for $\ell\ge1$ (see Fig~\ref{fig:arcs}):
\begin{align*}
	\Gamma_{\ell}(\dilp^{-2}):=\big\{
	z\in\ga_k^{\dilp}\colon
	\text{for all $1\le j<\ell$,
	$\dilp^{-2j}z$ is outside $\ga_k^{\dilp}$}
	\big\}.
\end{align*}

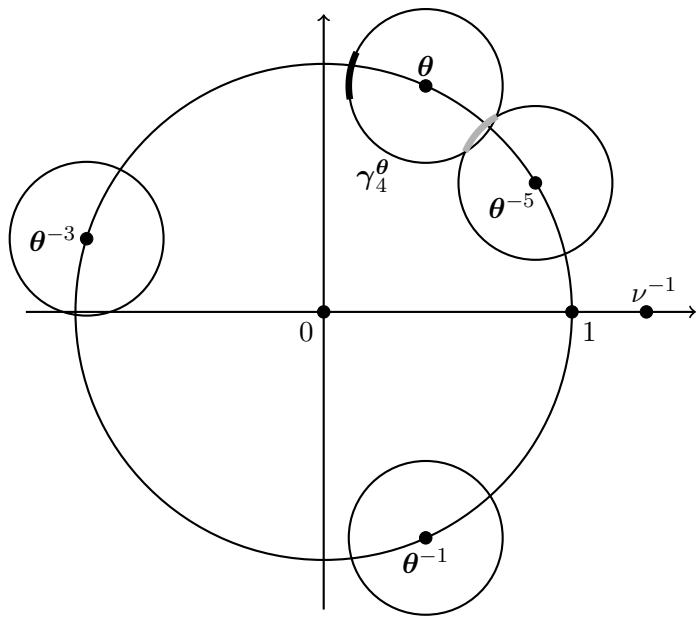
\begin{figure}[htb]
	\begin{center}
	\begin{tikzpicture}
		[scale=3.3, x=1cm,y=1cm]
		\def\pt{0.025}
		\def\dd{(0.411, 0.911)}
		\def\ddq{(0.411, -0.911)}
		\def\ddqq{(-0.955, 0.295)}
		\def\ddqqq{(0.853, 0.520)}
		\draw[->, thick] (-1.2,0) -- (1.5,0);
	  	\draw[->, thick] (0,-1.2) -- (0,1.2);
	  	\draw[fill] (1,0) circle (\pt) node [below right] {1};
	  	\draw[fill] (0,0) circle (\pt) node [below left] {$0$};
	  	\draw[fill] (1.3,0) circle (\pt) node [above, yshift=0,xshift=3] {$\nu^{-1}$};
	  	\draw[thick] (0,0) circle (1);
	  	\draw[fill] \dd circle (\pt) node [above] {$\dilp$};
	  	\draw[fill] \ddq circle (\pt) node [below] {$\dilp^{-1}$};
	  	\draw[fill] \ddqq circle (\pt) node [left] {$\dilp^{-3}$};
	  	\draw[fill] \ddqqq circle (\pt) node [below left, xshift=4] {$\dilp^{-5}$};
	  	\draw[ thick] \dd circle (.31) node (ddcont) {};
	  	\draw[ thick] \ddq circle (.31) node (ddqcont) {};
	  	\draw[ thick] \ddqq circle (.31) node (ddqqcont) {};
	  	\draw[ thick] \ddqqq circle (.31) node (ddqqqcont) {};
		\draw [white!40!gray, line width=2.5,domain=120:156] plot ({0.853+.31*cos(\x)}, {0.520+.31*sin(\x)});
		\draw [line width=2.5,domain=154:190] plot ({0.411+.31*cos(\x)}, {0.911+.31*sin(\x)});
		\node at (.2,.55) {$\ga_4^{\dilp}$};
	\end{tikzpicture}
	\end{center}
  	\caption{
  	Integration contours in Proposition \ref{prop:XXZ_Pli_expanded}
  	for $k=4$ and $\dilp$ relatively close to $e^{\i\pi/3}$.
  	The contour $\ga_4^{\dilp}$ is a fixed circle around $\dilp$.
  	The contours $\Gamma_{1}(\dilp^{-2})$,
  	$\Gamma_{2}(\dilp^{-2})$,
  	$\Gamma_{3}(\dilp^{-2})$
  	are all equal to $\ga_4^{\dilp}$, i.e., they are full
  	circles. The contour
  	$\Gamma_{4}(\dilp^{-2})$
  	is $\ga_4^{\dilp}$
  	without the thick black arc, because for $z$
  	in this arc, points $q^{3}z=\dilp^{-6}z$
  	are inside $\ga_4^{\dilp}$
  	(this is the thick gray arc of the contour $\dilp^{-6}\ga_4^{\dilp}$).
  	}
  	\label{fig:arcs}
\end{figure}

\begin{proposition}\label{prop:XXZ_Pli_expanded}
	The inverse transform
	$\Plie^{\dilp^{-2},\nu,\dilp}$
	can be written in the following form:
	\begin{align}\label{preXXZ_Pli_expanded}
		\begin{array}{>{\displaystyle}l>{\displaystyle}l}
		&(\Plie^{\dilp^{-2},\nu,\dilp}G)(\n)=
		\sum_{\la\vdash k}
		\oint_{\Gamma_{\la_1}(\dilp^{-2})}\ldots\oint_{\Gamma_{\la_{\ell(\la)}}(\dilp^{-2})}
		\dilp^{-k}d\Plm_\la^{(\dilp^{-2})}(\w)
		\\&\hspace{140pt}\times
		\prod_{j=1}^{\ell(\la)}\frac{1}{(w_j/\dilp;\dilp^{-2})_{\la_j}(\nu w_j;\dilp^{-2})_{\la_j}}
		\Psild_{\w\circ\la}(\n)
		G(w\circ\la).
		\end{array}
	\end{align}
	Consequently, the spatial Plancherel formula looks as
	\begin{align}\label{preXXZ_Plancherel}
		\sum_{\la\vdash k}
		\oint_{\Gamma_{\la_1}(\dilp^{-2})}\ldots\oint_{\Gamma_{\la_{\ell(\la)}}(\dilp^{-2})}
		\dilp^{-k}d\Plm^{(\dilp^{-2})}_\la(\w)
		\prod_{j=1}^{\ell(\la)}\frac{1}
		{(w_j/\dilp;\dilp^{-2})_{\la_j}(\nu w_j;\dilp^{-2})_{\la_j}}
		\Psild_{\w\circ\la}(\n)
		\Psird_{\w\circ\la}(\vec m)=\mathbf{1}_{\n=\vec m}.
	\end{align}
	Here $\n,\vec m\in\Weyl{k}$, and
	in both formulas in the eigenfunctions we set
	$q=\dilp^{-2}$.
\end{proposition}
\begin{proof}
	The first expression
	follows from the same argument
	as in \cite[Lemma 5.1]{BorodinCorwinPetrovSasamoto2013}.
	This implies the form of the spatial Plancherel
	theorem given in the second formula.
\end{proof}
\begin{remark}
	When $\dilp$ is a root of unity, $\dilp^{-2j}z$ and $z$
	may belong to the same contour for some $j$, so some
	of the contours $\Gamma_{\ell}(\dilp^{-2})$
	would be empty.
	On the other hand, formulas \eqref{preXXZ_Pli_expanded}--\eqref{preXXZ_Plancherel}
	contain expressions of the form $\dilp^{-2j}z_A-z_B$
	in the denominator, which leads to additional singularities.
	While we expect that
	these issues can be resolved by
	a suitable regularization, we do not pursue this direction here.
\end{remark}

Putting $\nu=\dilp$,
one can readily check that the conjugated $q$-Hahn eigenfunctions
are related to the XXZ eigenfunctions \eqref{XXZ} and \eqref{XXZ_R} as
\begin{align*}
	\Psild_{\z}(x_k,\ldots,x_1)\vert_{q=\dilp^{-2},\,\nu=\dilp}
	&=\Psi^{\mathrm{XXZ}}_{\z}(x_1,\ldots,x_k),\\
	\Psird_{\z}(x_k,\ldots,x_1)\vert_{q=\dilp^{-2},\,\nu=\dilp}
	&=\mathbf{1}_{x_1<\ldots<x_k}\cdot(\dilp^{2}-1)^{k}(\Refl\Psi^{\mathrm{XXZ}}_{\z})(x_1,\ldots,x_k)
\end{align*}
for all integers $x_1\le \ldots\le x_k$.

It is not immediately clear how to degenerate the 
spatial Plancherel formula \eqref{preXXZ_Plancherel} when setting $\nu=\dilp$ (cf. Remark~\ref{rmk:update_ASEP_statement}).
In the case $|\svd|<1$, there are Plancherel
formulas for the XXZ model
in the literature
\cite{BabbittGutkin}, \cite{Gutkin} which
involve integration over so-called Chebyshev
circles.
These are certain circular arcs in the spectral variables $\vxi$,
which in the variables $\z$ translate to rays
starting at the origin.
Plancherel measures (understood in a suitable way) 
are positive on such contours. It seems plausible that 
one can 
deform the integration contours in
\eqref{preXXZ_Plancherel}
so that 
after setting 
$\nu=\dilp$
they match the earlier known formulas, 
but we do not explore this direction here.

\medskip

We expect that the spectral Plancherel formula and the spectral biorthogonality
statement for XXZ eigenfunctions
hold (with some modifications) in the case $|\svd|<1$ (and thus complex
$\dilp$), but we do not pursue this direction here.



\section{Appendix. Further degenerations} 
\label{sub:further_degenerations_of_eigenfunctions}

The eigenfunctions discussed in the present paper
admit further scaling limits (in other words, degenerations),
all the way to the Bethe ansatz eigenfunctions
of the continuous delta Bose gas.
These degenerations correspond to
dashed arrows on Fig.~\ref{fig:big_scheme},
and in this appendix we briefly discuss them.
We will focus on the ``Boson side'' (as in \S \ref{sub:_q_hahn_boson_process}) and on the
corresponding left eigenfunctions (the right eigenfunctions scale in a very similar
way).
That is, we will not describe in detail the ``TASEP side''
dual to the ``Boson side'' as in \S \ref{sub:_q_hahn_tasep_and_markov_duality}.

\begin{remark}
	In principle, limit relations between various Bethe ansatz eigenfunctions
	we present in this appendix
	could be used to establish Plancherel theorems for them.
	However, for more degenerate eigenfunctions these
	results are already known, see the corresponding references for each case below.
\end{remark}

\subsection{$q$-Boson eigenfunctions} 
\label{sub:boson_eigenfunctions}

If we set $\nu=0$ in the $q$-Hahn TASEP (\S \ref{sub:_q_hahn_tasep_and_markov_duality}),
we obtain the discrete-time geometric $q$-TASEP of \cite{BorodinCorwin2013discrete}.
If we further pass to continuous time
(with the help of the parameter $\mu$), we arrive at the $q$-TASEP
introduced in \cite{BorodinCorwin2011Macdonald}.
As shown in \cite{BorodinCorwinSasamoto2012},
the dual process to the $q$-TASEP (in the same way as in \S \ref{sub:_q_hahn_tasep_and_markov_duality})
is the stochastic $q$-Boson particle system of \cite{SasamotoWadati1998}.

The generator of the stochastic $q$-Boson particle system
is equivalent to the following free generator
\begin{align*}
	(\mathcal{L}^{\text{$q$-Boson}}u)(\n)=
	(1-q)\sum_{j=1}^{k}(\nabla_{j}u)(\n),
	\qquad (\nabla f)(x):=f(x-1)-f(x),
\end{align*}
subject to two-body boundary conditions
\begin{align*}
	\left.\big(\nabla_{i}-q\nabla_{i+1}
	\big)u\right\vert_{\n\in\Z^{k}\colon n_i=n_{i+1}}=0
\end{align*}
for any $1\le i\le k-1$.
Note that the above boundary conditions are simply the $\nu=0$
degenerations of the $q$-Hahn boundary conditions \eqref{backward_two_body}.

The (left) eigenfunctions of the $q$-Boson particle system
are the $\nu=0$ degenerations of the $q$-Hahn eigenfunctions
$\Psil_{\z}(\n)$ \eqref{Psil}:
\begin{align}\label{q_Boson_eigenfunctions}
	\Psi^{\ell,\text{$q$-Boson}}_{\z}(\n)
	=
	\sum_{\sigma\in S(k)}\prod_{1\le B<A\le k}
	\frac{z_{\sigma(A)}-qz_{\sigma(B)}}
	{z_{\sigma(A)}-z_{\sigma(B)}}\prod_{j=1}^{k}
	\left(1-z_{\sigma(j)}\right)^{-n_j},
	\qquad\n\in\Weyl{k}.
\end{align}
They depend on a parameter $q\in(0,1]$
and on our discrete spatial variables $\n$.
These eigenfunctions can be constructed by applying the coordinate Bethe ansatz
to the $q$-Boson generator, see \cite[\S2.3]{BorodinCorwinPetrovSasamoto2013}.
Note that the system dual to the discrete-time geometric $q$-TASEP
also possesses the same eigenfunctions \eqref{q_Boson_eigenfunctions}.

The Plancherel spectral theory corresponding to the
$q$-Boson eigenfunctions was developed in \cite{BorodinCorwinPetrovSasamoto2013}
using different ideas, see Remarks \ref{rmk:difference_in_proofs_of_Plancherel}
and \ref{rmk:difference_with_BCPS_spectral} for details.
On the other hand, main results of \cite{BorodinCorwinPetrovSasamoto2013} can be
obtained as a rather straightforward $\nu=0$
degeneration of the corresponding results of the present paper.


\subsection{$q$-Boson $\to$ Semi-discrete delta Bose gas} 
\label{ssub:boson_to_semi_discrete_directed_polymer}

In a scaling limit,
the $q$-Boson generator converges to a certain
semi-discrete delta Bose gas
operator
governing evolution of the moments of
the
semi-discrete stochastic heat equation.
The latter equation is satisfied by the partition function
of the O'Connell--Yor semi-discrete directed polymer
introduced in \cite{OConnellYor2001}
(see also \cite{Oconnell2009_Toda}).
Details on passing from the $q$-TASEP level
to the semi-discrete polymer level are
discussed in \cite{BorodinCorwin2011Macdonald},
\cite{BorodinCorwinSasamoto2012},
\cite[\S6.3]{BorodinCorwinPetrovSasamoto2013},
see also
\cite{BorodinPetrov2013Lect} for a brief general account.

The semi-discrete delta Bose gas operator (related to the semi-discrete stochastic heat equation via a duality)
is equivalent to the following
free operator
\begin{align*}
	(\mathcal{L}^{\text{S--D}}u)(\n)=
	\sum_{j=1}^{k}(\nabla_{j}u)(\n)
\end{align*}
(with the same $\nabla_i$ as in \S \ref{sub:boson_eigenfunctions})
subject to two-body boundary conditions
\begin{align*}
	\left.\big(\nabla_{i}-\nabla_{i+1}-c
	\big)u\right\vert_{\n\in\Z^{k}\colon n_i=n_{i+1}}=0
\end{align*}
for any $1\le i\le k-1$. Here $c\in\R\setminus\{0\}$
is the so-called coupling constant which is a parameter of the model.

The coordinate Bethe ansatz yields the following eigenfunctions
at the semi-discrete polymer level (they depend on $c$ and on discrete
spatial variables):
\begin{align}\label{S_D_eigenfunctions}
	\Psi^{\ell,\text{S--D}}_{\z}(\n)=
	\sum_{\sigma\in S(k)}\prod_{B<A}
	\frac{z_{\sigma(A)}-z_{\sigma(B)}-c}
	{z_{\sigma(A)}-z_{\sigma(B)}}\prod_{j=1}^{k}
	z_{\sigma(j)}^{-n_j},\qquad \n\in\Weyl{k}.
\end{align}
The above eigenfunctions appeared earlier in \cite{Takeyama2012SD}.
Plancherel theory
and completeness of the Bethe ansatz
for the above setting
is discussed in \cite[\S6.3]{BorodinCorwinPetrovSasamoto2013}.

One can readily check the following convergence of
eigenfunctions \eqref{q_Boson_eigenfunctions} to \eqref{S_D_eigenfunctions}:
\begin{align*}
	\lim_{\epsilon\to0}\;
	(\epsilon/c)^{n_1+\ldots+n_k}
	\left.\Big(\Psi^{\ell,\text{$q$-Boson}}_{\z}(\n)\right
	\vert_{q=e^{-\epsilon},\; z_j=e^{-w_j\epsilon/c}}\Big)
	=\Psi^{\ell,\text{S--D}}_{\w}(\n).
\end{align*}


\subsection{Semi-discrete delta Bose gas $\to$ Continuous delta Bose gas} 
\label{ssub:_semi_discrete_directed_polymer_to_kpz}

A further degeneration of the semi-discrete delta Bose gas
described in \S \ref{ssub:boson_to_semi_discrete_directed_polymer}
takes us to the level of continuous delta Bose gas
\cite{BorodinCorwin2011Macdonald},
\cite{BorodinCorwinFerrari2012},
\cite{BorodinCorwinSasamoto2012},
\cite{MorenoQuastelRemenik2014OCYKPZ}.
The latter system is also referred to as the
Lieb--Liniger system, and it is dual
\cite{BertiniCancrini1995}, \cite[\S6]{BorodinCorwin2011Macdonald}
to the
stochastic heat equation, or, via the Hopf-Cole transform, to the KPZ
equation.

It is standard in physics literature
(e.g., see \cite{DotsenkoTW_onedim}, \cite{Calabrese_LeDoussal_Rosso})
to reduce the continuous delta Bose gas operator to the following
free operator
\begin{align}\label{KPZ_free_operator}
	(\mathcal{L}^{\text{KPZ}}u)(x_1,\ldots,x_k)=
	\frac{1}{2}\sum_{j=1}^{k}\frac{\partial^{2}}{\partial x_j^{2}}
	u(x_1,\ldots,x_k),\qquad
	x_1\le \ldots\le x_k,\quad x_i\in\R,
\end{align}
subject to two-body boundary conditions
(here $\tilde c\in\R\setminus\{0\}$
is the coupling constant)
\begin{align}\label{KPZ_BC}
	\left.\left(\frac{\partial}{\partial x_i}
	-\frac{\partial}{\partial x_{i+1}}-\tilde c
	\right)u\right\vert_{\x\in\R^{k}\colon x_i+0=x_{i+1}}=0
\end{align}
for any $1\le i\le k-1$ (where $x_i+0=x_{i+1}$
means the limit as $x_i\to x_{i+1}$ from below).
To our knowledge, this reduction has not been rigorously justified.
The system \eqref{KPZ_free_operator}--\eqref{KPZ_BC}
is referred to as the Yang system (of type A) in \cite{HeckmannOpdam1997},
and dates back to \cite{YangSystem1967}, \cite{YangSystem1968}.
See also \cite{Gaudin1971}, \cite{HeckmannOpdam1997}
for other root systems.

Applying the coordinate Bethe ansatz to
\eqref{KPZ_free_operator}--\eqref{KPZ_BC},
one constructs the eigenfunctions
(note that they depend on $\tilde c$ and on continuous spatial variables)
\begin{align}\label{KPZ_eigenfunctions}
	\Psi^{\ell,\text{KPZ}}_{\z}(\x)=
	\sum_{\sigma\in S(k)}\prod_{B<A}
	\frac{z_{\sigma(A)}-z_{\sigma(B)}-\tilde c}
	{z_{\sigma(A)}-z_{\sigma(B)}}\prod_{j=1}^{k}
	e^{x_jz_{\sigma(j)}}
	,\qquad
	x_1\le \ldots\le x_k,\quad x_i\in\R.
\end{align}
These eigenfunctions were first written down
by Lieb and Liniger \cite{LiebLiniger}.
The corresponding spatial Plancherel formula
was proven in various forms
in
\cite{Oxford1979},
\cite{HeckmannOpdam1997},
\cite{ProlhacSpohn2011}.
There are also certain accounts of spectral biorthogonality results
for the eigenfunctions \eqref{KPZ_eigenfunctions}
in the physics literature, e.g., see
\cite{CalabreseCaux2007},
\cite{DotsenkoTW_onedim}, \cite{Dotsenko2010_universal}.
Details on the spatial Plancherel formula in the
language similar to the rest of the present paper can
be found in \cite[\S7.1]{BorodinCorwinPetrovSasamoto2013}.
See also
\cite[Remark 6.2.5]{BorodinCorwin2011Macdonald}
for more historical background (in particular, concerning
completeness of the Bethe ansatz).

One can readily check the following convergence of
eigenfunctions \eqref{S_D_eigenfunctions} to
\eqref{KPZ_eigenfunctions}:
\begin{align*}
	\lim_{\epsilon\to0}\;
	g^{\sum_{j=1}^{k}[-x_j\frac g \epsilon]}
	\left.\Big(\Psi^{\ell,\text{S--D}}_{\z}
	\big(
	[{-x_k}\tfrac{g}{\epsilon}],\ldots,[-x_1 \tfrac{g}{\epsilon}]\big)\right
	\vert_{c=-\epsilon\tilde c,\; z_j=g+w_j \epsilon}\Big)
	=
	\Psi^{\ell,\text{KPZ}}_{\w}(x_1,\ldots,x_k).
\end{align*}
Here $x_i\in\R$ are as above, and $g>0$
is an arbitrary fixed constant.
The notation $[\cdots]$ means integer part.
Note that
the rescaling of the spectral variables is performed around
a point $g\ne0$, i.e., not around a singularity of the multiplicative
terms $z_{\sigma(j)}^{-n_j}$ in \eqref{S_D_eigenfunctions}.
This situation is different from the
rescaling of the spectral variables in
\S \ref{ssub:boson_to_semi_discrete_directed_polymer} and
in \S \ref{ssub:_q_boson_to_van_diejen_s_delta_bose_gas} below.


\subsection{$q$-Boson $\to$ Van Diejen's delta Bose gas} 
\label{ssub:_q_boson_to_van_diejen_s_delta_bose_gas}

Another scaling
\cite[\S 6.2]{BorodinCorwinPetrovSasamoto2013}
of the $q$-Boson generator (\S \ref{sub:boson_eigenfunctions})
takes us to a semi-discrete delta Bose gas
studied by Van Diejen
\cite{vanDiejen2004HL} (other models
of a similar nature are discussed in \cite{VanDiejen2013_2}, \cite{VanDiejen2013_1}).
The limiting operator is equivalent to the following
free operator subject to two-body boundary conditions
(we use notation of \S \ref{sub:_q_hahn_boson_process})
\begin{align*}
	(\mathcal{L}^{\text{VD}}u)(\n)=
	\sum_{j=1}^{k}u(\n_{i}^{-}),\qquad
	\left.\big(u(\n_i^{-})-qu(\n_{i+1}^{-})
	\big)\right\vert_{\n\in\Z^{k}\colon n_i=n_{i+1}}=0
\end{align*}
for any $1\le i\le k-1$.

Applying the coordinate Bethe ansatz, one arrives at the following eigenfunctions:
\begin{align}\label{VD_eigenfunctions}
	\Psi^{\ell,\text{VD}}_{\z}(\n)=
	\sum_{\sigma\in S(k)}\prod_{B<A}
	\frac{z_{\sigma(A)}-qz_{\sigma(B)}}
	{z_{\sigma(A)}-z_{\sigma(B)}}\prod_{j=1}^{k}
	z_{\sigma(j)}^{-n_j},\qquad
	\n\in\Weyl{k},
\end{align}
depending on $q\in(0,1]$
and on discrete spatial variables.
One can readily identify these eigenfunctions with the
Hall-Littlewood symmetric polynomials \cite[Ch. III]{Macdonald1995}.
Spatial Plancherel formula corresponding to the eigenfunctions \eqref{VD_eigenfunctions}
(and thus the completeness
of the Bethe ansatz)
was obtained in \cite{vanDiejen2004HL},
it is essentially equivalent to a similar statement
for Macdonald's spherical functions \cite{Macdonald1971padic},
\cite{Macdonald2000orthpoly}.
See also
\cite[\S 6.2]{BorodinCorwinPetrovSasamoto2013}
for the spectral biorthogonality statement which is implied by
the Cauchy identity for the Hall-Littlewood polynomials
\cite[Ch. III]{Macdonald1995}.

One can readily check the following
convergence of the $q$-Boson eigenfunctions
\eqref{q_Boson_eigenfunctions} to those of Van Diejen's model
\eqref{VD_eigenfunctions}:
\begin{align*}
	\lim_{\epsilon\to0}\;
	(-\epsilon)^{-n_1-\ldots-n_k}
	\left.\Big(\Psi^{\ell,\text{$q$-Boson}}_{\z}(\n)\right
	\vert_{z_j=w_j/\epsilon}\Big)
	=\Psi^{\ell,\text{VD}}_{\w}(\n).
\end{align*}


\subsection{Van Diejen's delta Bose gas $\to$ Continuous delta Bose gas} 
\label{ssub:van_diejen_s_delta_bose_gas_to_kpz}

A rigorous treatment
of convergence
of Van Diejen's semi-discrete delta Bose gas
(\S \ref{ssub:_q_boson_to_van_diejen_s_delta_bose_gas})
to the
continuous delta Bose gas (\S \ref{ssub:_semi_discrete_directed_polymer_to_kpz})
is performed in
\cite{vanDiejen2004HL}.

Let us record the convergence of the corresponding eigenfunctions,
\eqref{VD_eigenfunctions} to \eqref{KPZ_eigenfunctions}:
\begin{align*}
	\lim_{\epsilon\to0}\;
	\left.\Big(\Psi^{\ell,\text{VD}}_{\z}
	\big(
	[{x_1}\tfrac{\tilde c}{\epsilon}],\ldots,[x_k \tfrac{\tilde c}{\epsilon}]
	\big)
	\right
	\vert_{q=e^{-\epsilon},\; z_j=e^{-w_j\epsilon/\tilde c}}\Big)
	=\Psi^{\ell,\text{KPZ}}_{\w}(x_1,\ldots,x_k),
\end{align*}
where $x_i\in\R$ and $x_1\le \ldots\le x_k$.
We assumed that $\tilde c>0$,
but a similar statement can be readily written down for $\tilde c<0$.
Note that here (as in \S \ref{ssub:_semi_discrete_directed_polymer_to_kpz})
the spectral variables are rescaled around $1$, which
is not a singularity of the factors
$z_{\sigma(j)}^{-n_j}$ in \eqref{VD_eigenfunctions}.
One could also insert an arbitrary fixed constant $g$
around which the spectral variables are rescaled,
but we will not write this down.


\subsection{ASEP / XXZ $\to$ Continuous delta Bose gas} 
\label{ssub:asep_and_xxz_to_kpz_}

Convergence of the ASEP
to the KPZ equation (equivalently, to the logarithm of the stochastic heat equation)
was established in various senses in
\cite{bertiniGiacomin1997stochastic},
\cite{AmirCorwinQuastel2011}.
See also
\cite{ImamuraSasamotoSpohn2011} for a brief account
related to the corresponding ``Boson side''.

The ASEP eigenfunctions \eqref{PsiASEP_bwd} converge to
the continuous delta Bose ones \eqref{KPZ_eigenfunctions}
in the following way. Let $g\in\R$
be an arbitrary fixed constant such that $\frac{\tilde c}{g(1-g)}>0$,
$\tau=e^{-\sqrt \epsilon}$, $z_j=(g^{-1}-1)e^{w_j\sqrt \epsilon/\tilde c}$,
and
\begin{align*}
	x_j=\left[\frac{y_j}\epsilon\cdot\frac{\tilde c}{g(1-g)}\right],\qquad
	j=1,\ldots,k.
\end{align*}
Here $\x=(x_1<\ldots<x_k)$, $x_i\in\Z$ are the spatial variables for the ASEP,
and $\y=(y_1\le \ldots\le y_k)$, $y_i\in\R$
are the spatial variables for the continuous delta Bose gas.
Then we have
\begin{align*}
	\tau^{x_j(1-g)}\left(\frac{1+z_{\sigma(j)}}{1+z_{\sigma(j)}/\tau}\right)^{-x_j}
	\sim e^{\frac{1}{2} y_j \tilde c}e^{y_jw_{\sigma(j)}},
	\qquad
	j=1,\ldots,k.
\end{align*}
Therefore, under the scaling just described, we have
\begin{align}\label{ASEP_to_KPZ}
	\lim_{\epsilon\to0}\;
	\tau^{(1-g)(x_1+\ldots+x_k)}\Psiasep_{\z}(\x)=\
	e^{\frac 12\tilde c(y_1+\ldots+y_k)}\Psi^{\ell,\text{KPZ}}_{\w}(\y).
\end{align}
Note that as in \S \ref{ssub:_semi_discrete_directed_polymer_to_kpz} and \S \ref{ssub:van_diejen_s_delta_bose_gas_to_kpz} above, the spectral variables are rescaled away from the singularity.

The XXZ eigenfunctions are
related to the ASEP ones in a straightforward way (cf. Remark
\ref{rmk:ASEP_XXZ_relation}), and their convergence
to the eigenfunctions at the continuous delta Bose level
is very similar to \eqref{ASEP_to_KPZ}. We will not write down this statement.



\end{document}